%% file: main.tex
\documentclass[11pt]{article}
\usepackage[margin=1in]{geometry}

\usepackage{times}
\usepackage{color}
\usepackage{soul}
\usepackage{url}
\usepackage[colorlinks,allcolors=blue]{hyperref}
\usepackage[utf8]{inputenc}
\usepackage[small]{caption}
\usepackage{float}
\usepackage{graphicx}
\usepackage{mathtools}
\usepackage{amsmath}
\usepackage{amsthm}
\usepackage{booktabs}

\usepackage{natbib}
\usepackage{amssymb}
\usepackage{amsbsy}
\usepackage{amstext}

\usepackage{algpseudocode}

\algrenewcommand\algorithmicindent{1em}
\algtext*{EndWhile}
\algtext*{EndIf}
\algtext*{EndFor}
\algtext*{EndFunction}

\makeatletter

\floatstyle{ruled}
\newfloat{algorithm}{tbp}{loa}
\providecommand{\algorithmname}{Algorithm}
\floatname{algorithm}{\protect\algorithmname}

\makeatother

\providecommand{\definitionname}{Definition}
\providecommand{\lemmaname}{Lemma}
\providecommand{\theoremname}{Theorem}

\theoremstyle{plain}
\newtheorem{thm}{\protect\theoremname}[section]
\theoremstyle{definition}
\newtheorem{defn}[thm]{\protect\definitionname}
\theoremstyle{plain}
\newtheorem{lem}[thm]{\protect\lemmaname}

\global\long\def\Reals{\mathbb{R}}%
\global\long\def\Var{\mathrm{Var}}%
\global\long\def\E{\mathbb{E}}%
\global\long\def\opt{\mathrm{OPT}}%
\global\long\def\optguess{\widetilde{\opt}}%
\global\long\def\Reach{\mathrm{Reach}}%
\global\long\def\thresh{\boldsymbol{\theta}}%
\global\long\def\Active{\boldsymbol{A}}%
\global\long\def\epsest{\epsilon_{\mathrm{est}}}%
\global\long\def\cest{c_{\mathrm{est}}}%
\global\long\def\prest{\delta_{\mathrm{est}}}%
\global\long\def\mtot{m_{\mathrm{tot}}}%

\title{Efficient Algorithms for Influence Maximization in \\ General
Models and Observed Cascades}

\author{
Fabian Spaeh\thanks{Department of Computer Science, Boston University. \href{mailto:fspaeh@bu.edu}{\texttt{fspaeh@bu.edu}}} \and
Themistoklis Haris\thanks{Department of Computer Science, Boston University. \href{mailto:tharis@bu.edu}{\texttt{tharis@bu.edu}}} \and
Alina Ene\thanks{Department of Computer Science, Boston University. \href{mailto:aene@bu.edu}{\texttt{aene@bu.edu}}} \and
Huy L.\ Nguyen\thanks{Khoury College of Computer Sciences, Northeastern University. \href{mailto:hu.nguyen@northeastern.edu}{\texttt{hu.nguyen@northeastern.edu}}}
}

\begin{document}

\maketitle

\begin{abstract}
We study influence maximization in general stochastic models, the
observed cascades model, and the independent cascade (IC) model. For
general stochastic models with only black-box sample access, we introduce
a low-adaptivity optimization framework that improves sample complexity
and running time over Sadeh et al. (2020) and is instrumental to all
our results. We further introduce an adaptive algorithm guided by
empirical variance, avoiding pessimistic worst-case bounds. Combining our optimization framework with sketching, we obtain
the first algorithm with provable guarantees and nearly-linear running
time for influence maximization on observed cascades, optimal up to
logarithmic factors. For IC, we prove a novel tail bound replacing a
factor $n$ with $\tau$ (the number of diffusion steps) in sample complexity,
improving over prior work when $\tau$ is small, as is common due to
small-world phenomena. Experiments confirm substantial speedups while maintaining
solution quality.
\end{abstract}

\input{intro.tex}

\input{prelims.tex}

\input{algorithms.tex}

\input{experiments.tex}

\section{Conclusion}

We extend the low-adaptivity greedy framework of \citet{chen21} to work
with stochastic estimates of marginal gains, which is non-trivial since
individual samples lack submodularity. We develop new estimation techniques:
median-of-means with known variance bounds, an adaptive empirical variance
estimator that avoids pessimistic worst-case bounds, and sketch-based
estimators for observed cascades. For IC models, we prove a concentration
bound replacing $n$ with $\tau$ in sample complexity. These yield improved
complexity, faster running time, and scalability, including the first
nearly-linear algorithm for observed cascades with provable guarantees.
Experiments confirm substantial speedups while maintaining near-optimal
solution quality.

\section*{Acknowledgements}
FS was supported in part by NSF CAREER grant
CCF-1750333. TH was supported in part by an Alfred P.
Sloan Research Fellowship. AE was supported in part by NSF CAREER grant
CCF-1750333 and an Alfred P.
Sloan Research Fellowship. HN was supported in part by
NSF grant CCF-2311649. We thank Ta Duy Nguyen for helpful discussions in the preliminary stages of this work.

\bibliographystyle{plainnat}
\bibliography{references}

\clearpage

\appendix
\input{appendix.tex}

\end{document}

%% file: intro.tex
\section{Introduction}

Understanding diffusion processes in networks is a fundamental problem
in computer science, biology, sociology, economics, and other fields.
Diffusion processes model phenomena such as the spread of news on
social media and the outbreak of a disease \citep{chen13information,yao22,leskovec07patterns}.
A typical diffusion process starts with an initial set of active nodes,
called the seed set, and it proceeds in discrete time steps. In each
step, all previously activated nodes remain active and new nodes become
active according to an underlying model of diffusion \citep{kempe15}.
The influence of a seed set is the number of nodes that are active
at the end of the process. A fundamental optimization problem that
arises in many applications is to find a seed set of size at most $k$
(the budget) with maximum influence. Notable applications of influence maximization
are viral marketing \citep{richardson02} and outbreak detection in
networks \citep{leskovec07}. In viral marketing, the goal is to find
an initial seed set of ``influencers'' to advertise a product and
spread its adoption through the network. In outbreak detection, we
would like to monitor a network such as a water distribution network
using sensors, and the goal is to find the locations of the sensors
that will allow us to detect contaminations as quickly as possible.

Starting with the influential work of \citet{kempe15}, there has
been significant interest in designing algorithms for influence maximization
that achieve provable approximation guarantees in general settings.
As shown by \citet{kempe15} and subsequent work \citep{mossel07},
a very general class of models where provable approximation
guarantees are achievable in polynomial time are submodular stochastic
diffusion models, where the expected influence is submodular. Submodular models
have significant modeling power and wide-ranging applications, which
motivates the design of efficient algorithms that scale to large
networks. However, despite important developments for the independent
cascade (IC) model \citep{borgs14,guo22,tang14} and more general
submodular models \citep{sadeh20}, it remains challenging to obtain
scalable algorithms for models beyond IC. A key bottleneck is the
number of adaptive rounds: the classical Greedy algorithm requires
$k$ rounds, each evaluating $n$ candidates (where $n$ is the number
of nodes and $k$ the budget), with each evaluation potentially requiring
expensive sampling or graph traversals. Recent work on
low-adaptivity submodular optimization \citep{balkanski18,chen21}
shows that far fewer rounds suffice with exact evaluations, but extending
this to stochastic settings is non-trivial. In this work, we develop
such an extension, combining low-adaptivity optimization with tailored
estimation techniques to obtain faster algorithms with provable guarantees.

Beyond stochastic models, practitioners often work directly with observed
data when model parameters are unknown or hard to estimate \citep{narasimhan15}.
We therefore also consider the observed cascades setting, where diffusion
spreads deterministically according to observed data. Here we develop
the first algorithm with provable guarantees and nearly-linear running
time, which is optimal up to logarithmic factors.

We handle both unrestricted influence and settings with explicit limits
on the number of diffusion steps $\tau$. Restricting $\tau$ is necessary in time-critical
applications \citep{chen12,cohen14timed,du13,liu12,gomezrodriguez11},
e.g., detecting outbreaks early \citep{leskovec07}. Moreover, due
to ``small-world'' phenomena, few-step diffusion often approximates
unrestricted influence well \citep{milgram06}.

\subsection{Our contributions}

We provide new algorithms with theoretical guarantees for influence
maximization in general stochastic models, the observed cascades model,
and the independent cascade (IC) model. Our main contributions are:
\begin{itemize}
\item \textbf{Low-adaptivity optimization framework:} We adapt the
low-adaptivity Greedy algorithm of \citet{chen21} to work with stochastic
estimates. This is non-trivial: stochastic estimates are not submodular
\citep{kempe15,BalkanskiRS22}, so the original analysis fails. We
develop new techniques to handle the estimation error. This framework
is instrumental to all our results: it directly improves sample complexity
and running time over \citet{sadeh20} by a factor of (almost) $k$, and
makes variance-based estimation and sketching viable.
\item \textbf{Empirical variance estimation:} Building on our framework,
we introduce an adaptive sampling algorithm guided by empirical variance,
avoiding pessimistic worst-case variance bounds required by prior work \citep{sadeh20}.
\item \textbf{Nearly-linear time for observed cascades:} Combining our
framework with sketching \citep{cohen07}, we provide the first algorithm with provable
guarantees and nearly-linear $\widetilde{O}(m_{\mathrm{tot}})$ running time
(where $m_{\mathrm{tot}}$ is the input size) for observed cascades, improving
over the previously best achievable $\Omega(n\cdot m_{\mathrm{tot}})$.
Prior work \citep{cohen14} also uses sketching but with a different algorithmic
approach (see Section~\ref{sec:observed-cascades-alg}).
\item \textbf{New concentration bound for IC:} We prove a novel tail
bound for IC samples, improving sample complexity by replacing a factor $n$
with~$\tau$, yielding improved guarantees over \citet{borgs14,tang14,guo22}
when $\tau\ll n$.
\end{itemize}
Together, these contributions yield algorithms with provable guarantees
that substantially improve running time, memory, and sample complexity
for both restricted and unrestricted influence.

\paragraph{Practical baselines} Highly optimized practical algorithms
exist primarily for the IC model \citep{borgs14,tang14,tang15,nguyen16,guo22}. For
general stochastic models and observed cascades, the only prior algorithms
with theoretical guarantees are variants of Greedy \citep{sadeh20,cohen14},
which we significantly outperform. These settings are important in
practice when IC assumptions do not hold or model parameters are unknown.

\subsection{Related Work}

Influence maximization is well-studied with extensive prior work.
We focus on efficient algorithms with theoretical guarantees; for heuristic
methods, we refer the reader to \citet{li23survey}.

\paragraph{General Stochastic Models} \citet{kempe15} show that influence
maximization generalizes the maximum coverage problem, and thus it
is \textbf{NP}-hard to approximate to any factor better than $1-1/e$.
They also show that the influence function is submodular for the IC
and linear threshold (LT) models, and show how to obtain a nearly-optimal
$1-\frac{1}{e}-\epsilon$ approximation with high probability by implementing
the Greedy algorithm \citep{nemhauser78} using stochastic estimates
for the marginal gains obtained from simulations of the model. Even
though the running time of their approach is polynomial, it is prohibitive
in practice. \citet{kempe15} introduce a more general
stochastic diffusion model, called the General Threshold model, a
vast generalization of the IC, LT, and other diffusion models. Subsequently,
\citet{mossel07} show that the influence function is submodular in
General Threshold models with submodular activation functions. \citet{sadeh20}
study sample complexity of influence maximization in general models,
including General Threshold models, parameterized by the number of
diffusion steps, providing improved sample complexity and running~time. 

\paragraph{Independent Cascades} \citet{borgs14} introduce reverse
influence sampling which is specialized to the IC model, but offers
significantly improved running time and thus serves as the foundation
of most later works. \citet{tang14} introduce TIM+ which reduces
the number of reverse reachability samples and thus improves the overall
running time. \citet{tang15} improve the practical efficiency of
the algorithm further, and later algorithms such as SSA \citep{nguyen16,huang17,nguyen18}
or OPIM \citep{tang18} track the current solution quality to stop
early. \citet{guo22} are the first to offer wider improvements to
the asymptotic running time over the approach of \citet{borgs14}
via geometric sampling. \citet{zhang24ici} extend IC to model
invitation-aware diffusion with multi-stage behavior conversions.

\paragraph{Observed Cascades} \citet{cohen14} develop an algorithm
that uses reachability sketches to evaluate the influence efficiently.
In combination with the Greedy algorithm, they lazily populate sketches,
resulting in a more efficient optimization. Many works are dedicated
to inferring the underlying network structure from real observations
of diffusion processes \citep{narasimhan15,netrapalli12,Abadie15}.
\citet{chen21network,narayanaswamy19} provide two-stage algorithms
that first infer the model and then maximize~influence.

%% file: prelims.tex
\section{Preliminaries}

\label{sec:prelims}

\subsection{Influence maximization in stochastic diffusion models}

\label{sec:prelim-diffusion}

We study influence maximization under very general models of stochastic
diffusion in graphs. Let $G=(V,E)$ be a graph on $n=|V|$
nodes. Let $S\subseteq V$ be an initial set of nodes, called the
\textbf{seed set}. The diffusion process starts by activating all
of the nodes in $S$, and it progresses in steps. In each step $t$,
all previously activated nodes remain active and new nodes become
active according to a stochastic model (we give specific examples
below). The process stops when either no new nodes become active (this
happens after at most $n$ steps), or we reach a given limit on the
number of activation steps. Letting $\Active^{(t)}(S)$
denote the set of nodes that are active at the end of step $t$ starting
from seed set $S$, we obtain a nested sequence $\Active^{(0)}(S)=S\subseteq\Active^{(1)}(S)\subseteq\dots\subseteq\Active^{(n)}(S)$
of random sets.

The \textbf{$\tau$-step influence} of the seed set $S$ is the expected
size $\E[|\Active^{(\tau)}(S)|]$ of the active set after
$\tau$ steps of diffusion. We let 
\begin{equation}
I^{(\tau)}\colon2^{V}\to\Reals,I^{(\tau)}(S)=\E[|\Active^{(\tau)}(S)|]\quad\forall S\subseteq V\label{eq:influence-fn}
\end{equation}
 denote the $\tau$-step influence function.

\paragraph{Influence maximization} In the influence maximization problem,
we are given as input a number of steps $\tau\leq n$ and a budget
$k\leq n$, and the goal is to find a seed set $S$ with $|S|\leq k$
with maximum $\tau$-step influence, i.e., solve the optimization
problem
\begin{equation}
{\textstyle \max_{S\colon|S|\leq k}I^{(\tau)}(S)}\label{eq:influence-max}
\end{equation}
We let $S^{*}\in\arg\max_{S\colon|S|\leq k}I^{(\tau)}(S)$
be an optimal solution for the problem and we let $\opt:=I^{(\tau)}(S^{*})$
be its influence. A set $S$ with $|S|\leq k$ is a $c$-approximate
solution if $I^{(\tau)}(S)\geq c\cdot\opt$.

\paragraph{Influence estimation via sampling} Exact evaluation of the
influence of a seed set is intractable even in the simplest and most structured diffusion models: \citet{chen10scalable} showed that
evaluating the influence is \textbf{\#P}-hard in the independent cascades
model (we define this model below). However, in many important settings,
it is possible to compute a stochastic estimate of the influence function
in time that is polynomial in the size $|E|+|V|$
of the graph $G=(V,E)$, e.g., by direct simulation of
the diffusion process. In this work, we design algorithms for the
very general setting where we only have black-box access to such stochastic
estimates. A stochastic function $\widehat{f}\colon2^{V}\to\Reals$
is an unbiased estimate of the influence function $I^{(\tau)}$ if
$\E[\widehat{f}(S)]=I^{(\tau)}(S)$ for all
$S\subseteq V$. We refer to an unbiased estimate $\widehat{f}$ as
a \textbf{sample}. We assume black-box access to independent samples
that we can evaluate via a value oracle: a value oracle for a set
function $\widehat{f}$ is a black-box subroutine that, given a set
$S$ as input, returns the value $\widehat{f}(S)$. We
measure the running time as the number of samples taken, the number
of calls to the evaluation oracles of those samples, and the additional
arithmetic operations. 

\paragraph{Submodular diffusion models} Even if the influence function
can be evaluated exactly in polynomial time, it is intractable to
obtain any non-trivial approximation guarantees for the influence
maximization problem (\ref{eq:influence-max}) without further assumptions
on the influence function. In this work, we assume that the influence
function $I^{(\tau)}$ is a submodular and monotone set function.
A set function $f\colon2^{V}\to\Reals$ is submodular if $f(A)+f(B)\geq f(A\cap B)+f(A\cup B)$
for all subsets $A,B\subseteq V$. An equivalent definition is that
$f$ satisfies the following diminishing returns property: letting
$f(T\vert S):=f(T\cup S)-f(S)$ denote
the marginal gain of $T$ on top of $S$, the function $f$ is submodular
if and only if $f(\{ v\} \vert A)\geq f(\{ v\} \vert B)$
for all $A\subseteq B$ and all $v\in V\setminus B$. The function
is monotone if $f(A)\leq f(B)$ for all $A\subseteq B$.
This broad setting captures many specific models
in the literature (examples below). Note
that, although the influence function is submodular, individual samples
$\widehat{f}$ are generally not submodular~\citep{kempe15}.

The influence maximization problem captures the maximum coverage problem
as a special case, and thus it is \textbf{NP}-hard to obtain an approximation
factor better than $1-1/e$. The algorithms we propose in this work
achieve a nearly-optimal $1-1/e-\epsilon$ approximation with probability
$1-\delta$, for any $\epsilon,\delta>0$ given as input. 

\paragraph{Examples of diffusion models} A very general diffusion model
is given by the general threshold model \citep{kempe15,mossel07}.
In this model, each node $v$ is associated with a monotone activation
function, which is a set function on $v$'s neighbors, and a random
threshold. The thresholds are sampled independently and uniformly
at random from the interval $\left[0,1\right]$. The process starts
with the seed set being active, i.e., $\boldsymbol{A}^{(0)}=S$. In
each step $t$, all of the previously active nodes remain active (i.e.,
$\boldsymbol{A}^{(t-1)}\subseteq\boldsymbol{A}^{(t)}$) and each inactive
node becomes active if the value of its activation function on the
set of its active neighbors exceeds the threshold. \citet{mossel07}
showed that, if the activation functions are submodular, then the
unrestricted influence function $I^{(n)}$ is submodular. We can efficiently construct a sample $\widehat{f}$ of $I^{(\tau)}$
for a general threshold model as follows. We
sample the thresholds $\thresh$ and we implement the evaluation oracle
for $\widehat{f}$ via direct simulation of the activation process:
on input $S$, we construct the activated sets $\Active^{(t)}(S)$
based on the sampled thresholds. Individual samples are
generally not submodular~\citep{kempe15}.

The linear threshold (LT) model \citep{granovetter1978threshold}
is the special case of the general threshold model where the activation
functions are linear functions. In the independent cascades (IC) model,
each edge $e\in E$ in the graph is associated with a probability
$p_{e}\in\left[0,1\right]$. Each edge $e$ becomes live independently
at random with probability $p_{e}$. The set $\Active^{(t)}(S)$
of nodes that are active after $t$ steps is the set of all nodes
that are reachable from $S$ via live edge paths of length at most
$t$. \citet{kempe15} showed that the independent cascades model
is also a special case of the general threshold model. The $\tau$-step
influence function $I^{(\tau)}$ is submodular for all $\tau$ in
LT and IC models \citep{kempe15}. In independent cascade models,
a sample $\widehat{f}$ of the influence $I^{(\tau)}$ is obtained
by sampling the live edges and computing reachabilities in the live
edge subgraph. A sample is thus a coverage function and hence submodular
and monotone (we have $\widehat{f}(S)=|\bigcup_{v\in S}C_{v}|$
where $C_{v}$ is the set of nodes that are reachable from $v$ on
live edge paths of length at most $\tau$). In linear threshold models,
a sample constructed by direct simulation of the process as discussed
above is not submodular in general~\citep{kempe15}.

\paragraph{Further extensions} Our algorithms readily extend to the
setting where $I^{(\tau)}(S)=\E[g(\Active^{(\tau)}(S))]$
for some function $g\colon2^{V}\to\Reals_{\geq0}$, e.g., $g$ is
linear or more generally a monotone submodular function. As before,
we assume that the influence $I^{(\tau)}$ is non-negative, submodular
and monotone, and we are given black-box access to independent samples
that we can evaluate via a value oracle.

\subsection{Influence maximization on observed cascades}

\label{sec:prelim-cascades}

We also consider deterministic instances of the influence maximization
problem on observed cascades. Here we are given as input $\ell$ graphs
$G_{1}=(V_{1},E_{1}),\dots,G_{\ell}=(V_{\ell},E_{\ell})$.
We let $V=V_{1}\cup\dots\cup V_{\ell}$. We let $\Reach_{G_{i}}^{(\tau)}(S)$
denote the set of all nodes $v\in V_{i}$ such that there is a path
in $G_{i}$ of length at most $\tau$ from a node in $S\cap V_{i}$
to $v$. The $\tau$-step influence of a seed set $S\subseteq V$
is
\[
I^{(\tau)}(S)={\textstyle \frac{1}{\ell}{\textstyle \sum_{i=1}^{\ell}}|\Reach_{G_{i}}^{(\tau)}(S)|} .
\]

\medskip

\noindent
For convenience, we provide Table \ref{tab:notation} with common notation in 
in the appendix.

%% file: algorithms.tex
\section{Our algorithms}

\label{sec:algorithms}

In this section, we give an overview of our algorithms and results. We give more detailed descriptions of the algorithms with pseudocode and the analysis in the appendix. 
\subsection{Algorithmic template}

\begin{algorithm}
\caption{Low-adaptivity Greedy algorithm for influence maximization using
stochastic estimates for the marginal gains.}
\label{alg:template-idealized}
\begin{algorithmic}[1]
\Statex $\mathrm{LowAdaptiveGreedy}(\epsilon, \delta)$:
\State $S \gets \emptyset$; $\alpha \gets n$
\For{$r = 1, \dots, T_{\mathrm{outer}} = \frac{2}{\epsilon}\log n$}
    \State $U \gets V$
    \State $\alpha \gets \alpha(1 - \epsilon)$ \Comment{decrease threshold}
    \For{$t = 1, \dots, T_{\mathrm{inner}} = \frac{8}{\epsilon}\log(\frac{n}{\delta})$}
        \State $\{X(u \mid S) : u \in U\} \gets \mathrm{EstimateGains}(\{\{u\} : u \in U\}, S)$
        \State $U \gets \{v \in U : X(v \mid S) \geq \alpha\}$ \Comment{filtering}
        \If{$U = \emptyset$} \textbf{break}
        \EndIf
        \State Let $v_1, v_2, \dots, v_{|U|}$ be a random permutation of $U$
        \State $s \gets \min\{k - |S|, |U|\}$
        \State Let $T_i = \{v_1, v_2, \dots, v_i\}$ for all $1 \leq i \leq s$
        \State $\{X(T_i \mid S) : 1 \leq i \leq s\} \gets \mathrm{EstimateGains}(\{T_i : 1 \leq i \leq s\}, S)$
        \State $i^* \gets \arg\max\{1 \leq i \leq s : X(T_i \mid S) \geq \alpha i\}$
        \State $S \gets S \cup T_{i^*}$ \Comment{add to solution}
        \If{$|S| = k$} \Return $S$
        \EndIf
    \EndFor
    \If{$U \neq \emptyset$} \Return Failure
    \EndIf
\EndFor
\State \Return $S$
\end{algorithmic}
\end{algorithm}

All our optimization algorithms follow the template in Algorithm~\ref{alg:template-idealized}.
The algorithm is based on \citet{chen21} for submodular maximization
with exact evaluations. The key idea is to add multiple high-gain
elements per iteration (instead of one as in Greedy), reducing the
number of adaptive rounds from $O(n)$ to $O(\log^{2}n/\epsilon^{2})$.

Specifically, the algorithm multiplicatively decreases
a threshold $\alpha$ that controls the marginal density of sets $T\subseteq V$
which we add to the current solution $S$. The marginal density of
a set $T$ over the solution $S$ is the ratio $f(T\mid S)/|T\setminus S|$.
Instead of adding single elements as in the Greedy algorithm, we add
large sets such that we need fewer iterations. To accelerate the search
for sets with high marginal density, we use a filtering step where
we remove all elements whose marginal gain is less than the threshold
$\alpha$. We then test the marginal density on a random sequence
of prefixes $T_{1}\subseteq T_{2}\subseteq\cdots\subseteq T_{s}$
of size at most $k-|S|$ which we obtain by permuting the
elements. We add the largest set $T_{i}$ whose marginal density
exceeds the threshold $\alpha$ to the solution. Once the size of
our solution reaches the budget $k$, we are done and output $S$.
However, in order to guarantee our running time, we declare failure
if we go through all $T_{\mathrm{inner}}$ iterations of the inner
for loop without having filtered out all elements such that $U=\emptyset$. 

\paragraph{Adaptation to stochastic estimates} The algorithmic template above is due to \citet{chen21}; our contribution is extending it to work with stochastic estimates of the marginal gains, which is non-trivial. The original analysis relies on marginal
gains being deterministic and decreasing (submodularity). With stochastic
estimates, individual samples are not submodular \citep{kempe15},
and approximating a submodular function stochastically does not
generally suffice for good guarantees \citep{BalkanskiRS22}. We address this
by carefully accounting for estimation error throughout the analysis,
using fresh samples each iteration to maintain independence, and proving
that accumulated errors remain bounded. See Appendix~\ref{sec:low-adaptivity-greedy}
for the full analysis.

\paragraph{Estimation subroutine} The template uses a subroutine $\mathrm{EstimateGains}$
to estimate marginal influence gains. Intuitively, good estimates should
be close to the true values, with high probability. We formalize this:
\begin{defn}
\label{def:estimation-error} An estimation algorithm $\mathrm{EstimateGains}$
takes a collection $\mathcal{T}=\{ T_{1},T_{2},\dots\} $
of subsets and a set $S$, returning estimates $X(T_{i}\vert S)$
for the marginal gains $I^{(\tau)}(T_{i}\vert S)$. It
has \emph{multiplicative error} $\epsest$, \emph{additive error}
$\cest$, and \emph{failure probability} $\prest$ if:
\[
\Pr\big[\exists T_{i}\in\mathcal{T}\colon\big|X(T_{i}| S)-I^{(\tau)}(T_{i}| S)\big|>\epsest I^{(\tau)}(T_{i}| S)+\cest\big]\leq\prest
\]
\end{defn}
Note that $\prest$ bounds the joint failure probability (any bad estimate),
not each estimate individually, enabling a clean analysis.
We design subroutines achieving any target $\epsest,\cest,\prest$.

\paragraph{Algorithm guarantee} The following theorem states the main
guarantee we establish for the algorithm. To simplify notation, we
use the $\widetilde{O}(\cdot),\widetilde{\Omega}(\cdot),\widetilde{\Theta}(\cdot)$
notation to suppress factors that are poly-logarithmic in $n=|V|$,
i.e., $\widetilde{O}(T)=O(T\cdot(\log n)^{O(1)})$.
Since we analyze our algorithm in the setting where we only have access
to stochastic estimates to the marginal gains, our proof requires
substantial changes from the work of \citet{chen21}. The analysis
is complex and subtle due in part to the use of randomization, and
this is a challenge when trying to make it robust to errors due to
the stochastic estimations. We obtain:
\begin{thm}
\label{thm:template-ideal} Suppose we run Algorithm \ref{alg:template-idealized}
with an estimation routine $\mathrm{EstimateGains}$ with multiplicative
error $\epsest$, additive error $\cest$, and failure probability
$\prest$. Algorithm \ref{alg:template-idealized} makes $N=\widetilde{O}(\frac{1}{\epsilon^{2}}\log(\frac{1}{\delta}))$
calls to $\mathrm{EstimateGains}$ and it returns a solution $S$
with $|S|\le k$ and $I^{(\tau)}(S)\geq(1-\frac{1}{e}-\epsest-\epsilon)\opt-3k\cest$
with probability $1-N\prest-\delta$.
\end{thm}

\subsection{Our algorithm for influence maximization in general stochastic diffusion
models}

\label{sec:general-model-alg}

We now describe our algorithm for general diffusion models. Recall
from Section \ref{sec:prelim-diffusion} that we assume black-box
access to independent samples $\widehat{f}$ for the influence $I^{(\tau)}$
represented as a value oracle that, on input $S$, returns $\widehat{f}(S)$.
Our algorithm instantiates the template Algorithm \ref{alg:template-idealized}
using an estimation subroutine that uses independent samples from
the model to estimate the marginal influence gains.

We provide two versions of the estimation subroutine, both achieving additive error $\cest$, failure probability $\prest$, and no multiplicative error (Defn.~\ref{def:estimation-error}).

\paragraph{Estimation with known variance bound} 
Our first estimation algorithm (Algorithm~\ref{alg:estimate-gains-variance}) uses
median-of-means estimation: it partitions samples into groups, computes the mean
within each group, and returns the median of these means. This standard approach
requires knowing upfront how many samples to take. To guarantee additive error
$\cest$, the required sample size depends on the variance of single-sample estimates.
Since the algorithm must work for \emph{all} sets $T$ queried during optimization,
we need an upper bound $W$ on the maximum variance $\max_{S,T}\Var[\widehat{f}(T\vert S)]$.
This bound $W$ can be pessimistic since it must account for worst-case sets even
if actual queries have lower variance. The work of \citet{sadeh20} uses the same
approach and also requires~$W$. An immediate bound is $W \leq n \cdot \opt$;
tighter bounds are given in \citet{sadeh20}.
The median-of-means estimator achieves:
\begin{thm}
\label{lem:estimate-gains-variance}\label{thm:gtm-estimate-gains-variance}Let
$W$ be an upper bound on the maximum variance $\max_{S,T\subseteq V\colon|S\cup T|\leq k}\Var[\widehat{f}(T\vert S)]$
of the marginal gain estimates obtained from a single sample $\widehat{f}$. Given $W$, $\cest$,
and $\prest$ as inputs, Algorithm \ref{alg:estimate-gains-variance}
achieves an additive error $\cest$, failure probability $\prest$,
and no multiplicative error ($\epsest=0$) using $L=\widetilde{O}(\frac{1}{\cest^{2}}W\log(\frac{1}{\prest}))$
samples and $|{\cal T}|L$ calls to the evaluation oracles
for those samples.
\end{thm}
By combining the template Algorithm \ref{alg:template-idealized}
with the estimation Algorithm \ref{alg:estimate-gains-variance},
we obtain an algorithm with the following guarantee. We give the pseudocode
of the combined algorithm in Algorithm \ref{alg:general-model-variance}, which
has the following guarantee.
\begin{thm}
\label{thm:threshold-variance}\label{thm:gtm-estimate-gains-empirical-variance}Let
$W$ be an upper bound on the maximum variance $\max_{S,T\subseteq V\colon|S\cup T|\leq k}\Var[\widehat{f}(T\vert S)]$
of the marginal gain estimates. Given $W$, $\epsilon$ and $\delta$
as inputs, Algorithm \ref{alg:general-model-variance} uses $\widetilde{O}(NL)$
samples for $L=\widetilde{O}(\frac{k^{2}}{\epsilon^{2}\opt^{2}}W\log(\frac{1}{\delta\epsilon}\log(\frac{1}{\delta})))$
and it returns a solution $S$ satisfying $I^{(\tau)}(S)\geq(1-\frac{1}{e}-\epsilon)\opt$
with probability $1-\delta$.
\end{thm}
\paragraph{Comparison with prior work} The state of the art for general
models is due to \citet{sadeh20}. Compared to their algorithm, ours
improves sample complexity and running time by a factor of $k/N$,
which is $\Omega(k/\log n)$ for constant $\epsilon,\delta$. This is
significant for large seed sets: in viral marketing, campaigns may target
a substantial fraction of the network as influencers \citep{richardson02};
in outbreak detection, monitoring many locations improves coverage
\citep{leskovec07}; and empirical studies \citep{ling23} use seed sets
up to 20\% of nodes. When $k=\Theta(n)$, our factor-$k$ improvement
translates to a factor-$n$ speedup over prior~work.

\paragraph{Estimation via empirical variance} A drawback shared with
\citet{sadeh20} is that worst-case variance bounds are pessimistic.
We provide Algorithm~\ref{alg:estimate-gains-empirical-variance} that
instead uses empirical variance to adaptively determine the number
of samples. The algorithm iteratively draws samples and computes the empirical variance from observations. Using empirical Bernstein bounds, it constructs confidence intervals around the current estimate and terminates once these intervals are sufficiently tight. This adaptive approach avoids oversampling when actual variance is low, yielding significant practical improvements (Sec.~\ref{sec:experiments}). 
Crucially, the algorithm adapts to the actual variance without requiring any bound $W$ as input. Let $W$ be again any upper bound on
$\max_{S, T}\Var[\widehat{f}(T\vert S)]$. We show:
\begin{thm}
\label{lem:emp-var-conf}Given as input $\cest$ and $\prest$, with
probability $1-\prest$, Algorithm \ref{alg:estimate-gains-empirical-variance}
uses $L=\widetilde{O}((\frac{W}{\cest^{2}}+\frac{n}{\cest})\log(\frac{1}{\prest}))$
samples and $|\mathcal{T}|L$ calls to the evaluation oracles
for the samples, and it achieves additive error $\cest$, failure
probability $\prest$, and no multiplicative error ($\epsest=0$).
\end{thm}
By combining the template Algorithm \ref{alg:template-idealized}
with the estimation Algorithm \ref{alg:estimate-gains-empirical-variance},
we obtain an algorithm with the following guarantee. We give the pseudocode
of the combined algorithm in Algorithm \ref{alg:general-model-empirical-variance},
which achieves: 
\begin{thm}
\label{thm:threshold-empirical-variance}Given $\epsilon$ and $\delta$
as input, with probability $1-\delta$, Algorithm \ref{alg:general-model-empirical-variance}
uses $\widetilde{O}(NL)$ samples for $L=\widetilde{O}((\frac{k^{2}W}{\epsilon^{2}\opt^{2}}+\frac{kn}{\epsilon\opt})\log(\frac{N}{\delta}))$
and it returns a solution $S$ satisfying $I^{(\tau)}(S)\geq(1-\frac{1}{e}-\epsilon)\opt$.
\end{thm}

\subsection{Our algorithm for observed cascades}

\label{sec:observed-cascades-alg}

We now present our algorithm for observed cascades, an important setting
where we have $\ell$ observed diffusion graphs $G_{1}=(V_1,E_1),\dots,G_{\ell}=(V_\ell,E_\ell)$
with $\mtot=\sum_{i}|E_{i}|$ total edges. We assume $|V_i|\leq|E_i|$
(isolated vertices can be removed). The influence is the average
reachability: $I^{(\tau)}(S)=\frac{1}{\ell}\sum_{i}|\Reach_{G_{i}}^{(\tau)}(S)|$.
Unlike stochastic models, here we can evaluate $I^{(\tau)}(S)$ exactly
in $O(\mtot)$ time via BFS. Greedy thus gives $1-1/e$ approximation
in $O(nk\cdot\mtot)$ time; this can be reduced to $O(n\cdot\mtot)$
at a small approximation loss \citep{MirzasoleimanBK15}. We achieve
nearly-linear $\widetilde{O}(\mtot)$ time by combining our low-adaptivity framework with reachability sketches.

\paragraph{Estimation via reachability sketches} Our estimation subroutines use bottom-$b$ min-hash sketches \citep{cohen97} to estimate reachability set sizes: Algorithm~\ref{alg:estimate-gains-sketch-unrestricted} handles unrestricted influence ($\tau=n$), while Algorithm~\ref{alg:estimate-gains-sketch-restricted} handles restricted influence ($\tau<n$). We build fresh sketches in each of
the $O(\log^{2}n/\epsilon^{2})$ adaptive rounds. The key insight is
that low-adaptivity makes this affordable: rebuilding sketches $O(\log^{2}n)$
times adds only logarithmic overhead to the total $\widetilde{O}(\mtot)$
running time. Prior work \citep{cohen14} uses sketching with
Greedy, but reuses sketches across $k$ iterations, creating correlations
that invalidate their analysis; fixing it would require rebuilding
$k$ times, which is prohibitive. We thus provide the first
nearly-linear time algorithm with provable~guarantees.
We show that our sketch-based estimators achieve multiplicative error with no additive error:
\begin{thm}
\label{thm:oc-estimation-unrestricted} Consider the setting $\tau=n$ (unrestricted influence).
Given as input $\epsest$ and $\prest$, Algorithm \ref{alg:estimate-gains-sketch-unrestricted}
runs in time $O(b\mtot)$ for $b=\widetilde{O}(\frac{1}{\epsest^{2}}\log(\frac{\ell}{\prest}))$
and it achieves a multiplicative error $\epsest$, failure probability
$\prest$, and no additive error ($\cest=0$).
\end{thm}
\begin{thm}
\label{thm:oc-estimation-restricted} Consider the setting $\tau<n$ (restricted influence).
Given as input $\epsest$ and $\prest$, Algorithm \ref{alg:estimate-gains-sketch-restricted}
runs in expected time $\widetilde{O}(b\mtot\log b)$ for
$b=\widetilde{O}(\frac{1}{\epsest^{2}}\log(\frac{\ell}{\prest}))$
and it achieves a multiplicative error $\epsest$, failure probability
$\prest$, and no additive error ($\cest=0$).
\end{thm}
Our overall algorithm combines the template Algorithm \ref{alg:template-idealized}
with the estimation Algorithms \ref{alg:estimate-gains-sketch-unrestricted} and \ref{alg:estimate-gains-sketch-restricted}. We give the pseudocode in
Algorithm \ref{alg:observed-cascades} and show:
\begin{thm}
\label{thm:observed-cascades-main} For $\tau=n$, given as input
$\epsilon$ and $\delta$, Algorithm \ref{alg:observed-cascades}
runs in time $\widetilde{O}(Nb\mtot)$ for $N=\widetilde{O}(\frac{1}{\epsilon^{2}}\log(\frac{1}{\delta}))$
and $b=\widetilde{O}(\frac{1}{\epsilon^{2}}\log(\frac{\ell N}{\delta}))$
and it returns a solution $S$ satisfying $I^{(\tau)}(S)\geq(1-\frac{1}{e}-\epsilon)\opt$
with probability $1-\delta$. For
$\tau<n$, the running time is $\widetilde{O}(Nb\mtot\log b)$
in expectation. 
\end{thm}

Note that for observed cascades, the dominant memory cost is the bottom-$b$ sketches.
Since we store one sketch of size $b$ per node in each of the $\ell$ graphs,
we require $O(bn\ell)$ total space, which is almost linear in the input size.

\subsection{Our algorithm for influence maximization in independent cascade models}

\label{sec:ic-model-alg}

In this section, we consider the influence maximization in independent
cascade models. As discussed in Sec. \ref{sec:prelims}, we can obtain
a sample $\widehat{f}$ of the influence $\widehat{I^{(\tau)}}$ by
sampling a live-edge graph $G$ and letting $\widehat{f}(S)=|\Reach_{G}^{(\tau)}(S)|$
for all $S\subseteq V$. Our algorithm is simply to take $\ell$ such
samples (for an appropriate number $\ell$ that we discuss below),
and use our algorithm for observed cascades to approximately maximize
the empirical influence $\widehat{I^{(\tau)}}(S):=\frac{1}{\ell}\sum_{i=1}^{\ell}|\Reach_{G_{i}}^{(\tau)}(S)|$
where $G_{1},\dots,G_{\ell}$ are the sampled live-edge graphs.

\paragraph{Our improved sample complexity}  To ensure that the algorithm
constructs a $1-1/e-O(\epsilon)$ approximation with probability $1-O(\delta)$,
it suffices to take enough samples to ensure that the empirical influence
$\widehat{I^{(\tau)}}$ approximates the true influence $I^{(\tau)}$
on all feasible sets in the following sense. Letting $\mathcal{F}=\{ S\subseteq V\colon|S|\leq k\} $
denote the set of all feasible solutions, it suffices to ensure that
\begin{equation}
\Pr\big[\forall S\in\mathcal{F}\colon|\widehat{I^{(\tau)}}(S)-I^{(\tau)}(S)|\leq\epsilon\opt\big]\geq1-\delta\label{eq:ic-additive-error-guarantee}
\end{equation}
where the above probability is over samples from the~model.

Since we have $|\Reach_{G_{i}}^{(\tau)}(S)|\leq n$, the
Chernoff inequality implies that the number of samples needed to guarantee
(\ref{eq:ic-additive-error-guarantee}) is $\ell=O(\frac{n\log(|\mathcal{F}|/\delta)}{\epsilon^{2}})=O(\frac{nk\log(n/\delta)}{\epsilon^{2}})$.
In this work, we show the following improved sample complexity.
\begin{thm}
\label{thm:ic-sample-complexity}Let $\mathcal{F}=\{ S\subseteq V\colon|S|\leq k\} $
be the set of all feasible solutions. Let $\widehat{I^{(\tau)}}\colon2^{V}\to\Reals$,
$\widehat{I^{(\tau)}}(S)=\frac{1}{\ell}\sum_{i=1}^{\ell}|\Reach_{G_{i}}^{(\tau)}(S)|$
where $G_{1},\dots,G_{\ell}$ are live-edge graphs sampled independently
from the independent cascade model. For any values $\epsilon,\delta\in[0,1]$,
we can ensure that 
\[
\Pr\big[\forall S\in\mathcal{F}\colon|\widehat{I^{(\tau)}}(S)-I^{(\tau)}(S)|\leq\epsilon\opt\big]\geq1-\delta
\]
 using $\ell=O(\frac{M\tau\log(|\mathcal{F}|/\delta)}{\opt\epsilon^{2}})$
samples, where $M=\max_{v\in V}I^{(\tau)}(v)$ is the maximum
singleton influence and $\opt=\max_{S\colon|S|\leq k}I^{(\tau)}(S)$
is the value of the optimal solution. Since $M\leq\opt$ and $|\mathcal{F}|=O(n^{k})$,
we have $\ell\leq O(\frac{\tau k\log(n/\delta)}{\epsilon^{2}})$.
\end{thm}
Compared to the sample complexity obtained via the Chernoff inequality,
our sample complexity replaces the dominant term $n=|V|$
by the number $\tau\leq n$ of propagation steps. This is beneficial
since in many applications $\tau$ is very small compared to $n$,
e.g., several applications exhibit small-world phenomena where the
restricted influence $I^{(\tau)}$ for a small $\tau$ sufficiently
approximates the unrestricted influence $I^{(n)}$ \citep{milgram06}.
Another example is in time-critical applications where $\tau$ is
explicitly set to a small value \citep{liu12,gomezrodriguez11,chen12}.

\paragraph{Novel concentration bound} Standard Chernoff bounds use $X\leq n$,
yielding $O(n)$ in sample complexity. We exploit the structure of IC:
the reachability set grows layer by layer over $\tau$ steps. We
decompose $X=|\Reach_{G}^{(\tau)}(S)|$ across propagation steps and
bound each layer's contribution, yielding $O(\tau)$ instead of $O(n)$.
Technically, we bound the moment generating function (all moments),
extending \citet{vondrak10}'s concentration for submodular functions.
Prior work \citep{sadeh20} only bounds variance, requiring the less
efficient median-of-means estimator. Our bound enables the submodular
sample average, allowing us to exploit sketching for fast running~time.

It is also instructive to compare our number of samples to the number
of samples of the reverse reachability approach used by the state
of the art algorithm of \citet{borgs14}. In the latter approach,
a single sample is generated by selecting a single vertex $v$ uniformly
at random and identifying the set of nodes that can reach $v$ in
a single live-edge graph $G$ sampled from the model. There are simple
examples showing that this sampling approach requires an $\Omega(n)$
number of samples even for $\tau=1$. Indeed, consider the case when
the graph is a matching; to estimate the ($1$-step) influence of
a single node in this instance, we need to sample its only neighbor
in the graph, which requires $\Omega(n)$ samples even
in expectation.

\paragraph{Our algorithm} Our Algorithm \ref{alg:independent-cascades}
for independent cascade models takes $\ell=O(\frac{\tau k\log(n/\delta)}{\epsilon^{2}})$
samples, and it uses our Algorithm \ref{alg:observed-cascades} for
observed cascades to approximately maximize the empirical influence
$\widehat{I^{(\tau)}}(S)=\frac{1}{\ell}\sum_{i=1}^{\ell}|\Reach_{G_{i}}^{(\tau)}(S)|$.
Using algorithms for sampling from the Geometric distribution \citep{BringmannF13},
we can construct the sampled subgraphs in expected time $O(m+\ell\overline{m})$
where $\overline{m}=\sum_{e}p_{e}$ is the expected number of edges
in a single sample (see Lemma \ref{lem:ic-sampling} in the supplement).
By combining Theorems \ref{thm:observed-cascades-main} and \ref{thm:ic-sample-complexity},
we obtain the following guarantees. 
\begin{thm}
\label{thm:ic-main} For any $\epsilon,\delta$ given as input, Algorithm
\ref{alg:independent-cascades} achieves an approximation of $1-\frac{1}{e}-\epsilon$
with probability $1-\delta$ using $\ell=O(\frac{\tau k}{\epsilon^{2}}\log(\frac{n}{\delta}))$
samples. The running time is $\widetilde{O}(m+Nb\ell\overline{m}\log b)$
for $N=\widetilde{O}(\frac{1}{\epsilon^{2}}\log(\frac{1}{\delta}))$
and $b=\widetilde{O}(\frac{1}{\epsilon^{2}}\log(\frac{\ell N}{\delta}))$.
\end{thm}

\paragraph{Comparison with prior work} State-of-the-art algorithms use
reverse influence sampling \citep{borgs14,tang14}, running in time
$O(\frac{1}{\epsilon^{2}}mk\log(n/\delta))$. Only \citet{guo22} improve
this asymptotically under certain conditions. Our algorithm is faster
by a factor of $\frac{m}{\widetilde{\Theta}(\tau\overline{m})}$ when
$m\gg\overline{m}\tau$. One common setting where this holds is when
edge probabilities scale with degrees and $\tau$ is~constant. We note
that our IC algorithm is not designed to compete with highly optimized
implementations of reverse influence sampling (RIS) such as
\citep{tang15,nguyen16,tang18} but rather offers improved asymptotic
guarantees in the regime where $\tau\ll n$.

%% file: experiments.tex
\section{\label{sec:experiments} Experimental evaluation}

\begin{figure}
\centering{}\includegraphics[width=0.6\linewidth]{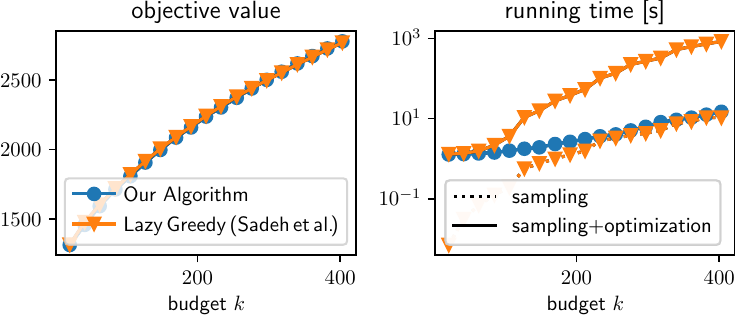}
\caption{\label{fig:lt-facebook-large-2} Influence maximization for a general
threshold model on Facebook for $\tau=5$. }

\bigskip

\centering{}\includegraphics[width=0.6\linewidth]{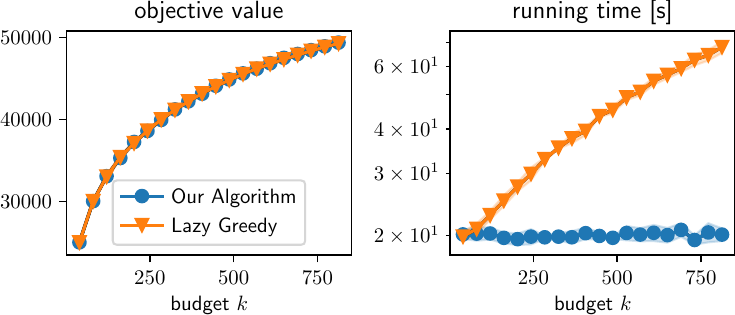}
\caption{\label{fig:lt-facebook-large-1-2} Influence maximization for $\tau=5$
steps on $\ell=100$ observed cascades from the Twitter network.}

\bigskip

\centering{}\includegraphics[width=0.6\linewidth]{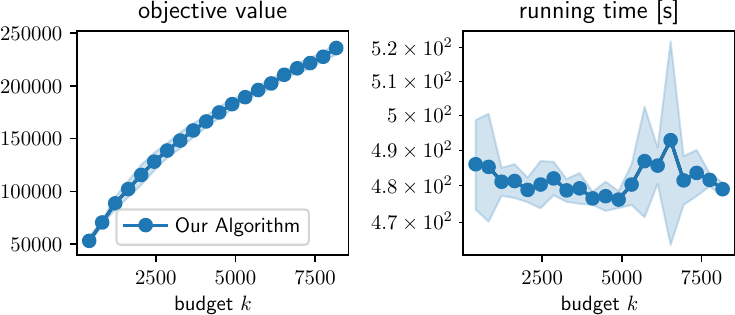}
\caption{\label{fig:lt-facebook-large-1-1-1} Influence maximization for $\tau=n$
steps on $\ell=100$ observed cascades from the Pokec network.}
\end{figure}

\begin{table}
\begin{centering}
\caption{\label{tab:datasets} Dataset Statistics}
\par\end{centering}
\centering{}
\small
\begin{tabular}{lrrr}
\hline 
\rule{0pt}{2.4ex}Dataset & $n$ & $m$ & $m_{\mathrm{tot}}$\\[0.5ex]
\hline 
\rule{0pt}{2.4ex}Facebook & $4\,039$ & $88\,234$ & $1\,529\,531\pm\phantom{1\,}757$\\
Twitter & $81\,306$ & $1\,768\,149$ & $35\,678\,239\pm1\,734$\\
Google-Plus & $107\,614$ & $13\,673\,453$ & $141\,708\,703\pm7\,858$\\
Pokec & $1\,632\,803$ & $30\,622\,564$ & $294\,689\,370\pm2\,357$\\[0.5ex]
\hline 
\end{tabular}
\end{table}

We evaluate on datasets in Table~\ref{tab:datasets} \citep{snapnets};
further details and results are in the appendix. Experiments run on
an AMD Ryzen 7900X3D with 94GB RAM; code in C++ with GCC O2. We repeat
5 times and report mean $\pm$ std. We provide our code in the supplementary~materials.

\subsection{Influence maximization in general stochastic diffusion models}

We use a general threshold model with linear weights $b_{(u,v)}=\frac{2}{\deg(v)}$,
which has no live-edge characterization (weights sum to $>1$). We compare our Algorithm~\ref{alg:general-model-empirical-variance}
with \citet{sadeh20} (Algorithm~\ref{alg:estimate-gains-variance}).
Results are in Figure~\ref{fig:lt-facebook-large-2} for $\tau=5$;
see appendix for $\tau\in\{10,n\}$. These are the largest networks
in the literature for general threshold models.

\paragraph{Discussion} Our Alg.~\ref{alg:general-model-empirical-variance}
suffers only a small decrease in objective value while taking substantially
fewer samples and running faster. For large budgets, our
algorithm uses 10x fewer samples and is 200x faster. 

\subsection{Influence maximization on observed cascades}

We obtain a single graph $G_{i}=\left(V,E_{i}\right)$ by including
$(u,v)\in E_{i}$ with probability $\frac{1}{\deg(u)}$
for each directed edge $(u,v)\in E$. To obtain a full instance for
the observed cascades model, we sample $\ell=100$ such graphs independently.
Table \ref{tab:datasets} shows the total number of edges $m_{\mathrm{tot}}$
in the instances we create. We use our Algorithm \ref{alg:observed-cascades}
and estimate the influence via Algorithm \ref{alg:estimate-gains-sketch-unrestricted}
for the unrestricted influence ($\tau=n$) and Algorithm \ref{alg:estimate-gains-sketch-restricted}
for the restricted influence ($\tau<n$). We compare our algorithm
with the state-of-the-art algorithm with theoretical guarantees which
is the greedy algorithm implemented with lazy evaluations; all other
algorithms with theoretical guarantees are for stochastic models \citep{tang15,guo22}.
Figures \ref{fig:lt-facebook-large-1-2} and \ref{fig:lt-facebook-large-1-1-1}
show results for the Twitter and Pokec networks and we provide further
results in the appendix. Lazy greedy does not scale
to Pokec and Google-Plus instances, so we exclude it there.

\paragraph{Discussion} Our algorithm vastly outperforms lazy greedy
in running time while achieving nearly the same
objective value, and easily scales to networks like Pokec with 100M
edges, requiring only a few minutes even for large~budgets.
Additional experimental results, including further datasets, parameter
settings, and diffusion step values, are provided in the appendix.

%% file: appendix.tex
\global\long\def\iter{(r,t)}%

\global\long\def\Tinner{T_{\mathrm{inner}}}%

\global\long\def\Touter{T_{\mathrm{outer}}}%

\section*{Outline of the appendix}

Our appendix is structured as follows:
\begin{itemize}
\item Appendix \ref{sec:appendix-concentration-bounds} contains concentration
bounds that we use throughout our proofs in this paper.
\item Appendix \ref{sec:appendix-greedy} contains the complete pseudocode
and proof of our low-adaptivity Greedy algorithm.
\item Appendix \ref{sec:appendix-general-models} contains the complete
pseudocode and proofs for our estimation routines in a general stochastic
model, and pseudocode and proofs for our algorithms for influence
maximization in this model.
\item Appendix \ref{sec:appendix-oc} contains the complete pseudocode and
proofs for our estimation routines in the observed cascades model,
and the pseudocode and proof for our algorithm for influence maximization
in the observed cascades model.
\item Appendix \ref{sec:concentration} contains the proof for our concentration
bound in the independent cascades model.
\item Appendix \ref{sec:analysis-live-edge} contains pseudocode and the
proof for our algorithm for influence maximization in the independent
cascades model.
\item Appendix \ref{sec:appendix-experiments} contains further experimental
details and results that we deferred from the main body.
\end{itemize}

\begin{table}[h]
\centering
\caption{Summary of notation.}
\label{tab:notation}
\begin{tabular}{ll}
\toprule
Symbol & Description \\
\midrule
$n = |V|$ & Number of nodes \\
$m = |E|$ & Number of edges \\
$k$ & Budget (maximum seed set size) \\
$\tau$ & Number of diffusion steps \\
$S$ & Seed set \\
$\mtot$ & Total edges across observed cascades: $\sum_i |E_i|$ \\[4pt]
\midrule
\rule{0pt}{2.4ex}$I^{(\tau)}(S)$ & $\tau$-step influence: $\E[|\Active^{(\tau)}(S)|]$ \\
$\Reach_G^{(\tau)}(S)$ & Nodes reachable from $S$ in $G$ within $\tau$ steps \\
$M$ & Maximum singleton influence: $\max_{v \in V} I^{(\tau)}(v)$ \\
$\widehat{f}$ & Sample of the influence function $I^{(\tau)}$ \\
$\ell$ & Number of samples or observed cascades \\
$L$ & Number of samples per call to $\mathrm{EstimateGains}$ \\
$N$ & Number of calls to $\mathrm{EstimateGains}$ \\[4pt]
\midrule
\rule{0pt}{2.4ex}$\epsilon$ & Approximation error parameter \\
$\delta$ & Failure probability parameter \\
$\epsest$ & Multiplicative estimation error \\
$\cest$ & Additive estimation error \\
$\prest$ & Estimation failure probability \\
\bottomrule
\end{tabular}
\end{table}

\section{\label{sec:appendix-concentration-bounds} Concentration bounds}

We now provide theorem statements for concentration bounds that we
use throughout the paper.
\begin{thm}
\label{thm:empirical-bernstein} (Empirical Bernstein bound: Theorem
4 in \citet{maurer09}) Let $X_{1},X_{2},\dots,X_{\ell}$ be independent
and identically distributed random variables satisfying $\left|X_{i}\right|\leq M$
for all $i\in\left\{ 1,2,\dots,\ell\right\} $. Let $\widetilde{V}=\frac{1}{\ell\left(\ell-1\right)}\sum_{1\leq i<j\leq\ell}\left(X_{i}-X_{j}\right)^{2}$
be the empirical variance of the samples $X_{1},\dots,X_{\ell}$.
Then, with probability $1-\delta$, we have 
\[
\left|\frac{1}{\ell}\sum_{i=1}^{\ell}X_{i}-\E\left[\frac{1}{\ell}\sum_{i=1}^{\ell}X_{i}\right]\right|\leq\sqrt{\frac{2\widetilde{V}\ln\left(4/\delta\right)}{\ell}}+\frac{7M\ln\left(4/\delta\right)}{3\left(\ell-1\right)}
\]
\end{thm}
\begin{thm}
\label{thm:chernoff-lower} (Chernoff Bound Lower Tail ) Let $X_{1},X_{2},\dots,X_{\ell}$
be independent and identically distributed random variables satisfying
$\left|X_{i}\right|\leq M$ for all $i\in\left\{ 1,2,\dots,\ell\right\} $.
Let $\mu=\E\left[\sum_{i=1}^{\ell}X_{i}\right]$. Then, for any $0<\delta<1$,
\[
\Pr\left[\sum_{i=1}^{\ell}X_{i}\le\left(1-\delta\right)\mu\right]\le\exp\left(-\frac{\mu\delta^{2}}{2M}\right).
\]
\end{thm}
\begin{thm}
\label{thm:chernoff-upper} (Chernoff Bound Upper Tail) Let $X_{1},X_{2},\dots,X_{\ell}$
be independent and identically distributed random variables satisfying
$\left|X_{i}\right|\leq M$ for all $i\in\left\{ 1,2,\dots,\ell\right\} $.
Let $\mu=\E\left[\sum_{i=1}^{\ell}X_{i}\right]$. Then, for any $\delta>0$,
\[
\Pr\left[\sum_{i=1}^{\ell}X_{i}\ge\left(1+\delta\right)\mu\right]\le\exp\left(-\frac{\mu\delta^{2}}{M\left(2+\delta\right)}\right).
\]
\end{thm}
\begin{thm}
\label{thm:submodular-mgf} (Theorem 3.1 in \citet{vondrak10}) Let
$f\colon2^{[n]}\to\Reals_{\ge0}$ be a monotone submodular function
with maximum marginal gain $\max_{v\in V}\left\{ f(v)-f(\emptyset)\right\} \leq1$.
For any $\lambda\in\Reals$ and any $p\in\left[0,1\right]^{n}$, we
have
\begin{align*}
\E_{S\sim p}\left[\exp\left(\lambda f(S)\right)\right] & \leq\exp\left(\left(e^{\lambda}-1\right)\E_{S\sim p}\left[f\left(S\right)\right]\right)
\end{align*}
\end{thm}
\begin{thm}
\label{thm:chen} \textup{(Lemma~5, \citet{chen21})} Suppose there
is a sequence of $n$ Bernoulli trials $X_{1},\dots,X_{n}$
where the success probability of $X_{i}$ depends on the results
of the preceding trials $X_{1},\dots,X_{i-1}$. Suppose it holds that 
\[
\Pr\left[X_{i}=1\mid X_{1}=x_{1},X_{2}=x_{2},\dots,X_{i-1}=x_{i-1}\right]\ge\eta,
\]
where $\eta>0$ is a constant and $x_{1},\dots,x_{i-1}$ are arbitrary.
Then, if $Y_{1},\dots,Y_{n}$ are independent Bernoulli trials, each
with probability $\eta$ of success, then
\[
\Pr\left[\sum_{i=1}^{n}X_{i}\le b\right]\le\Pr\left[\sum_{i=1}^{n}Y_{i}\le b\right]
\]
where $b$ is an arbitrary integer. 
\end{thm}

\section{\label{sec:appendix-greedy} Our low-adaptivity Greedy algorithm}

\label{sec:low-adaptivity-greedy}

In this section, we showcase the complete pseudocode for our template
low-adaptivity Greedy algorithm, which we introduced in an idealized
way in the main body as Algorithm \ref{alg:template-idealized}. We
then prove Theorem \ref{thm:template-ideal}. 

\begin{algorithm}
\begin{raggedright}
\caption{Low-adaptive algorithm for maximizing the influence.}
 \label{alg:optimization}
\par\end{raggedright}
\begin{raggedright}
$\mathrm{LowAdaptiveGreedy}\left(\epsilon,\delta,\epsest,\cest,\prest,\optguess\right)$:
\par\end{raggedright}
\begin{raggedright}
\textbf{\textcolor{blue}{Input: }}\textcolor{blue}{Approximation error
$\epsilon>0$, failure probability for the optimization $\delta>0$,
estimation error $\cest>0$, failure probability for the estimation
$\prest$, and an estimate $\optguess$ to the optimum solution value
$\opt$}
\par\end{raggedright}
\begin{raggedright}
\textbf{\textcolor{blue}{Output:}}\textcolor{blue}{{} With probability
$1-N\prest-\delta$, the algorithm returns a solution $S$ with $\left|S\right|\le k$
and $I^{(\tau)}(S)\geq\left(1-\frac{1}{e}-\epsilon-\epsest\right)\opt-2k\cest$
using $N=\frac{2}{\epsilon}\Tinner$ calls to $\mathrm{EstimateGains}$}
\par\end{raggedright}
\begin{raggedright}
$S\gets\emptyset$
\par\end{raggedright}
\begin{raggedright}
\textbf{for }$r=1,\dots,\Touter=\frac{2}{\epsilon}$ \textbf{do}
\par\end{raggedright}
\begin{raggedright}
$\qquad$$\alpha\gets\frac{1}{k}\optguess\left(1-\epsilon\right)^{r}$
\par\end{raggedright}
\begin{raggedright}
$\qquad$$U\gets V$
\par\end{raggedright}
\begin{raggedright}
$\qquad$\textbf{for} $t=1,\dots,\Tinner=\frac{8}{\epsilon}\log\left(\frac{n}{\delta}\right)$
\textbf{do}
\par\end{raggedright}
\begin{raggedright}
$\qquad\qquad$$\left\{ X\left(\left\{ u\right\} \mid S\right)\colon u\in U\right\} \gets\mathrm{EstimateGains}\left(\left\{ \left\{ u\right\} \colon u\in U\right\} ,S,\epsest,\cest,\prest\right)$
\par\end{raggedright}
\begin{raggedright}
$\qquad\qquad$$U\gets\left\{ u\in U:X\left(\left\{ u\right\} \mid S\right)\ge\alpha\right\} $\textcolor{blue}{$\hfill$/$\negmedspace$/
filtering}
\par\end{raggedright}
\begin{raggedright}
$\qquad\qquad$\textbf{if} $U=\emptyset$ \textbf{then}
\par\end{raggedright}
\begin{raggedright}
\textbf{$\qquad\qquad\qquad$break }
\par\end{raggedright}
\begin{raggedright}
$\qquad\qquad$Let $v_{1},v_{2},\dots,v_{\left|U\right|}$ be a random
permutation of the elements in $U$
\par\end{raggedright}
\begin{raggedright}
$\qquad\qquad$$s\gets\min\left\{ k-\left|S\right|,\left|U\right|\right\} $
\par\end{raggedright}
\begin{raggedright}
$\qquad\qquad$Let $T_{i}=\left\{ v_{1},v_{2},\dots,v_{i}\right\} $
for all $1\le i\le s$\textcolor{blue}{$\hfill$/$\negmedspace$/
prefixes of random permutation}
\par\end{raggedright}
\begin{raggedright}
$\qquad\qquad$$\left\{ X\left(T_{i}\mid S\right)\colon1\le i\le s\right\} \gets\mathrm{EstimateGains}\left(\left\{ T_{i}:1\le i\le s\right\} ,S,\epsest,\cest,\prest\right)$
\par\end{raggedright}
\begin{raggedright}
$\qquad\qquad$$i^{*}\gets\arg\max\left\{ 1\le i\le s:X\left(T_{i}\mid S\right)\ge\left(1-\epsilon\right)\left(\left(1-\epsest\right)\alpha-\cest\right)i\right\} $
\par\end{raggedright}
\begin{raggedright}
$\qquad\qquad$$S\gets S\cup T_{i^{*}}$\textcolor{blue}{$\hfill$/$\negmedspace$/
adding elements to the solution}
\par\end{raggedright}
\begin{raggedright}
$\qquad\qquad$\textbf{if} $\left|S\right|=k$ \textbf{then}
\par\end{raggedright}
\begin{raggedright}
$\qquad\qquad\qquad$\textbf{return} $S$
\par\end{raggedright}
\begin{raggedright}
$\qquad$\textbf{if} $U\not=\emptyset$ \textbf{then}
\par\end{raggedright}
\begin{raggedright}
$\qquad\qquad$\textbf{return} Failure
\par\end{raggedright}
\raggedright{}\textbf{return $S$}
\end{algorithm}

\begin{thm}
(Theorem \ref{thm:template-ideal}) Suppose we run Algorithm \ref{alg:optimization}
with an estimation subroutine $\mathrm{EstimateGains}$ with multiplicative
error $\epsest$, additive error $\cest$, and failure probability
$\prest$. Algorithm \ref{alg:optimization} makes $N=\widetilde{O}\left(\frac{1}{\epsilon^{2}}\log\left(\frac{1}{\delta}\right)\right)$
calls to $\mathrm{EstimateGains}$. If $\frac{1}{2}\opt\le\optguess\le2\opt$,
the solution $S$ satisfies $I^{(\tau)}\left(S\right)\ge\left(1-\frac{1}{e}-\epsest-\epsilon\right)\opt-3k\cest$
with probability $1-N\prest-\delta$. 
\end{thm}
Let us denote with superscript $(r,t)$ the value of a variable during
the $r$-th iteration of the outer loop and at the end of the $t$-th
iteration of the inner loop. We denote with $\Tinner$ the last iteration
of the inner for loop (even though we may break break the loop earlier)
and let $(r,\Tinner)=(r+1,0)$. 
\begin{lem}
\label{lem:optimization-failure} Suppose we run Algorithm \ref{alg:optimization}
conditioned on the event that the estimation subroutine\linebreak
$\mathrm{EstimateGains}$ has multiplicative error $\epsest$ and
additive error $\cest$. Then, Algorithm \ref{alg:optimization}
declares failure with probability at most $\delta$. 
\end{lem}
\begin{proof}
Let $1\le r\le\Touter$ be any iteration of the outer for loop and
$1\le t<\Tinner$ an iteration of the inner for loop such that $\left|S^{(r,t)}\right|<k$.
We first show:

Claim 1: $\Pr\left[\left|U^{(r,t+1)}\right|\le\left(1-\frac{1}{2}\epsilon\right)\left|U^{(r,t)}\right|\right]\ge\frac{1}{2}$,
i.e. we filter out at least an $\frac{1}{2}\epsilon$-fraction of
the elements with probability at least $\frac{1}{2}$.

We now prove Claim 1. Let $\left\{ v_{1},v_{2},\dots\right\} =U^{(r,t)}$
be the elements that were not filtered, in the ordering according
to the random permutation of iteration $t$. We define the sets 
\[
A_{i}=\left\{ v\in U^{(r,t)}:I\left(v\mid S^{(r,t-1)}\cup T_{i}^{(r,t)}\right)<\left(1-\epsest\right)\alpha^{(r)}-\cest\right\} 
\]
where we define the prefix $T_{i}^{(r,t)}=\left\{ v_{1},v_{2},\dots,v_{i}\right\} $
for $0\le i\le\left|U^{(r,t)}\right|$. The sequence $A_{0},A_{1},\dots,A_{\left|U^{(r,t)}\right|}$
has the following properties:
\begin{enumerate}
\item $A_{0}=\emptyset$: For any element $v\in U^{(r,t)}$ holds $X\left(v\mid S^{(r,t-1)}\right)\ge\alpha^{(r)}$
(as otherwise, the element $v$ would have been filtered out at the
beginning of iteration $t$) and thus
\[I\left(v\mid S^{(r,t-1)}\cup T_{0}^{(r,t)}\right)=I\left(v\mid S^{(r,t-1)}\right)\ge\left(1-\epsest\right)X\left(v\mid S^{(r,t-1)}\right)-\cest\ge\left(1-\epsest\right)\alpha^{(r)}-\cest.\]
\item $A_{\left|U^{(r,t)}\right|}=U^{(r,t)}$: For all elements $v\in U^{(r,t)}$
holds for that
\[I\left(v\mid S^{(r,t-1)}\cup T_{\left|U^{(r,t)}\right|}^{(r,t)}\right)=I\left(v\mid U^{(r,t)}\right)=0<\left(1-\epsest\right)\alpha^{(r)}-\cest.\]
\item $A_{i}\subseteq A_{i+1}$: Consider an element $v\in A_{i}$. Due
to submodularity,
\[I\left(v\mid S^{(r,t-1)}\cup T_{i+1}^{(r,t)}\right)\le I\left(v\mid S^{(r,t-1)}\cup T_{i}^{(r,t)}\right)<\left(1-\epsest\right)\alpha^{(r)}-\cest\]
and therefore $v\in A_{i+1}$.
\end{enumerate}
These properties allow us to define 
\[
i'=\min\left\{ i\in\mathbb{N}:\left|A_{i}\right|\ge\frac{\epsilon}{2}\left|U^{(r,t)}\right|\right\} 
\]
as the index where $\left|A_{i'}\right|\ge\frac{\epsilon}{2}\left|U^{(r,t)}\right|$
and $\left|A_{i'-1}\right|<\frac{\epsilon}{2}\left|U^{(r,t)}\right|$.
We now need the following two sub-claims about $i'$:
\begin{itemize}
\item Claim 1a: If $i^{*}\ge i'$ then $\left|U^{(r,t+1)}\right|\le\left(1-\frac{1}{2}\epsilon\right)\left|U^{(r,t)}\right|$.
\item Claim 1b: $\Pr\left[i^{*}\ge i'\right]\ge\frac{1}{2}$.
\end{itemize}
Let us first prove Claim 1a: Assume that $i^{*}\ge i'$ and consider
an element $v\in A_{i'}$. We have
\begin{align*}
X\left(v\mid S^{(r,t)}\right) & =X\left(v\mid S^{(r,t-1)}\cup T_{i^{*}}^{(r,t)}\right)\\
 & \le\left(1+\epsest\right)I\left(v\mid S^{(r,t-1)}\cup T_{i^{*}}^{(r,t)}\right)+\cest\\
 & \le\left(1+\epsest\right)I\left(v\mid S^{(r,t-1)}\cup T_{i}^{(r,t)}\right)+\cest\\
 & \le\left(1+\epsest\right)\left(\left(1-\epsest\right)\alpha^{(r)}-\cest\right)+\cest\\
 & \le\alpha^{(r)}.
\end{align*}
where the first inequality holds because we condition on the accuracy
of the estimation routine, and the second inequality follows from
submodularity. Thus, every element in $A_{i'}$ will be filtered in
the next iteration, i.e. $A_{i'}\cap U^{(r,t+1)}=\emptyset$. Since
$A_{i'}\subseteq U^{(r,t)}$, we have $\left|U^{(r,t+1)}\right|\le\left|U^{(r,t)}\right|-\left|A_{i'}\right|\le\left(1-\frac{\epsilon}{2}\right)\left|U^{(r,t)}\right|$.
This completes the proof for Claim 1a.

We now prove Claim 1b: By description of the algorithm, $i^{*}<i'$
is equivalent to 
\[
I\left(T_{i}^{(r,t)}\mid S^{(t,r-1)}\right)<\left(1-\epsilon\right)\left(\left(1-\epsest\right)\alpha^{(r)}-\cest\right)i\qquad\text{for all }s\ge i\ge i'
\]
which implies that $I\left(T_{i'}^{(r,t)}\mid S^{(r,t-1)}\right)<\left(1-\epsilon\right)\left(\left(1-\epsest\right)\alpha^{(r)}-\cest\right)i'$.
Therefore,
\[
\Pr\left[i^{*}<i'\right]\le\Pr\left[I\left(T_{i'}^{(r,t)}\mid S^{(r,t-1)}\right)>\left(1-\epsilon\right)\left(\left(1-\epsest\right)\alpha^{(r)}-\cest\right)i'\right].
\]
We say that $v_{i}$ is bad if $I\left(v_{i}\mid S^{(r,t-1)}\cup T_{i-1}^{(r,t)}\right)<\left(1-\epsest\right)\alpha-\cest$.
With this definition,
\begin{align*}
I\left(T_{i'}^{(r,t)}\mid S^{(t-1)}\right) & =\sum_{i=1}^{i'}I\left(v_{i}\mid S^{(r,t-1)}\cup T_{i-1}^{(r,t)}\right)\\
 & \ge\left(i'-\left|\left\{ 1\le i\le i':v_{i}\text{ is bad}\right\} \right|\right)\left(\left(1-\epsest\right)\alpha-\cest\right)
\end{align*}
and thus
\[
\Pr\left[I\left(T_{i'}^{(r,t)}\mid S^{(r,t-1)}\right)<\left(1-\epsilon\right)\left(\left(1-\epsest\right)\alpha-\cest\right)i'\right]\le\Pr\left[\left|\left\{ 1\le i\le i':v_{i}\text{ is bad}\right\} \right|>\epsilon i'\right].
\]
We can model the quantity $\left|\left\{ 1\le i\le i':v_{i}\text{ is bad}\right\} \right|$
as a sum of dependent Bernoulli random variables. Specifically, let
$Y_{i}=1$ if $v_{i}$ is bad and $Y_{i}=0$, otherwise, such that
$\left|\left\{ 1\le i\le i':v_{i}\text{ is bad}\right\} \right|=\sum_{i=1}^{i'}Y_{i}$.
Recall that we randomly permute the elements in $U^{(r,t)}$ which
is equivalent to uniformly picking the $i$-th element from the set
$U^{(r,t)}\setminus T_{i-1}^{(r,t)}$. Thus, conditioned on the outcomes
of the previous $i-1$ trials, 
\begin{align*}
\Pr\left[Y_{i}=1\right] & =\frac{\left|\left\{ v\in U^{(r,t)}\setminus T_{i-1}^{(r,t)}:I\left(v\mid S^{(r,t-1)}\cup T_{i-1}^{(r,t)}\right)<\left(1-\epsest\right)\alpha-\cest\right\} \right|}{\left|U^{(r,t)}\setminus T_{i-1}^{(r,t)}\right|}\\
 & \le\frac{\left|\left\{ v\in U^{(r,t)}:I\left(v\mid S^{(r,t-1)}\cup T_{i-1}^{(r,t)}\right)<\left(1-\epsest\right)\alpha-\cest\right\} \right|}{\left|U^{(r,t)}\right|}\\
 & =\frac{\left|A_{i-1}\right|}{\left|U^{(r,t)}\right|}\le\frac{\left|A_{i'-1}\right|}{\left|U^{(r,t)}\right|}\le\frac{\epsilon}{2}.
\end{align*}
We thus have $\E\left[\sum_{i=1}^{i'}Y_{i}\right]\le\frac{\epsilon}{2}i'$.
Therefore, by Markov's inequality,
\[
\Pr\left[i^{*}<i'\right]\le\Pr\left[\sum_{i=1}^{i'}Y_{i}\ge\epsilon i'\right]\le\frac{\E\left[\sum_{i=1}^{i'}Y_{i}\right]}{\epsilon i'}=\frac{1}{2}.
\]
We have thus established Claim 1b.

Combining Claim 1a and Claim 1b immediately implies Claim 1, i.e.\ that
\[\Pr\left[\left|U^{(r,t+1)}\right|\ge\left(1-\tfrac{1}{2}\epsilon\right)\left|U^{(r,t)}\right|\right]\ge\tfrac{1}{2}\]
for any iteration $t$ in which $\left|S^{(r,t)}\right|<k$.

It remains to argue about the failure probability. We begin with showing
that the probability of declaring failure for a single iteration of
the outer for loop is at most $\frac{\epsilon\delta}{2}$. We only
declare failure if $U^{(r,\Tinner)}\not=\emptyset$, i.e. we did not
filter out all elements. To analyze this event, let us define for
each $1\le t<\Tinner$ the indicator random variable $X_{t}$ for
the event that $\left|U^{(r,t+1)}\right|\le\left(1-\frac{1}{2}\epsilon\right)\left|U^{(r,t)}\right|$.
For $X=\sum_{t=1}^{\Tinner-1}X_{t}$ we then have $\left|U^{(r,\Tinner)}\right|\le\left|U^{(r,0)}\right|\left(1-\frac{1}{2}\epsilon\right)^{X}\le n\left(1-\frac{1}{2}\epsilon\right)^{X}$.
Therefore,
\begin{align*}
\Pr\left[U^{(r,t)}\not=\emptyset\right] & =\Pr\left[\left|U^{(r,t)}\right|\ge1\right]\\
 & \le\Pr\left[n\left(1-\frac{1}{2}\epsilon\right)^{X}\ge1\right]\\
 & =\Pr\left[X\le\frac{\log\left(\frac{1}{n}\right)}{\log\left(1-\frac{1}{2}\epsilon\right)}\right]\\
 & \le\Pr\left[X\le2\frac{\log n}{\epsilon}\right]
\end{align*}
From Claim 1, we know that $\Pr\left[X_{t}=1\right]\ge\frac{1}{2}$.
For each $1\le t\le\Tinner$ let now $Y_{t}$ be a Bernoulli random
variable with $\Pr\left[Y_{t}=1\right]\ge\frac{1}{2}$ and $Y=\sum_{t=1}^{\Tinner-1}Y_{t}$.
By Theorem \ref{thm:chen} and a Chernoff bound (Theorem \ref{thm:chernoff-lower}),
we have
\[
\Pr\left[X\le2\frac{\log n}{\epsilon}\right]\le\Pr\left[Y\le2\frac{\log n}{\epsilon}\right]
\le\Pr\left[Y\le\left(1-\frac{1}{2}\right)\frac{1}{2}\left(\Tinner-1\right)\right]
\le\exp\left(-\frac{\Tinner-1}{8}\right)\le\frac{\epsilon\delta}{2}
\]
by definition of $\Tinner$ and for sufficiently small $\epsilon,\delta$.

Since there are at most $\frac{2}{\epsilon}$ iterations of the outer
loop, we obtain an overall failure probability of at most $\delta$
by a union bound.
\end{proof}

\begin{lem}
\label{lem:optimization-guarantee} Suppose we run Algorithm \ref{alg:optimization}
conditioned on the event that the estimation subroutine\linebreak
$\mathrm{EstimateGains}$ obtains multiplicative error $\epsest$
and additive error $\cest$ in every call. Let us also condition on
the event that Algorithm \ref{alg:optimization} does not declare failure.
Then, Algorithm \ref{alg:optimization} constructs a solution $S$ satisfying
$I^{(\tau)}(S)\geq(1-\frac{1}{e}-\epsilon-\epsest)\opt-3k\cest$.
\end{lem}
\begin{proof}
Consider an iteration $1\le r\le\Touter$ of the outer for loop. We
first prove the following two claims:
\begin{itemize}
\item Claim 1: The increase in density is at least
\[
\frac{I\left(S^{(r,\Tinner)}\mid S^{(r,0)}\right)}{\left|S^{(r,\Tinner)}\setminus S^{(r,0)}\right|}\ge\left(1-\epsilon-2\epsest\right)\alpha^{(r)}-2\cest.
\]
\item Claim 2: If $\left|S^{(r,\Tinner)}\right|<k$ then $I\left(v\mid S^{(r,\Tinner)}\right)<\left(1+\epsest\right)\alpha^{(r)}+\cest$
for all $v\in V$.
\end{itemize}
To show Claim 1, we bound
\begin{align*}
I\left(S^{(r,\Tinner)}\mid S^{(r,0)}\right) & \ge\left(1-\epsest\right)X\left(S^{(r,\Tinner)}\mid S^{(r,0)}\right)-\cest\\
 & =\left(1-\epsest\right)\sum_{t=1}^{\Tinner}X\left(T_{i^{*}}^{(r,t)}\mid S^{(r,t-1)}\right)-\cest\\
 & \ge\left(1-\epsest\right)\left(1-\epsilon\right)\left(\left(1-\epsest\right)\alpha^{(r)}-\cest\right)\sum_{t=1}^{\Tinner}\left|T_{i^{*}}^{(r,t)}\right|-\cest\\
 & =\left(1-\epsest\right)\left(1-\epsilon\right)\left(\left(1-\epsest\right)\alpha^{(r)}-\cest\right)\left|S^{(r,\Tinner)}\setminus S^{(r,0)}\right|-\cest\\
 & =\left(1-\epsest\right)\left(1-\epsilon\right)\left(\left(1-\epsest\right)\alpha^{(r)}-\cest\right)\left|S^{(r,\Tinner)}\setminus S^{(r,0)}\right|-\cest\\
 & \ge\left(\left(1-\epsilon-2\epsest\right)\alpha^{(r)}-2\cest\right)\left|S^{(r,\Tinner)}\setminus S^{(r,0)}\right|
\end{align*}
where the first inequality holds because we conditioned on the accuracy
of the estimation routine, and the second inequality is by the algorithm
design. Equivalently,
\[
\frac{I\left(S^{(r,\Tinner)}\mid S^{(r,0)}\right)}{\left|S^{(r,\Tinner)}\setminus S^{(r,0)}\right|}\ge\left(1-\epsilon-2\epsest\right)\alpha^{(r)}-2\cest.
\]

To show Claim 2, consider now the case when $\left|S^{(r,\Tinner)}\right|<k$.
Since we conditioned on the event that we do not declare failure,
this means that we did filter out all elements, i.e. that $U^{(r,\Tinner)}=\emptyset$.
Let $v\in V$ be an element and let $1\le t\le\Tinner$ be maximal
such that $v\in U^{(r,t)}$. That is, we will filter $v$ in the following
iteration. By submodularity,
\[
I\left(v\mid S^{(r,\Tinner)}\right)\le I\left(v\mid S^{(r,t)}\right)\le\left(1+\epsest\right)X\left(v\mid S^{(r,t)}\right)+\cest<\left(1+\epsest\right)\alpha^{(r)}+\cest.
\]
where the first inequality holds because we conditioned on the accuracy
of the estimation.

We now use those two claims to establish the approximation guarantee.
Let us first consider the case where Algorithm \ref{alg:optimization}
outputs a set of size $\left|S\right|<k$. By the algorithm design,
this means that we go through all $\Touter=\frac{2}{\epsilon}$ iterations
of the outer loop and exit the algorithm with a threshold
\begin{equation}
\alpha^{(\Touter)}=\frac{1}{k}\optguess\left(1-\epsilon\right)^{2/\epsilon}\le\frac{1}{k}\optguess e^{-2}\le\frac{1}{2ek}\optguess.\label{eq:21}
\end{equation}
By Claim 2, all elements $v\in V$ have marginal gain $I\left(v\mid S^{(\Touter,\Tinner)}\right)<\left(1+\epsest\right)\alpha^{(\Touter)}+\cest$.
We can thus bound
\begin{align*}
\opt-I\left(S^{(\Touter,\Tinner)}\right) & \le I\left(\opt\mid S^{(\Touter,\Tinner)}\right) & \text{(monotonicity)}\\
 & \le\sum_{v\in S^{*}}I\left(v\mid S^{(\Touter,\Tinner)}\right) & \text{(submodularity)}\\
 & \le\sum_{v\in S^{*}}\left(\left(1+\epsest\right)\alpha^{(\Touter)}+\cest\right)\\
 & <\sum_{v\in S^{*}}\left(\frac{1+\epsest}{2ek}\optguess+\cest\right) & \text{(by (\ref{eq:21}))}\\
 & =\frac{1+\epsest}{2e}\optguess+k\cest\\
 & \le\frac{1+\epsest}{e}\opt+k\cest
\end{align*}
which is equivalent to 
\[
I\left(S\right)=I\left(S^{(\Touter,\Tinner)}\right)\ge\left(1-\frac{1+\epsest}{e}\right)\opt-k\cest\ge\left(1-\frac{1}{e}-\epsest\right)\opt-k\cest.
\]
We have thus shown that if $\left|S\right|<k$, we have a solution
such that $I\left(S\right)\ge\left(1-\frac{1}{e}-\epsest\right)\opt-k\cest$.

Let us now consider the second case where $\left|S\right|=k$. Let
$1\le r\le\Touter$ be the iteration in which we return a solution
$S=S^{(r,\Tinner)}$ and let $1\le r'\le r$. Recall that we use indices
$\left(r+1,0\right)=\left(r,\Tinner\right)$. We have $\left|S^{(r'-1,\Tinner)}\right|<k$
and therefore $\left(1+\epsest\right)\alpha^{(r'-1)}+\cest>I\left(v\mid S^{(r'-1,\Tinner)}\right)$
for all $v\in V$ by Claim 2. By averaging over the elements $v\in S^{*}$,
\begin{equation}
\begin{split}
\alpha^{(r'-1)}&\ge\frac{1}{1+\epsest}\left(\frac{1}{k}\sum_{v\in S^{*}}I\left(v\mid S^{(r'-1,\Tinner)}\right)-\cest\right)\\
&\ge\frac{1-\epsest}{k}\sum_{v\in S^{*}}I\left(v\mid S^{(r'-1,\Tinner)}\right)-\cest\ge\frac{1-\epsest}{k}\left(\opt-I\left(S^{(r'-1,\Tinner)}\right)\right)
\end{split}
\label{eq:3-1}
\end{equation}
where the final inequality is due to submodularity. We can now lower
bound the function value gain
\begin{align*}
I\left(S^{(r'+1,0)}\right)-I\left(S^{(r',0)}\right) & =I\left(S^{(r',\Tinner)}\right)-I\left(S^{(r',0)}\right)\\
 & =I\left(S^{(r',\Tinner)}\mid S^{(r',0)}\right)\\
 & \ge\left|S^{(r',\Tinner)}\setminus S^{(r',0)}\right|\left(\left(1-\epsilon-2\epsest\right)\alpha^{(r')}-2\cest\right)\\
 & \ge\left|S^{(r',\Tinner)}\setminus S^{(r',0)}\right|\left(\left(1-2\epsilon-2\epsest\right)\alpha^{(r'-1)}-2\cest\right)\\
 & \ge\left|S^{(r',\Tinner)}\setminus S^{(r',0)}\right|\left(\frac{1-2\epsilon-3\epsest}{k}\left(\opt-I\left(S^{(r'-1,\Tinner)}\right)\right)-2\cest\right)\\
 & \ge\left|S^{(r',\Tinner)}\setminus S^{(r',0)}\right|\left(\frac{1-2\epsilon-3\epsest}{k}\left(\opt-I\left(S^{(r'-1,\Tinner)}\right)\right)-3\cest\right)\\
 & =\left|S^{(r'+1,0)}\setminus S^{(r',0)}\right|\left(\frac{1-2\epsilon-3\epsest}{k}\left(\opt-I\left(S^{(r',0)}\right)\right)-3\cest\right)
\end{align*}
where the first inequality is due to Claim (1), the second since $\alpha^{(r')}=\left(1-\epsilon\right)\alpha^{(r'-1)}$,
and the third inequality is due to Equation \ref{eq:3-1}. Equivalently,
\begin{equation}
\begin{split}
\opt-I\left(S^{(r'+1,0)}\right)&\le\left(1-\left|S^{(r'+1,0)}\setminus S^{(r',0)}\right|\frac{1-2\epsilon-3\epsest}{k}\right)\left(\opt-I\left(S^{(r',0)}\right)\right)\\
&\quad+3\cest\left|S^{(r'+1,0)}\setminus S^{(r',0)}\right|.
\end{split}
\label{eq:22}
\end{equation}
We use this to show by induction over the rounds $0\le r'\le r$ that
\begin{equation}
\begin{split}
\opt-I\left(S^{(r+1,0)}\right)&\le3\cest\left|S^{(r+1,0)}\setminus S^{(r'+1,0)}\right|\\
&\quad+\left(\opt-I\left(S^{(r'+1,0)}\right)\right)\prod_{j=r'+1}^{r}\left(1-\left|S^{(j+1,0)}\setminus S^{(j,0)}\right|\frac{1-2\epsilon-3\epsest}{k}\right).
\end{split}
\label{eq:20}
\end{equation}
This is clearly true for $r'=r$. Let us thus assume that (\ref{eq:20})
is true for $r$. We show that it also holds for $r'-1\ge0$:
\begin{align*}
 & \opt-I\left(S^{(r+1,0)}\right)\\
 & \le3\cest\left|S^{(r+1,0)}\setminus S^{(r'+1,0)}\right|+\left(\opt-I\left(S^{(r',0)}\right)\right)\prod_{j=r'+1}^{r}\left(1-\left|S^{(j+1,0)}\setminus S^{(j,0)}\right|\frac{1-2\epsilon-3\epsest}{k}\right)\\
 & \le3\cest\left|S^{(r+1,0)}\setminus S^{(r'+1,0)}\right|\\
 & \quad+\left(1-\left|S^{(r'+1,0)}\setminus S^{(r',0)}\right|\frac{1-2\epsilon-2\epsest}{k}+3\cest\left|S^{(r'+1,0)}\setminus S^{(r',0)}\right|\right)\\
 & \qquad\cdot\left(\opt-I\left(S^{(r'-1,0)}\right)\right)\prod_{j=r'+1}^{r}\left(1-\left|S^{(j+1,0)}\setminus S^{(j,0)}\right|\frac{1-2\epsilon-3\epsest}{k}\right)\\
 & \le3\cest\left|S^{(r+1,0)}\setminus S^{(r'+1,0)}\right|\\
 & \quad+\left(\opt-I\left(S^{(r'-1,0)}\right)\right)\prod_{j=r'}^{r}\left(1-\left|S^{(j+1,0)}\setminus S^{(j,0)}\right|\frac{1-2\epsilon-3\epsest}{k}\right)\\
 & \quad+3\cest\left|S^{(r'+1,0)}\setminus S^{(r',0)}\right|\underbrace{\prod_{j=r'+1}^{r}\left(1-\left|S^{(j+1,0)}\setminus S^{(j,0)}\right|\frac{1-2\epsilon-3\epsest}{k}\right)}_{\le1}\\
 & \le3\cest\left|S^{(r,0)}\setminus S^{(r'+1,0)}\right|+3\cest\left|S^{(r,0)}\setminus S^{(r',0)}\right|\\
 & \quad+\left(\opt-I\left(S^{(r'-1,0)}\right)\right)\prod_{j=r'}^{r}\left(1-\left|S^{(j+1,0)}\setminus S^{(j,0)}\right|\frac{1-2\epsilon-3\epsest}{k}\right)\\
 & \le3\cest\left|S^{(r,0)}\setminus S^{(r',0)}\right|+\left(\opt-I\left(S^{(r'-1,0)}\right)\right)\prod_{j=r'}^{r}\left(1-\left|S^{(j+1,0)}\setminus S^{(j,0)}\right|\frac{1-2\epsilon-3\epsest}{k}\right)
\end{align*}
We have thus established (\ref{eq:20}) by induction for all $0\le r'\le r$.
Applying (\ref{eq:20}) for $r'=0$ implies
\begin{align*}
\opt-I\left(S^{(r+1,0)}\right) & \le3\cest\left|S^{(r+1,0)}\right|+\opt\prod_{j=1}^{r}\left(1-\left|S^{(j+1,0)}\setminus S^{(j,0)}\right|\frac{1-2\epsilon-3\epsest}{k}\right)\\
 & \le3k\cest+\opt\prod_{j=1}^{r}\left(1-\left|S^{(j+1,0)}\setminus S^{(j,0)}\right|\frac{1-2\epsilon-3\epsest}{k}\right)\\
 & \le3k\cest+\opt\exp\left(-\frac{1-2\epsilon-3\epsest}{k}\sum_{j=1}^{r}\left|S^{(j+1,0)}\setminus S^{(j,0)}\right|\right)\\
 & =3k\cest+\opt\exp\left(-1+2\epsilon+3\epsest\right)\\
 & \le3k\cest+\opt\left(\frac{1}{e}+\epsilon+\epsest\right)
\end{align*}
where the last inequality holds for sufficiently small $\epsilon,\epsest$.
Therefore,
\[
I\left(S\right)=I\left(S^{(r+1,0)}\right)\ge\opt\left(1-\frac{1}{e}-\epsilon-\epsest\right)-3k\cest.
\]
\end{proof}

\begin{lem}
\label{lem:optimization-num-calls} Algorithm \ref{alg:optimization}
makes at most $N=\widetilde{O}(\frac{1}{\epsilon^{2}}\log(\frac{1}{\delta}))$
calls to the estimation subroutine\linebreak $\mathrm{EstimateGains}$.
\end{lem}
\begin{proof}
There are at most $\frac{1}{\epsilon}$ iterations of the outer loop.
During each iteration of the outer loop, there are at most $T$ iterations
of the inner loop. During each iteration of the inner for loop, we
make two calls to $\mathrm{EstimateGains}$. In total, there are thus
at most $N=\frac{2}{\epsilon}T=O\left(\frac{1}{\epsilon^{2}}\log\left(\frac{n}{\delta}\right)\right)=\widetilde{O}\left(\frac{1}{\epsilon^{2}}\log\left(\frac{1}{\delta}\right)\right)$
calls to $\mathrm{EstimateGains}$. 
\end{proof}

\begin{proof}
(Theorem \ref{thm:template-ideal}) We have two failure events. The
first failure event is that any call to $\mathrm{EstimateGains}$
fails to obtain multiplicative error $\epsest$ or additive error
$\cest$ with probability at most $\prest$. Since we use a total
of $N$ calls during the execution of $\mathrm{DescendingThresholdGreedy}$
due to Lemma \ref{lem:optimization-num-calls}, we know by a union
bound that all calls to $\mathrm{EstimateGains}$ provide multiplicative
error $\epsest$ and additive error $\cest$ with probability at least
$1-N\prest$. The second failure event is that the $\mathrm{DescendingThresholdGreedy}$
returns failure. Conditioned on the first failure event, Lemma \ref{lem:optimization-failure}
shows that Algorithm \ref{alg:optimization} does not declare failure
with probability at least $1-\delta$. By another union bound, we
know that neither failure event happens with probability at least
$1-N\prest-\delta$. Conditioned on this, Lemma \ref{lem:optimization-guarantee}
states that the algorithm constructs a solution $S$ satisfying $I^{(\tau)}\left(S\right)\geq\left(1-\frac{1}{e}-\epsilon-\epsest\right)\opt-3k\cest$. 
\end{proof}

\section{\label{sec:appendix-general-models} Our algorithms for influence
maximization in general stochastic diffusion models}

We now provide two versions of our algorithm for general stochastic
diffusion models. In the first version, we assume that the algorithm
is given as input an upper bound $W$ on the variance
$\max_{S,T\subseteq V:\left|S\cup T\right|\le k}\Var[\widehat{f}(T\mid S)]$
and we estimate the marginal gains via a median-of-means estimator.
We provide a full pseudocode for this estimation subroutine. We also
provide the pseudocode for the instantiation of \ref{alg:optimization}
and we prove Theorem \ref{thm:gtm-estimate-gains-variance}. In the
second version, the algorithm is not given such an upper bound on
the variance, and instead uses the empirical variance of the samples
to determine how many samples to take. We also provide a full pseudocode
for both the estimation and optimization routine and prove Theorem
\ref{thm:gtm-estimate-gains-empirical-variance}. 

\subsection{Our algorithm using a median-of-means estimator}

\begin{algorithm}
\begin{raggedright}
\caption{\label{alg:general-model-variance} Our algorithm for general models
that instantiates the template algorithm \ref{alg:optimization} using
the estimation algorithm \ref{alg:estimate-gains-variance}.}
\par\end{raggedright}
\begin{raggedright}
$\mathrm{GeneralModel}\left(\epsilon,\delta,W\right):$
\par\end{raggedright}
\begin{raggedright}
\textbf{\textcolor{blue}{Input:}}\textcolor{blue}{{} Error parameters
$\epsilon$ and $\delta$, upper bound $W$ on the maximum variance
$\max_{S,T\subseteq V\colon\left|S\cup T\right|\leq k}\Var\left[\widehat{f}\left(T\vert S\right)\right]$
of the marginal gain estimates obtained from a single sample $\widehat{f}$}
\par\end{raggedright}
\begin{raggedright}
\textbf{\textcolor{blue}{Output:}}\textcolor{blue}{{} With probability
$1-\delta$, we obtain a solution $S\subseteq V$ with $\left|S\right|\le k$
such that $I^{(\tau)}\left(S\right)\ge\left(1-\frac{1}{e}-\epsilon\right)\opt$}
\par\end{raggedright}
\begin{raggedright}
$\optguess=n$
\par\end{raggedright}
\begin{raggedright}
\textbf{while} $\optguess\ge1$ \textbf{do}
\par\end{raggedright}
\begin{raggedright}
$\qquad$$\optguess\gets\frac{1}{2}\optguess$
\par\end{raggedright}
\begin{raggedright}
$\qquad$$S\gets\mathrm{GeneralModelGuess}\left(\epsilon,\frac{\delta}{\log n},\optguess\right)$
\par\end{raggedright}
\begin{raggedright}
$\qquad$$\left\{ X\left(S\mid\emptyset\right)\right\} \gets\mathrm{EstimateGainsVariance}\left(\left\{ S\right\} ,\emptyset,\epsilon\optguess,\frac{\delta}{\log n}\right)$
\par\end{raggedright}
\begin{raggedright}
$\qquad$\textbf{if} $X\left(S\mid\emptyset\right)\ge\optguess$ \textbf{then}
\par\end{raggedright}
\begin{raggedright}
$\qquad\qquad$\textbf{break}
\par\end{raggedright}
\begin{raggedright}
\textbf{return} $S$
\par\end{raggedright}
\begin{raggedright}
$\mathrm{GeneralModelGuess}\left(\epsilon,\delta',\optguess\right)$:
\par\end{raggedright}
\begin{raggedright}
$S\gets\emptyset$
\par\end{raggedright}
\begin{raggedright}
$\Tinner\gets\frac{2}{\epsilon}$
\par\end{raggedright}
\begin{raggedright}
$\Touter\gets\frac{8}{\epsilon}\log\left(\frac{n}{\delta'}\right)$
\par\end{raggedright}
\begin{raggedright}
$N\gets\Tinner\cdot\Touter$
\par\end{raggedright}
\begin{raggedright}
$\prest=\frac{\delta'}{N}$
\par\end{raggedright}
\begin{raggedright}
$\cest=\frac{\epsilon}{k}\optguess$
\par\end{raggedright}
\begin{raggedright}
\textbf{for }$r=1,\dots,\Touter$ \textbf{do}
\par\end{raggedright}
\begin{raggedright}
$\qquad$$\alpha\gets\frac{1}{k}\optguess\left(1-\epsilon\right)^{r}$
\par\end{raggedright}
\begin{raggedright}
$\qquad$$U\gets V$
\par\end{raggedright}
\begin{raggedright}
$\qquad$\textbf{for} $t=1,\dots,\Tinner$ \textbf{do}
\par\end{raggedright}
\begin{raggedright}
$\qquad\qquad$$\left\{ X\left(\left\{ u\right\} \mid S\right)\colon u\in U\right\} \gets\mathrm{EstimateGainsVariance}\left(\left\{ \left\{ u\right\} \colon u\in U\right\} ,S,\cest,\prest\right)$
\par\end{raggedright}
\begin{raggedright}
$\qquad\qquad$$U\gets\left\{ u\in U:X\left(\left\{ u\right\} \mid S\right)\ge\alpha\right\} $\textcolor{blue}{$\hfill$/$\negmedspace$/
filtering}
\par\end{raggedright}
\begin{raggedright}
$\qquad\qquad$\textbf{if} $U=\emptyset$ \textbf{then}
\par\end{raggedright}
\begin{raggedright}
\textbf{$\qquad\qquad\qquad$break }
\par\end{raggedright}
\begin{raggedright}
$\qquad\qquad$Let $v_{1},v_{2},\dots,v_{\left|U\right|}$ be a random
permutation of the elements in $U$
\par\end{raggedright}
\begin{raggedright}
$\qquad\qquad$$s\gets\min\left\{ k-\left|S\right|,\left|U\right|\right\} $
\par\end{raggedright}
\begin{raggedright}
$\qquad\qquad$Let $T_{i}=\left\{ v_{1},v_{2},\dots,v_{i}\right\} $
for all $1\le i\le s$\textcolor{blue}{$\hfill$/$\negmedspace$/
prefixes of random permutation}
\par\end{raggedright}
\begin{raggedright}
$\qquad\qquad$$\left\{ X\left(T_{i}\mid S\right)\colon1\le i\le s\right\} \gets\mathrm{EstimateGainsVariance}\left(\left\{ T_{i}:1\le i\le s\right\} ,S,\cest,\prest\right)$
\par\end{raggedright}
\begin{raggedright}
$\qquad\qquad$$i^{*}\gets\arg\max\left\{ 1\le i\le s:X\left(T_{i}\mid S\right)\ge\left(1-\epsilon\right)\left(\alpha-\cest\right)i\right\} $
\par\end{raggedright}
\begin{raggedright}
$\qquad\qquad$$S\gets S\cup T_{i^{*}}$\textcolor{blue}{$\hfill$/$\negmedspace$/
adding elements to the solution}
\par\end{raggedright}
\begin{raggedright}
$\qquad\qquad$\textbf{if} $\left|S\right|=k$ \textbf{then}
\par\end{raggedright}
\begin{raggedright}
$\qquad\qquad\qquad$\textbf{return} $S$
\par\end{raggedright}
\begin{raggedright}
$\qquad$\textbf{if} $U\not=\emptyset$ \textbf{then}
\par\end{raggedright}
\begin{raggedright}
$\qquad\qquad$\textbf{return} Failure
\par\end{raggedright}
\raggedright{}\textbf{return $S$}
\end{algorithm}
 
\begin{algorithm}
\begin{raggedright}
\caption{\label{alg:estimate-gains-variance} Estimation algorithm for the
marginal influence gain $I^{(\tau)}\left(T\mid S\right)$ for sets
$T\in{\cal T}$ with $T\subseteq V$ over a set $S\subseteq V$.}
\par\end{raggedright}
\begin{raggedright}
$\mathrm{EstimateGainsVariance}\left({\cal T},S,\cest,\prest,W\right):$
\par\end{raggedright}
\begin{raggedright}
\textbf{\textcolor{blue}{Input:}}\textcolor{blue}{{} Sets $T\subseteq V$
for $T\in{\cal T}$ and $S\subseteq V$, additive error $\cest>0$,
failure probability $\prest>0$, upper bound $W$ on the maximum variance
$\max_{S,T\subseteq V\colon\left|S\cup T\right|\leq k}\Var\left[\widehat{f}\left(T\vert S\right)\right]$
of the marginal gain estimates obtained from a single sample $\widehat{f}$}
\par\end{raggedright}
\begin{raggedright}
\textbf{\textcolor{blue}{Output:}}\textcolor{blue}{{} The algorithm
succeeds with probability $1-\prest$. If the algorithm succeeds,
it constructs estimates $I^{(\tau)}\left(T\mid S\right)$ such that
$\left|X\left(T\mid S\right)-I^{(\tau)}\left(T\mid S\right)\right|\le\cest$
for all $T\in{\cal T}$.}
\par\end{raggedright}
\begin{raggedright}
Let $n_{p}=48\log\left(\left|{\cal T}\right|/\prest\right)$ and $n_{s}=\frac{3}{\cest^{2}}W$
\par\end{raggedright}
\begin{raggedright}
\textbf{for }$i=1,\dots,n_{p}$ \textbf{do}
\par\end{raggedright}
\begin{raggedright}
\textbf{$\qquad$for} $j=1,\dots,n_{s}$ \textbf{do}
\par\end{raggedright}
\begin{raggedright}
$\qquad\qquad$Obtain a fresh sample $\widehat{f}\colon2^{V}\to\Reals$
\par\end{raggedright}
\begin{raggedright}
$\qquad\qquad$Let $X_{i,j}\left(T\mid S\right)\gets\widehat{f}\left(T\mid S\right)$
for all $T\in{\cal T}$
\par\end{raggedright}
\begin{raggedright}
$\qquad$Let $X_{i}\left(T\mid S\right)=\frac{1}{n_{s}}\sum_{j=1}^{n_{s}}X_{i,j}\left(T\mid S\right)$
for all $T\in{\cal T}$
\par\end{raggedright}
\begin{raggedright}
Let $X\left(T\mid S\right)$ be the median element among $\left\{ X_{i}\left(T\mid S\right)\right\} _{1\le i\le n_{p}}$
for all $T\in{\cal T}$
\par\end{raggedright}
\raggedright{}\textbf{return} $\left\{ X\left(T\mid S\right)\colon T\in{\cal T}\right\} $
\end{algorithm}

\begin{lem}
(Theorem \ref{thm:gtm-estimate-gains-variance}) Let $W$ be an upper
bound on the maximum variance
\[\max_{S,T\subseteq V:\left|S\cup T\right|\le k}\Var\left[\widehat{f}\left(T\mid S\right)\right].\]
Given the value $W$ as input as well as $\cest$ and a family of
sets ${\cal T}$ with $\left|{\cal T}\right|\le n$, Algorithm \ref{alg:estimate-gains-variance}
uses
$L=\widetilde{O}\left(\frac{1}{\cest^{2}}W\log\left(\frac{1}{\prest}\right)\right)$
samples, and it achieves additive error $\cest$, failure probability
$\prest$, and no multiplicative error.
\end{lem}
\begin{proof}
Fix a set $T\in{\cal T}$. For all $1\le i\le n_{\mathrm{p}}$, let
$Y_{i}\in\left\{ 0,1\right\} $ be an indicator random variable for
the event that $X_{i}\left(T\mid S\right)$ is a good estimate, i.e.
we set $Y_{i}=1$ if $\left|X_{i}\left(T\mid S\right)-I\left(T\mid S\right)\right|\le\cest$
and $Y_{i}=0$ otherwise. By Chebyshev's inequality and definition
of $n_{\mathrm{s}}$,
\[
\Pr\left[Y_{i}=0\right]=\Pr\left[\left|X_{i}\left(T\mid S\right)-I\left(T\mid S\right)\right|>\cest\right]\le\frac{W}{n_{\mathrm{s}}\cest^{2}}=\frac{1}{3}.
\]
Hence, $\E\left[\sum_{i=1}^{n_{\mathrm{p}}}Y_{i}\right]\ge\frac{2}{3}n_{\mathrm{p}}$.
If the median element is not a good estimate, at least half the estimates
are bad. We bound the probability of this event via a Chernoff bound
(Theorem \ref{thm:chernoff-lower}):
\begin{align*}
\Pr\left[\sum_{i=1}^{n_{\mathrm{p}}}Y_{i}\le\frac{n_{\mathrm{p}}}{2}\right] & =\Pr\left[\sum_{i=1}^{n_{\mathrm{p}}}Y_{i}\le\left(1-\frac{1}{4}\right)\frac{2}{3}n_{\mathrm{p}}\right]\\
 & \le\Pr\left[\sum_{i=1}^{n_{\mathrm{p}}}Y_{i}\le\left(1-\frac{1}{4}\right)\E\left[\sum_{i=1}^{n_{\mathrm{p}}}Y_{i}\right]\right]\\
 & \le\exp\left(-\E\left[\sum_{i=1}^{n_{\mathrm{p}}}Y_{i}\right]\frac{1}{2\cdot4^{2}}\right)\\
 & \le\exp\left(-\frac{n_{\mathrm{p}}}{48}\right)\\
 & \le\frac{\prest}{\left|{\cal T}\right|}
\end{align*}
by definition of $n_{\mathrm{p}}$. By a union bound over all sets
$T\in{\cal T}$, the algorithm achieves error $\prest$ on all sets
in ${\cal T}$. The number of samples is $L=n_{\mathrm{p}}\cdot n_{\mathrm{s}}=\widetilde{O}\left(\frac{W}{\cest^{2}}\log\left(\frac{1}{\prest}\right)\right)$.
\end{proof}

\begin{thm}
(Theorem \ref{thm:threshold-variance}) Let $W$ be an upper bound
on the maximum variance
\[\max_{S,T\subseteq V\colon\left|S\cup T\right|\leq k}\Var\left[\widehat{f}\left(T\vert S\right)\right]\]
of the marginal gain estimates obtained from a single sample $\widehat{f}$.
Given the value $W$ as input as well as $\epsilon$ and $\delta$,
Algorithm \ref{alg:general-model-variance} uses $\widetilde{O}\left(NL\right)$
samples where
\[L=\widetilde{O}\left(\frac{k^{2}}{\epsilon^{2}\opt^{2}}W\log\left(\frac{1}{\delta\epsilon}\log\left(\frac{1}{\delta}\right)\right)\right)\]
and it returns a solution $S$ satisfying $I^{(\tau)}\left(S\right)\geq\left(1-\frac{1}{e}-\epsilon\right)\opt$
with probability $1-\delta$.
\end{thm}
\begin{proof}
The proof follows by analyzing the outer search in $\mathrm{GeneralModel}$
which selects guesses $\optguess$ for $\opt$, in conjunction with
Theorem \ref{thm:template-ideal}. We first show that there are at
most $2\log n$ iterations of the search: If $i$ counts the search
iterations, then the current guess is $\optguess=2^{-i}n$. We exit
the loop when $1\ge\optguess=2^{-i}n$ which implies that $i\le2\log n$.
By Lemma \ref{lem:estimate-gains-variance}, any call Algorithm \ref{alg:estimate-gains-variance}
fails to be $\cest$-additively correct with probability at most $\frac{\delta}{\log n}$.
By a union bound over all iterations, we obtain that all calls are
$\cest$-additively correct with probability at least $1-\delta$.
For the remainder of the proof, let us thus condition on the event
that these calls are $\cest$-additive correct. By Theorem \ref{thm:template-ideal},
any call to $\mathrm{GeneralModelGuess}$ fails with probability at
most $\delta'=N\prest=\frac{\delta}{\log n}$. By a union bound over
all $2\log n$ iterations of the outer search, the overall algorithm
fails with probability at most $O(\delta)$. For the remainder of
the proof, let us also condition on the event that no call to $\mathrm{GeneralModelGuess}$
fails. 

In order to bound the number of samples, we need to ensure that $\optguess$
does not become too small compared to $\opt$. Let us thus assume
that we are in an iteration where $\optguess\le\frac{1}{4}\opt$.
In this case, we set $\cest=\frac{\epsilon}{k}\optguess\le\frac{\epsilon}{4k}\opt$
and our algorithm guarantees $I^{(\tau)}\left(S\right)\ge\left(1-\frac{1}{e}-\frac{\epsilon}{2}\right)\opt-2k\cest\ge\left(1-\frac{1}{e}-\epsilon\right)\opt$.
Our estimation routine guarantees that $X\left(S\mid\emptyset\right)\ge I^{(\tau)}\left(S\mid\emptyset\right)-\cest\ge\frac{1}{4}\opt\ge\optguess$.
By design of the algorithm, we will then break the loop. In particular,
this ensures that $\optguess\ge\frac{1}{8}\opt$ and thus that $\cest=\frac{\epsilon}{k}\optguess\ge\frac{\epsilon}{8k}\opt$
for all calls to Algorithm \ref{alg:estimate-gains-variance}. 

By Lemma \ref{lem:estimate-gains-variance}, each call to Algorithm
\ref{alg:estimate-gains-variance} requires
\[L=\widetilde{O}\left(\frac{1}{\cest^{2}}W\log\left(\frac{1}{\prest}\right)\right)=\widetilde{O}\left(\frac{k^{2}}{\epsilon^{2}\opt^{2}}W\log\left(\frac{1}{\delta\epsilon}\log\left(\frac{1}{\delta}\right)\right)\right)\]
samples and we make a total of $N=\widetilde{O}\left(\frac{1}{\epsilon^{2}}\log\left(\frac{1}{\delta}\right)\right)$
calls. In total, we require $\widetilde{O}\left(NL\right)$ samples. 
\end{proof}

\subsection{Our algorithm using an estimator based on the empirical variance}

\begin{algorithm}
\begin{raggedright}
\caption{\label{alg:general-model-empirical-variance} Our algorithm for general
models that instantiates the template algorithm \ref{alg:optimization}
using the estimation algorithm \ref{alg:estimate-gains-empirical-variance}.}
\par\end{raggedright}
\begin{raggedright}
$\mathrm{GeneralModelEmpiricalVariance}\left(\epsilon,\delta\right):$
\par\end{raggedright}
\begin{raggedright}
\textbf{\textcolor{blue}{Input:}}\textcolor{blue}{{} Error parameters
$\epsilon$ and $\delta$}
\par\end{raggedright}
\begin{raggedright}
\textbf{\textcolor{blue}{Output:}}\textcolor{blue}{{} With probability
$1-\delta$, we obtain a solution $S\subseteq V$ with $\left|S\right|\le k$
such that $I^{(\tau)}\left(S\right)\ge\left(1-\frac{1}{e}-\epsilon\right)\opt$.}
\par\end{raggedright}
\begin{raggedright}
$\optguess=n$
\par\end{raggedright}
\begin{raggedright}
\textbf{while} $\optguess\ge1$ \textbf{do}
\par\end{raggedright}
\begin{raggedright}
$\qquad$$\optguess\gets\frac{1}{2}\optguess$
\par\end{raggedright}
\begin{raggedright}
$\qquad$$S\gets\mathrm{GeneralModelEmpiricalVarianceGuess}\left(\epsilon,\frac{\delta}{\log n},\optguess\right)$
\par\end{raggedright}
\begin{raggedright}
$\qquad$$\left\{ X\left(S\mid\emptyset\right)\right\} \gets\mathrm{EstimateGainsVariance}\left(\left\{ S\right\} ,\emptyset,\epsilon\optguess,\frac{\delta}{\log n}\right)$
\par\end{raggedright}
\begin{raggedright}
$\qquad$\textbf{if} $X\left(S\mid\emptyset\right)\ge\optguess$ \textbf{then}
\par\end{raggedright}
\begin{raggedright}
$\qquad\qquad$\textbf{break}
\par\end{raggedright}
\begin{raggedright}
\textbf{return} $S$
\par\end{raggedright}
\begin{raggedright}
$\mathrm{GeneralModelEmpiricalVarianceGuess}\left(\epsilon,\delta',\optguess\right)$:
\par\end{raggedright}
\begin{raggedright}
$S\gets\emptyset$
\par\end{raggedright}
\begin{raggedright}
$\Tinner\gets\frac{2}{\epsilon}$
\par\end{raggedright}
\begin{raggedright}
$\Touter\gets\frac{8}{\epsilon}\log\left(\frac{n}{\delta'}\right)$
\par\end{raggedright}
\begin{raggedright}
$N\gets\Tinner\cdot\Touter$
\par\end{raggedright}
\begin{raggedright}
$\prest=\frac{\delta'}{N}$
\par\end{raggedright}
\begin{raggedright}
$\cest=\frac{\epsilon}{k}\optguess$
\par\end{raggedright}
\begin{raggedright}
\textbf{for }$r=1,\dots,\Touter$ \textbf{do}
\par\end{raggedright}
\begin{raggedright}
$\qquad$$\alpha\gets\frac{1}{k}\optguess\left(1-\epsilon\right)^{r}$
\par\end{raggedright}
\begin{raggedright}
$\qquad$$U\gets V$
\par\end{raggedright}
\begin{raggedright}
$\qquad$\textbf{for} $t=1,\dots,\Tinner$ \textbf{do}
\par\end{raggedright}
\begin{raggedright}
$\qquad\qquad$$\left\{ X\left(\left\{ u\right\} \mid S\right)\colon u\in U\right\} \gets\mathrm{EstimateGainsEmpiricalVariance}\left(\left\{ \left\{ u\right\} \colon u\in U\right\} ,S,\cest,\prest\right)$
\par\end{raggedright}
\begin{raggedright}
$\qquad\qquad$$U\gets\left\{ u\in U:X\left(\left\{ u\right\} \mid S\right)\ge\alpha\right\} $\textcolor{blue}{$\hfill$/$\negmedspace$/
filtering}
\par\end{raggedright}
\begin{raggedright}
$\qquad\qquad$\textbf{if} $U=\emptyset$ \textbf{then}
\par\end{raggedright}
\begin{raggedright}
\textbf{$\qquad\qquad\qquad$break }
\par\end{raggedright}
\begin{raggedright}
$\qquad\qquad$Let $v_{1},v_{2},\dots,v_{\left|U\right|}$ be a random
permutation of the elements in $U$
\par\end{raggedright}
\begin{raggedright}
$\qquad\qquad$$s\gets\min\left\{ k-\left|S\right|,\left|U\right|\right\} $
\par\end{raggedright}
\begin{raggedright}
$\qquad\qquad$Let $T_{i}=\left\{ v_{1},v_{2},\dots,v_{i}\right\} $
for all $1\le i\le s$\textcolor{blue}{$\hfill$/$\negmedspace$/
prefixes of random permutation}
\par\end{raggedright}
\begin{raggedright}
$\qquad\qquad$$\left\{ X\left(T_{i}\mid S\right)\colon1\le i\le s\right\} \gets\mathrm{EstimateGainsEmpiricalVariance}\left(\left\{ T_{i}:1\le i\le s\right\} ,S,\cest,\prest\right)$
\par\end{raggedright}
\begin{raggedright}
$\qquad\qquad$$i^{*}\gets\arg\max\left\{ 1\le i\le s:X\left(T_{i}\mid S\right)\ge\left(1-\epsilon\right)\left(\alpha-\cest\right)i\right\} $
\par\end{raggedright}
\begin{raggedright}
$\qquad\qquad$$S\gets S\cup T_{i^{*}}$\textcolor{blue}{$\hfill$/$\negmedspace$/
adding elements to the solution}
\par\end{raggedright}
\begin{raggedright}
$\qquad\qquad$\textbf{if} $\left|S\right|=k$ \textbf{then}
\par\end{raggedright}
\begin{raggedright}
$\qquad\qquad\qquad$\textbf{return} $S$
\par\end{raggedright}
\begin{raggedright}
$\qquad$\textbf{if} $U\not=\emptyset$ \textbf{then}
\par\end{raggedright}
\begin{raggedright}
$\qquad\qquad$\textbf{return} Failure
\par\end{raggedright}
\raggedright{}\textbf{return $S$}
\end{algorithm}
 
\begin{algorithm}
\begin{raggedright}
\caption{\label{alg:estimate-gains-empirical-variance} Estimation algorithm
for the marginal influence gain $I^{(\tau)}\left(T\mid S\right)$
for sets $T\in{\cal T}$ with $T\subseteq V$ over a set $S\subseteq V$.}
\par\end{raggedright}
\begin{raggedright}
$\mathrm{EstimateGainsEmpiricalVariance}\left({\cal T},S,\cest,\prest\right):$
\par\end{raggedright}
\begin{raggedright}
\textbf{\textcolor{blue}{Input:}}\textcolor{blue}{{} Sets $T\subseteq V$
for $T\in{\cal T}$ and $S\subseteq V$, additive error $\cest>0$,
failure probability $\prest>0$.}
\par\end{raggedright}
\begin{raggedright}
\textbf{\textcolor{blue}{Output:}}\textcolor{blue}{{} The algorithm
succeeds with probability $1-\prest$. If the algorithm succeeds,
it constructs estimates $X\left(T\mid S\right)$ such that $\left|X\left(T\mid S\right)-I^{(\tau)}\left(T\mid S\right)\right|\le\cest$
for all $T\in{\cal T}$.}
\par\end{raggedright}
\begin{raggedright}
$L_{\max}\gets\frac{90}{\cest^{2}}\left(n^{2}+n\cest\right)\log\left(\frac{n}{\prest\cest}\right)$
\par\end{raggedright}
\begin{raggedright}
$\widetilde{\delta}\gets\frac{\prest}{2nL_{\max}}$
\par\end{raggedright}
\begin{raggedright}
\textbf{for} $\ell=1,\dots,L_{\max}$ \textbf{do}
\par\end{raggedright}
\begin{raggedright}
$\qquad$Obtain a fresh sample $\widehat{f}\colon2^{V}\to\Reals$
\par\end{raggedright}
\begin{raggedright}
$\qquad$\textbf{for} $T\in{\cal T}$ \textbf{do}
\par\end{raggedright}
\begin{raggedright}
$\qquad\qquad$$X_{\ell}\left(T\mid S\right)\gets\widehat{f}\left(T\mid S\right)$
$\hfill$\textcolor{blue}{/$\negmedspace$/ evaluate sample}
\par\end{raggedright}
\begin{raggedright}
$\qquad\qquad$$\widetilde{V}\left(T\right)\gets\frac{1}{\ell\left(\ell-1\right)}\sum_{1\le i<j\le\ell}\left(X_{i}\left(T\mid S\right)-X_{j}\left(T\mid S\right)\right)^{2}$
\par\end{raggedright}
\begin{raggedright}
$\qquad\qquad$$C\left(T\right)\gets\sqrt{\frac{2\widetilde{V}\left(T\right)\log\left(4/\widetilde{\delta}\right)}{\ell}}+n\frac{7\log\left(4/\widetilde{\delta}\right)}{3\left(\ell-1\right)}$
\par\end{raggedright}
\begin{raggedright}
\textcolor{blue}{$\hfill$/$\negmedspace$/ compute observed confidence
interval}
\par\end{raggedright}
\begin{raggedright}
$\qquad$\textbf{if} $C\left(T\right)\le c$ for all $T\in{\cal T}$
\textbf{then}
\par\end{raggedright}
\begin{raggedright}
$\qquad\qquad$Let $X\left(T\mid S\right)=\frac{1}{\ell}\sum_{i=1}^{\ell}X_{i}\left(T\mid S\right)$
for each $T\in{\cal T}$
\par\end{raggedright}
\begin{raggedright}
\textcolor{blue}{$\hfill$/$\negmedspace$/ compute estimate}
\par\end{raggedright}
\begin{raggedright}
$\qquad\qquad$\textbf{return} $\left\{ X\left(T\mid S\right)\colon T\in{\cal T}\right\} $
\par\end{raggedright}
\raggedright{}\textbf{return }Failure
\end{algorithm}

\begin{lem}
\label{lem:emp-var-num} Let
\[L^{*}=\frac{18}{\cest^{2}}\left(\Var\left[\widehat{f}\left(T\mid S\right)\right]+n\cest\right)\log\left(\frac{1}{\widetilde{\delta}}\right)\]
and let
$\widetilde{V}_{L^{*}}=\frac{1}{L^{*}(L^{*}-1)}\sum_{i\le i<j\le L^{*}}(X_{i}-X_{j})^{2}$
where $X_{1},\dots,X_{L^{*}}$ are random independent samples
$X_{i}=\hat{f}(T\mid S)$ for all
$i\in\{1,\dots,L^{*}\}$ for sets $T,S\subseteq V$.
With probability $1-\widetilde{\delta}$,
\[
\sqrt{\frac{2\widetilde{V}_{L^{*}}\ln\left(4/\widetilde{\delta}\right)}{L^{*}}}+n\frac{7\ln\left(4/\widetilde{\delta}\right)}{3\left(L^{*}-1\right)}\le\cest.
\]
\end{lem}
\begin{proof}
We first decompose
\begin{align*}
&\sqrt{\frac{2\widetilde{V}_{L^{*}}\log(4/\widetilde{\delta})}{L^{*}}}+n\frac{7\log(4/\widetilde{\delta})}{3(L^{*}-1)}\\
&\quad \le\sqrt{\frac{2\bigl(\Var[\widehat{f}(T\mid S)]+(\widetilde{V}_{L^{*}}-\Var[\widehat{f}(T\mid S)])\bigr)\log(4/\widetilde{\delta})}{L^{*}}}+n\frac{7\log(4/\widetilde{\delta})}{3(L^{*}-1)}\\
&\quad \le\sqrt{\frac{2}{L^{*}}\Var\bigl[\widehat{f}(T\mid S)\bigr]\log(4/\widetilde{\delta})}
+\sqrt{\frac{2}{L^{*}}\bigl(\widetilde{V}_{L^{*}}-\Var\bigl[\widehat{f}(T\mid S)\bigr]\bigr)\log(4/\widetilde{\delta})}\\
&\qquad+n\frac{7\log(4/\widetilde{\delta})}{3(L^{*}-1)}
\end{align*}
and we will show that for the given $L^{*}$, each of the three terms
are bounded by $\frac{1}{3}\cest$ with probability $1-\widetilde{\delta}$. 

We use the choice of $L^{*}$ to bound the first term
\[
\frac{2\Var\left[\widehat{f}\left(T\mid S\right)\right]\log\left(4/\widetilde{\delta}\right)}{L^{*}}\le\frac{1}{9}\cest^{2}
\]
and therefore have $\sqrt{\frac{2}{L^{*}}\Var\left[\widehat{f}\left(T\mid S\right)\right]\log\left(4/\widetilde{\delta}\right)}\le\frac{1}{3}\cest$.

The second term is the only one that is probabilistic. We rearrange
\[
\frac{2}{L^{*}}\left(\widetilde{V}_{L^{*}}-\mathrm{Var}\left[\widehat{f}\left(T\mid S\right)\right]\right)\log\left(4/\widetilde{\delta}\right)\le\cest^{2}\iff\widetilde{V}_{L^{*}}-\mathrm{Var}\left[\widehat{f}\left(T\mid S\right)\right]\le\frac{L^{*}\cest^{2}}{2\log\left(4/\widetilde{\delta}\right)}
\]
We now apply the bound \citep{maurer09}
\[
\Pr\left[\widetilde{V}_{L^{*}}-\Var\left[\widehat{f}\left(T\mid S\right)\right]>\lambda\right]\le\exp\left(-\frac{\left(L^{*}-1\right)\lambda^{2}}{2\Var\left[\widehat{f}\left(T\mid S\right)\right]+\lambda}\right)
\]
to the empirical variance. This yields a failure probability of at
most
\begin{align*}
&\Pr\biggl[\widetilde{V}_{L^{*}}-\Var\bigl[\widehat{f}(T\mid S)\bigr]\ge\frac{L^{*}\cest^{2}}{2\log(4/\widetilde{\delta})}\biggr]\\
&\quad \le\exp\biggl(-\frac{(L^{*}-1)(L^{*})^{2}\cest^{4}}{8\Var[\widehat{f}(T\mid S)]\log(4/\widetilde{\delta})^{2}+2\log(4/\widetilde{\delta})L^{*}\cest^{2}}\biggr)\\
&\quad \le\exp\bigl(-\tfrac{1}{\cest}n^{2}\log(4/\widetilde{\delta})\bigr)
 \le\widetilde{\delta}
\end{align*}
which also follows by our choice of $L^{*}$.

For the third term, we also obtain by the choice of $L^{*}$ that
\[
n\frac{7\log\left(4/\widetilde{\delta}\right)}{3\left(L^{*}-1\right)}\le\frac{1}{3}\cest
\]

Thus overall,
\[
\sqrt{\frac{2\widetilde{V}_{L^{*}}\log\left(4/\widetilde{\delta}\right)}{L^{*}}}+n\frac{7\log\left(4/\widetilde{\delta}\right)}{3\left(L^{*}-1\right)}\le\cest
\]
with probability $1-\widetilde{\delta}$.
\end{proof}

\begin{lem}
\label{lem:emp-var-L} With probability $1-\frac{1}{2}\prest$, Algorithm
\ref{alg:estimate-gains-empirical-variance} stops at iteration $L^{*}$
and does not declare failure. 
\end{lem}
\begin{proof}
We show that $L^{*}\le L_{\max}$ and apply Lemma \ref{lem:emp-var-num}:
First, note that
\[
\left(\frac{n}{\prest\cest}\right)^{5}\ge\frac{2n}{\prest}\frac{90}{\cest^{2}}\left(n^{2}+n\cest\right)\log\left(\frac{n}{\prest\cest}\right)=\frac{2nL_{\max}}{\prest}=\frac{1}{\widetilde{\delta}}
\]
for sufficiently large values of $n$. Furthermore, we have $n^{2}\ge W\ge\Var\left[\widehat{f}\left(T\mid S\right)\right]$
for all $T,S\subseteq V$. Combining both, we obtain
\begin{align*}
L_{\max} & =\frac{90}{\cest^{2}}\left(n^{2}+n\cest\right)\log\left(\frac{n}{\prest\cest}\right)\\
 & \ge\frac{18}{\cest^{2}}\left(n^{2}+n\cest\right)\log\left(\left(\frac{n}{\prest\cest}\right)^{5}\right)\\
 & \ge\frac{18}{\cest^{2}}\left(n^{2}+n\cest\right)\log\left(\frac{1}{\widetilde{\delta}}\right)\\
 & \ge\frac{18}{\cest^{2}}\left(\Var\left[\widehat{f}\left(T\mid S\right)\right]+n\cest\right)\log\left(\frac{1}{\widetilde{\delta}}\right)=L^{*}
\end{align*}
Thus, by Lemma \ref{lem:emp-var-num} we do not break the loop at
iteration $L^{*}\le L_{\max}$ with probability at most $n\widetilde{\delta}$
which is $\le\frac{1}{2}\prest$.
\end{proof}

\begin{lem}
\label{lem:emp-var-conf-failure} Conditioned on the event Algorithm
\ref{alg:estimate-gains-empirical-variance} does not declare failure,
the estimates satisfy
\[\left|X\left(T\mid S\right)-I^{(\tau)}\left(T\mid S\right)\right|\le\cest\]
for all $T\in{\cal T}$ with probability at least $1-\frac{1}{2}\prest$. 
\end{lem}
\begin{proof}
If we do not declare failure, we exit the loop before reaching $\ell=L_{\max}$.
We use to Theorem \ref{thm:empirical-bernstein} which states that
with probability $1-\widetilde{\delta}$,
\begin{equation}
\left|\frac{1}{\ell}\sum_{i=1}^{\ell}X_{a,\ell}-\E\left[\frac{1}{\ell}\sum_{i=1}^{\ell}X_{a,\ell}\right]\right|\le\sqrt{\frac{2W_{a,\ell}\log\left(4/\widetilde{\delta}\right)}{\ell}}+n\frac{7\log\left(4/\widetilde{\delta}\right)}{3\left(\ell-1\right)}.\label{eq:2}
\end{equation}
By a union bound, this bound holds for all $1\le\ell\le L$ and $1\le a\le n$
with probability $1-\widetilde{\delta}nL_{\max}=1-\frac{1}{2}\delta_{2}$.
Conditioned on this event, if we break the loop early for $\ell<L$,
the RHS of (\ref{eq:2}) is $\le\cest$ which in turn implies that
\[
\left|\frac{1}{\ell}\sum_{i=1}^{\ell}X_{a,\ell}-\E\left[\frac{1}{\ell}\sum_{i=1}^{\ell}X_{a,\ell}\right]\right|\le\cest.
\]
\end{proof}

\begin{thm}
(Theorem \ref{lem:emp-var-conf}) Given $\cest$ and $\prest$ as
input, with probability $1-\prest$, Algorithm \ref{alg:estimate-gains-empirical-variance}
uses $L=\widetilde{O}\left(\left(\frac{W}{\cest^{2}}+\frac{n}{\cest}\right)\log\left(\frac{1}{\prest}\right)\right)$
samples, where $W=\max_{T\in{\cal T}}\Var\left[\widehat{f}\left(T\mid S\right)\right]$,
and it achieves additive error $\cest$, failure probability $\prest$,
and no multiplicative error ($\gamma=0$).
\end{thm}
\begin{proof}
By Lemma \ref{lem:emp-var-L}, we break the outer for loop at iteration
$L^{*}$ with probability at least $1-\frac{1}{2}\prest$. Let us
condition on this event. We take a single sample per iteration of
the outer for loop, which makes for a total of $L=L^{*}=\widetilde{O}\left(\left(\frac{W}{\cest^{2}}+\frac{n}{\cest}\right)\log\left(\frac{1}{\prest}\right)\right)$
samples. Since we break the outer for loop at iteration $L^{*}<L_{\max}$,
we do not declare failure. Therefore, by Lemma \ref{lem:emp-var-conf-failure},
the estimates have additive error at most $\cest$ probability at
least $1-\frac{1}{2}\prest$. Overall, the estimates provided by Algorithm
\ref{alg:estimate-gains-empirical-variance} thus have additive error
at most $\cest$ with probability at least $1-\prest$. 
\end{proof}

\begin{thm}
(Theorem \ref{thm:threshold-empirical-variance}) Given $\epsilon$
and $\delta$ as input, with probability $1-\delta$, Algorithm \ref{alg:general-model-empirical-variance}
uses $\widetilde{O}\left(NL\right)$ samples for
\[L=\widetilde{O}\left(\left(\frac{k^{2}W}{\epsilon^{2}\opt^{2}}+\frac{kn}{\epsilon\opt}\right)\log\left(\frac{N}{\delta}\right)\right)\]
and it returns a solution $S$ satisfying $I^{(\tau)}\left(S\right)\geq\left(1-\frac{1}{e}-\epsilon\right)\opt$.
\end{thm}
\begin{proof}
The proof is identical to the proof of Theorem \ref{thm:threshold-variance}.
\end{proof}

\section{\label{sec:appendix-oc} Influence maximization on observed cascades}

In this section, we provide a description and pseudocode for our estimation
subroutines in the observed cascades model. We first showcase our
sketching approach for the setting where the influence is unrestricted,
i.e. we run the diffusion process to completion. We provide the pseudocode
for the estimation and prove its estimation guarantee Theorem \ref{thm:oc-estimation-unrestricted}.
We then showcase our sketching approach for the setting where the
influence is restricted to $\tau$ steps. Again, we provide the pseudocode
and estimation guarantee Theorem \ref{thm:oc-estimation-restricted}.
In the final subsection, we provide the pseudocode for our algorithm
for influence maximization on observed cascades which uses the estimation
routines. We also prove Theorem \ref{thm:observed-cascades-main}.

\global\long\def\bottom{\mathrm{bot}_{b}}%

For a finite set $Y\subseteq\Reals$, we define $\bottom\left(Y\right)$
as the set of the $b$ smallest values in $Y$, and let $\max\left(Y\right)$
denote the largest value in $Y$. To keep our presentation self-contained,
we first provide a general guarantee for bottom-$b$ sketches over
arbitrary sets $R\subseteq V$. Similar theorems can be found in the
literature \citep{cohen97}. 
\begin{lem}
\label{lem:min-hash} Given a subset $R\subseteq V$, let $y_{v}\sim\mathrm{Uniform}\left(\left[0,1\right]\right)$
for all $v\in R$ and $Y=\bottom\left(\left\{ y_{v}:v\in R\right\} \right)$
for $b=\frac{6}{\epsest^{2}}\log\left(\frac{2}{\widetilde{\delta}}\right)$.
We define the estimate $X=\left|Y\right|$ if $\left|Y\right|<b$
and $X=\frac{b}{\max\left(Y\right)}-1$ if $\left|Y\right|=b$. Then,
$X$ has less than $\epsest$-multiplicative error, i.e. $\left(1-\epsest\right)\left|R\right|\le X\le\left(1+\epsest\right)\left|R\right|$
with probability $1-\widetilde{\delta}$. 

\begin{proof}
If $\left|Y\right|<b$, we have $\left|Y\right|=\left|R\right|$ and
the estimate is exact. We thus focus on the case when $\left|Y\right|=b$
and $X=\frac{b}{\max Y}-1$. We consider the two failure events $X>\left(1+\epsest\right)\left|R\right|$
and $X<\left(1-\epsest\right)\left|R\right|$ for any $1\ge\epsest>0$.
Let us first bound the probability of the first failure event. By
definition of the estimate, this event is equivalent to
\[
\frac{b}{\max Y}-1=X>\left(1+\epsest\right)\left|R\right|\iff\max Y<\frac{b}{\left(1+\epsest\right)\left|R\right|+1}\eqqcolon\zeta^{+}
\]
which in turn is equivalent to the event that there are at least $b$
elements $v\in R$ such that $y_{v}<\zeta^{+}$. We bound the probability
of this event with a Chernoff bound (Theorem \ref{thm:chernoff-upper})
for which we define the random variables $Z_{v}=1$ if $y_{v}<\zeta^{+}$
and $Z_{v}=0$ otherwise for all $v\in R$. Since $y_{v}\sim\mathcal{U}\left(\left[0,1\right]\right)$
we have $\E\left[\sum_{v\in R}Z_{v}\right]=\left|R\right|\zeta^{+}$
and thus
\begin{align*}
\Pr\left[\sum_{v\in R}Z_{v}\ge b\right] & =\Pr\left[\sum_{v\in R}Z_{v}\ge\left(1+\frac{\epsest\left|R\right|+1}{\left|R\right|}\right)\left|R\right|\zeta^{+}\right]\\
 & \le\exp\left(-\frac{\zeta^{+}\left(\epsest\left|R\right|+1\right)^{2}}{2\left|R\right|+\epsest\left|R\right|+1}\right)\\
 & =\exp\left(-b\frac{\left(\epsest\left|R\right|+1\right)^{2}}{\left(\left(1+\epsest\right)\left|R\right|+1\right)\left(2\left|R\right|+\epsest\left|R\right|+1\right)}\right)\\
 & \le\exp\left(-\frac{1}{6}b\epsest^{2}\right)
\end{align*}
for $b=\frac{6}{\epsest^{2}}\log\left(\frac{2}{\widetilde{\delta}}\right)$. 

We proceed analogously for the second failure event. Again, by definition
of the estimate,

\[
\frac{b}{y^{(b)}}-1=X<\left(1-\epsest\right)\left|R\right|\iff y^{(b)}>\frac{b}{\left(1-\epsest\right)\left|R\right|+1}\eqqcolon\zeta^{-}
\]
Let $Z_{v}=1$ if $y_{v}<\zeta^{-}$ and $0$ otherwise for all $v\in R$.
Again, $\E\left[\sum_{v\in R}Z_{v}\right]=\left|R\right|\zeta^{-}$.
We use the same choice for $b$, which implies that $\left|R\right|\ge b\ge\frac{6}{\epsest^{2}}$
and thus $\epsest\left|R\right|\ge6$. Therefore, by a Chernoff bound
(Theorem \ref{thm:chernoff-lower}),
\begin{align*}
\Pr\left[\sum_{v\in R}Z_{v}<b\right] & =\Pr\left[\sum_{v\in R}Z_{v}<\left(1-\frac{\epsest\left|R\right|-1}{\left|R\right|}\right)\left|R\right|\zeta^{-}\right]\\
 & \le\exp\left(-\frac{\zeta^{-}\left(\epsest\left|R\right|-1\right)^{2}}{2\left|R\right|}\right)\\
 & \le\exp\left(-b\frac{\left(\epsest\left|R\right|-1\right)^{2}}{2\left|R\right|\left(\left(1-\epsest\right)\left|R\right|+1\right)}\right)\\
 & \le\exp\left(-\frac{1}{6}b\epsest^{2}\right)\\
 & \le\frac{\widetilde{\delta}}{2}.
\end{align*}
for the same choice of $b$.
\end{proof}
\end{lem}

\subsection{Our estimation routine for the unrestricted influence}

\begin{algorithm}
\begin{raggedright}
\caption{\label{alg:estimate-gains-sketch-unrestricted} Sketching algorithm
for estimating the unrestricted influence}
\par\end{raggedright}
\begin{raggedright}
$\mathrm{EstimateGainsSketch}\left({\cal T},S,\epsest,\prest\right)$:
\par\end{raggedright}
\begin{raggedright}
\textbf{\textcolor{blue}{Input:}}\textcolor{blue}{{} Sets $T\subseteq V$
for all $T\in{\cal T}$ and $S\subseteq V$, multiplicative error
$\epsest$, failure probability $\prest$.}
\par\end{raggedright}
\begin{raggedright}
\textbf{\textcolor{blue}{Output:}}\textcolor{blue}{{} With probability
$1-\prest$, the algorithm constructs estimates $X\left(T\mid S\right)$
such that $\left(1-\epsest\right)I^{(n)}\left(T\mid S\right)\le X\left(T\mid S\right)\le\left(1+\epsest\right)I^{(n)}\left(T\mid S\right)$
for all $T\in{\cal T}$.}
\par\end{raggedright}
\begin{raggedright}
Define $\bottom\left(Y\right)$ as the set of the $b$ smallest values
in a finite set $Y\subseteq\Reals$
\par\end{raggedright}
\begin{raggedright}
Define $\max\left(Y\right)$ as the largest value of a finite set
$Y\subseteq\Reals$
\par\end{raggedright}
\begin{raggedright}
Define $N_{G}^{\mathrm{out}}\left(v\right)=\left\{ v:(u,v)\in E(G)\right\} $
as the set of outgoing neighbors of $u$ in a graph $G$
\par\end{raggedright}
\begin{raggedright}
Let $b=\frac{6}{\epsest^{2}}\log\left(\frac{2\ell n}{\prest}\right)$
\par\end{raggedright}
\begin{raggedright}
\textbf{for }$i=1,\dots,\ell$ \textbf{do}
\par\end{raggedright}
\begin{raggedright}
$\qquad$Let ${\cal C}=\left\{ C_{1},C_{2},\dots\right\} $ be the
strongly connected components in $G_{i}$
\par\end{raggedright}
\begin{raggedright}
$\qquad$Contract the strongly connected components in $G_{i}$ and
let $G'_{i}$ be the resulting directed acyclic graph
\par\end{raggedright}
\begin{raggedright}
$\qquad$\textbf{for} $C\in{\cal C}$ in reverse topological order
\textbf{do}
\par\end{raggedright}
\begin{raggedright}
\textcolor{blue}{$\hfill$/$\negmedspace$/ building bottom-$b$
sketches $Y_{i,C}$ for the sets $C\setminus\Reach_{G_{i}}^{(n)}\left(S\right)$
for all components $C\in{\cal C}$}
\par\end{raggedright}
\begin{raggedright}
$\qquad\qquad$\textbf{if} $C\subseteq\Reach_{G_{i}}^{(n)}\left(S\right)$
\textbf{then}
\par\end{raggedright}
\begin{raggedright}
$\qquad\qquad\qquad$$Y_{i,C}\gets\emptyset$\textcolor{blue}{$\hfill$/$\negmedspace$/
ignore the contribution of vertices in $\Reach_{G_{i}}^{(n)}\left(S\right)$}
\par\end{raggedright}
\begin{raggedright}
$\qquad\qquad$\textbf{else}
\par\end{raggedright}
\begin{raggedright}
$\qquad\qquad\qquad$Sample $y_{i,v}\sim U\left(\left[0,1\right]\right)$
for each $v\in C$
\par\end{raggedright}
\begin{raggedright}
$\qquad\qquad\qquad$$Y_{i,C}\gets\bottom\left(\left\{ y_{i,v}:v\in C\right\} \right)$
\par\end{raggedright}
\begin{raggedright}
\textcolor{blue}{$\hfill$/$\negmedspace$/ building bottom-$b$ sketches
$R_{i,C}$ for the sets $\Reach_{G_{i}}^{(n)}\left(C\right)\setminus\Reach_{G_{i}}^{(n)}\left(S\right)$
for all components $C\in{\cal C}$}
\par\end{raggedright}
\begin{raggedright}
$\qquad$\textbf{for} $C\in{\cal C}$ in reverse topological order
\textbf{do}
\par\end{raggedright}
\begin{raggedright}
$\qquad\qquad$$R_{i,C}\gets Y_{i,C}$
\par\end{raggedright}
\begin{raggedright}
$\qquad\qquad$\textbf{for} all $C'\in N_{G'_{i}}^{\mathrm{out}}\left(C\right)$
\textbf{do}
\par\end{raggedright}
\begin{raggedright}
$\qquad\qquad\qquad$$R_{i,C}\gets\bottom\left(R_{i,C}\cup R_{i,C'}\right)$
\par\end{raggedright}
\begin{raggedright}
\textcolor{blue}{$\hfill$/$\negmedspace$/ building bottom-$b$ sketches
$R_{i,T}$ for the sets $\Reach_{G_{i}}^{(\tau)}\left(T\right)\setminus\Reach_{G_{i}}^{(\tau)}\left(S\right)$
for all $T\in{\cal T}$}
\par\end{raggedright}
\begin{raggedright}
$\qquad$\textbf{if} ${\cal T}$ is a sequence of prefixes $T_{j}=\left\{ v_{1},v_{2},\dots,v_{j}\right\} $
for all $j\in\left\{ 1,\dots,\left|{\cal T}\right|\right\} $ \textbf{\textcolor{black}{then}}
\par\end{raggedright}
\begin{raggedright}
$\qquad\qquad$for $j=2,\dots,\left|{\cal T}\right|$ do
\par\end{raggedright}
\begin{raggedright}
$\qquad\qquad\qquad$Let $C\in{\cal C}$ be the component with $v_{j}\in C$
\par\end{raggedright}
\begin{raggedright}
$\qquad\qquad\qquad$$R_{i,T_{j}}\gets\bottom\left(R_{i,T_{j-1}}\cup R_{i,C}\right)$
\par\end{raggedright}
\begin{raggedright}
$\qquad$\textbf{else}
\par\end{raggedright}
\begin{raggedright}
$\qquad\qquad$\textbf{for} $\{v\}\in{\cal T}$ \textbf{do}
\par\end{raggedright}
\begin{raggedright}
$\qquad\qquad\qquad$Let $C\in{\cal C}$ be the component with $v\in C$
\par\end{raggedright}
\begin{raggedright}
$\qquad\qquad\qquad$$R_{i,\{v\}}=R_{i,C}$
\par\end{raggedright}
\begin{raggedright}
$\qquad$\textbf{for} $T\in{\cal T}$ \textbf{do}
\par\end{raggedright}
\begin{raggedright}
\textcolor{blue}{$\hfill$/$\negmedspace$/
evaluating bottom-$b$ sketches for the sets $\Reach_{G_{i}}^{(\tau)}\left(T\right)\setminus\Reach_{G_{i}}^{(\tau)}\left(S\right)$
for all $T\in{\cal T}$}
\par\end{raggedright}
\begin{raggedright}
$\qquad\qquad$\textbf{if} $\left|R_{i,T}\right|<b$ \textbf{then}
\par\end{raggedright}
\begin{raggedright}
$\qquad\qquad\qquad X_{i}\left(T\mid S\right)\gets\left|R_{i,T}\right|$
\par\end{raggedright}
\begin{raggedright}
$\qquad\qquad$\textbf{else}
\par\end{raggedright}
\begin{raggedright}
$\qquad\qquad\qquad$$X_{i}\left(T\mid S\right)\gets\frac{b}{\max\left(R_{i,T}\right)}-1$
\par\end{raggedright}
\begin{raggedright}
Let $X\left(T\mid S\right)=\frac{1}{\ell}\sum_{i=1}^{\ell}X_{i}\left(T\mid S\right)$
for all $T\in{\cal T}$
\par\end{raggedright}
\raggedright{}\textbf{return} $\left\{ X\left(T\mid S\right):T\in{\cal T}\right\} $
\end{algorithm}

We use Algorithm \ref{alg:estimate-gains-sketch-unrestricted} to
compute and evaluate bottom-$b$ sketches for the sets $\Reach_{G_{i}}^{(n)}\left(T\right)\setminus\Reach_{G_{i}}^{(n)}\left(S\right)$
for all $T\in{\cal T}$. For each graph $G_{i}$, the algorithm goes
through the following steps: First, we compute all strongly connected
components ${\cal C}$ in $G_{i}$ and contract them, resulting in
a directed acyclic graph $G'_{i}$ on the vertex set ${\cal C}$.
We then compute the bottom-$b$ sketches $Y_{i,C}$ for all components
$C\in{\cal C}$, whereas we set $Y_{i,C}=\emptyset$ if $C$ is already
covered by $\Reach_{G_{i}}^{(n)}\left(S\right)$. We then populate
the sketches $R_{i,C}$ such that at the end, $R_{i,C}=\Reach_{G_{i}}^{(n)}\left(C\right)\setminus\Reach_{G_{i}}^{(n)}\left(S\right)$.
Since the graph $G'_{i}$ is directed and acyclic, we can construct
$R_{i,C}$ via dynamic programming by considering all components $C\in{\cal C}$
in reverse topological order. After this, we construct the sketches
$R_{i,T}$ for the sets $T\in{\cal T}$. If ${\cal T}$ is a sequence
of prefixes, we merge the corresponding sketches $R_{i,T}$ one after
another. Finally, we evaluate the bottom-$b$ sketches and output
the average over all graphs $G_{i}$. 
\begin{lem}
\label{lem:unrestricted-sketch-correctness} At the end of Algorithm
\ref{alg:estimate-gains-sketch-unrestricted}, $R_{i,T}=\bottom\left(\left\{ y_{i,v}:v\in\Reach_{G_{i}}^{(n)}\left(T\right)\setminus\Reach_{G_{i}}^{(n)}\left(S\right)\right\} \right)$
for all $T\in{\cal T}$.
\end{lem}
\begin{proof}
Fix a live-edge graph $G_{i}$ in the outer for loop. Since the influence
is unrestricted, we have
\[\Reach_{G_{i}}^{(n)}\left(T\cup S\right)\setminus\Reach_{G_{i}}^{(n)}\left(S\right)=\Reach_{G_{i}}^{(n)}\left(T\right)\setminus\Reach_{G_{i}}^{(n)}\left(S\right)\]
for sets $T\in{\cal T}$. To compute a set $\Reach_{G_{i}}^{(n)}\left(T\right)\setminus\Reach_{G_{i}}^{(n)}\left(S\right)$,
it is sufficient to identify the vertices that are reachable from
$T$ but not reachable from $S$. Contracting the strongly connected
components ${\cal C}$ in $G_{i}$ preserves the reachability of all
vertices in the graph and we have $\Reach_{G_{i}}^{(n)}\left(\left\{ v\right\} \right)=\bigcup\Reach_{G'_{i}}^{(n)}\left(C\right)$
if $v\in C$. We therefore begin with showing that
\[R_{i,C}=\bottom\left(\left(\bigcup\Reach_{G'_{i}}^{(n)}\left(C\right)\right)\setminus\Reach_{G_{i}}^{(n)}\left(S\right)\right)\]
for all $C\in{\cal C}$. Since we set $Y_{i,C}$ for all $C\in{\cal C}$
such that $Y_{i,C}=\bottom\left(\left\{ y_{i,v}:v\in V\right\} \right)$
if $C\not\subseteq\Reach_{G_{i}}^{(n)}\left(S\right)$ and $Y_{i,C}=\emptyset$
otherwise, this is equivalent to showing that
\[R_{i,C}=\bottom\left(\bigcup\left\{ Y_{i,C'}:C'\in\Reach_{G'_{i}}^{(n)}\left(C\right)\right\} \right).\]
To see that this is exactly what we compute, note that the contraction
yields a directed acyclic graph $G'_{i}$. Hence, we can show by induction
over all $C\in{\cal C}$ in reverse topological order that
\[R_{i,C}=\bottom\left(\bigcup\left\{ Y_{i,C'}:C'\in\Reach_{G'_{i}}^{(n)}\left(C\right)\right\} \right)\]
for all components $C\in{\cal C}$. For all leaves $C\in{\cal C}$,
we set $R_{i,C}=Y_{i,C}$ and are done. Assume now that
\[R_{i,C'}=\bottom\left(\bigcup\left\{ Y_{i,C''}:C''\in\Reach_{G'_{i}}^{(n)}\left(C'\right)\right\} \right)\]
holds true for all outgoing neighbors $C'\in N_{G'_{i}}^{\mathrm{out}}\left(C\right)$.
Since $\Reach_{G'_{i}}^{(n)}\left(C\right)=\bigcup_{C\in N_{G'_{i}}^{\mathrm{out}}\left(C\right)}\Reach_{G'_{i}}^{(n)}\left(C'\right)$,
we also have
\[R_{i,C}=\bottom\left(\bigcup_{C\in N_{G'_{i}}^{\mathrm{out}}\left(C\right)}R_{i,C'}\right)=\bottom\left(\bigcup\Reach_{G'_{i}}^{(n)}\left(C\right)\right)\]
and we are done. We have thus established that
\[\Reach_{G_{i}}^{(n)}\left(\left\{ v\right\} \right)\setminus\Reach_{G_{i}}^{(n)}\left(S\right)=R_{i,C}\]
where $C\in{\cal C}$ is such that $v\in C$. Therefore, we compute
the sketches $R_{i,T}$ correctly by merging over $R_{i,C}$ where
$C\in{\cal C}$ is such that $v\in C$. 
\end{proof}

\begin{lem}
\label{lem:unrestricted-sketch-runtime} Algorithm \ref{alg:estimate-gains-sketch-unrestricted}
runs in time $O\left(b\mtot\right)$ for $b=\widetilde{O}\left(\frac{1}{\epsest^{2}}\log\left(\frac{\ell}{\prest}\right)\right)$
and succeeds with probability $\prest$. 
\end{lem}
\begin{proof}
Fix a live-edge graph $G_{i}$ let $n_{i}$ and $m_{i}$ be the number
of vertices edges in $G_{i}$ such that $\mtot=\sum_{i=1}^{\ell}m_{i}$.
We analyze the running time for each iteration of the outer for loop
over the live edge graphs $G_{i}$. The running time is determined
by the time needed to compute and order the strongly connected components
in reverse topological order, and to build and evaluate the bottom-$b$
sketches. Computing the strongly connected components and sorting
them topologically takes time only $O\left(n_{i}\right)$ via the
algorithm of \citet{tarjan72}. It thus remains to analyze the running
time for building the sketches. We first visit all components $C\in{\cal C}$
and create the bottom-$b$ sketch $Y_{i,C}$ for the component $C$
itself, which can be done in time $O\left(n_{i}\log b\right)$ using
a priority queue for each component. In the next step, we then merge
the bottom-$b$ sketches whereas we use one merge operation per edge
in $G'_{i}$ in order to obtain $R_{i,C}$ for all $C\in{\cal C}$.
A single merge of two bottom-$b$ sketches takes time $O\left(b\right)$
since we can store the sketches in sorted order. Hence, all merge
operations take time $O\left(b\left|E\left(G'_{i}\right)\right|\right)=O\left(bm_{i}\right)$.
Let us now analyze the time needed to build the final sketches $R_{i,T}$
for all $T\in{\cal T}$ and to evaluate them. Note here that we may
assume that in a preprocessing step at the beginning of Algorithm
\ref{alg:estimate-gains-sketch-unrestricted}, we have, for each graph
$G_{i}$, reduced the sets $T\in{\cal T}$ such that $T\subseteq V_{i}$.
We can do this in linear in $\left|{\cal T}\right|$ using a list
each vertex $v\in V$ that holds indices $i$ of graphs $G_{i}$ with
$v\in V_{i}$. As such, merging the bottom-$b$ sketches in order
to obtain $R_{i,T}$ for all $T\in{\cal T}$ also takes time only
$O\left(bn_{i}\right)$. At the end, evaluating the sketches for all
sets $T\in{\cal T}$ requires time also $O\left(n_{i}\right)$. Recall
that we may assume that $n_{i}\le m_{i}$. The total time to build
and evaluate sketches is therefore $O\left(bm_{i}\right)$. Summing
over all $\ell$ graphs, the total running time is $O\left(b\mtot\right)$.
\end{proof}

\begin{thm}
\label{lem:sketch} Assume that ${\cal T}$ consists either of singletons
or is a sequence of $k$ prefixes. Given as input $\epsest$ and $\prest$,
Algorithm \ref{alg:estimate-gains-sketch-unrestricted} runs in time
$O\left(b\mtot\right)$ for $b=\widetilde{O}\left(\frac{1}{\epsest^{2}}\log\left(\frac{\ell}{\prest}\right)\right)$
and it achieves a multiplicative error $\epsest$, failure probability
$\prest$. 
\end{thm}
\begin{proof}
By Lemma \ref{lem:unrestricted-sketch-correctness}, we compute sketches
$R_{i,T}$ such that
\[R_{i,T}=\bottom\left(\left\{ y_{i,v}:v\in\Reach_{G_{i}}^{(n)}\left(T\right)\setminus\Reach_{G_{i}}^{(n)}\left(S\right)\right\} \right).\]
Combined with Lemma \ref{lem:min-hash}, we therefore output estimates
that have multiplicative error at most $\epsest$ with probability
$\widetilde{\delta}=\frac{\prest}{\ell n}$ since
$b=\frac{6}{\epsest^{2}}\log\left(\frac{2\ell n}{\prest}\right)$.
By a union bound over the $\ell$ graphs $G_{i}$ and at most $n$
sets $T\in{\cal T}$, we thus have multiplicative error at most $\epsest$
with probability $\prest$. The running time follows from Lemma \ref{lem:unrestricted-sketch-runtime}. 
\end{proof}

\subsection{Our estimation routine for the restricted influence}

\begin{algorithm}
\begin{raggedright}
\caption{\label{alg:estimate-gains-sketch-restricted} Sketching algorithm
for estimating the restricted influence.}
\par\end{raggedright}
\begin{raggedright}
$\mathrm{EstimateGainsSketchRestricted}\left({\cal T},S,\epsest,\prest\right)$:
\par\end{raggedright}
\begin{raggedright}
\textbf{\textcolor{blue}{Input:}}\textcolor{blue}{{} Sets $T\subseteq V$
for all $T\in{\cal T}$ and $S\subseteq V$, multiplicative error
$\epsest$, failure probability $\prest$.}
\par\end{raggedright}
\begin{raggedright}
\textbf{\textcolor{blue}{Output:}}\textcolor{blue}{{} With probability
$1-\prest$, the algorithm constructs estimates $X\left(T\mid S\right)$
such that $\left(1-\epsest\right)I^{(\tau)}\left(T\mid S\right)\le X\left(T\mid S\right)\le\left(1+\epsest\right)I^{(\tau)}\left(T\mid S\right)$
for all $T\in{\cal T}$.}
\par\end{raggedright}
\begin{raggedright}
Define $\bottom\left(Y\right)$ as the set of the $b$ smallest values
in a finite set $Y\subseteq\Reals$
\par\end{raggedright}
\begin{raggedright}
Define $\max\left(Y\right)$ as the largest value of a finite set
$Y\subseteq\Reals$
\par\end{raggedright}
\begin{raggedright}
Define $N_{G}^{\mathrm{in}}\left(u\right)=\left\{ v:(v,u)\in E(G)\right\} $
as the set of incoming neighbors of $u$ in $G$
\par\end{raggedright}
\begin{raggedright}
Let $b=\frac{6}{\epsest^{2}}\log\left(\frac{2\ell n}{\prest}\right)$
\par\end{raggedright}
\begin{raggedright}
\textbf{for }$i=1,\dots,\ell$ \textbf{do}
\par\end{raggedright}
\begin{raggedright}
\textcolor{blue}{$\hfill$/$\negmedspace$/
building bottom-$b$ sketches $R_{i,\{u\}}$ for the sets $\Reach_{G_{i}}^{(\tau)}\left(\left\{ u\right\} \right)\setminus\Reach_{G_{i}}^{(\tau)}\left(S\right)$
for all $u\in V$}
\par\end{raggedright}
\begin{raggedright}
$\qquad$Sample $y_{i,v}\sim\mathrm{Uniform}\left(\left[0,1\right]\right)$
for all $v\in V$
\par\end{raggedright}
\begin{raggedright}
$\qquad$Initialize the distance multiset $D_{i,v}\gets\emptyset$
for all $v\in V$
\par\end{raggedright}
\begin{raggedright}
$\qquad$Initialize the sketch $Y_{i,v}\gets\emptyset$ for all $v\in V$
\par\end{raggedright}
\begin{raggedright}
$\qquad$\textbf{for} $v\in V\setminus\Reach_{G_{i}}^{(\tau)}\left(S\right)$
in ascending order of $y_{i,v}$ \textbf{do}
\par\end{raggedright}
\begin{raggedright}
\textcolor{blue}{$\hfill$/$\negmedspace$/ putting $y_{i,v}$ into
$R_{i,\{w\}}$ if $y_{i,v}\in\bottom\left(\left\{ y_{i,v'}:v'\in\Reach_{G_{i}}^{(\tau)}\left(\left\{ u\right\} \right)\setminus\Reach_{G_{i}}^{(\tau)}\left(S\right)\right\} \right)$
via reverse breadth-first search}
\par\end{raggedright}
\begin{raggedright}
$\qquad\qquad$$Q\gets\{v\}$
\par\end{raggedright}
\begin{raggedright}
$\qquad\qquad$$\mathrm{visited}[u]\gets\mathrm{false}$ for all $u\in V$
\par\end{raggedright}
\begin{raggedright}
$\qquad\qquad$\textbf{for} $d=0,\dots,\tau$ \textbf{do}
\par\end{raggedright}
\begin{raggedright}
$\qquad\qquad\qquad$$Q_{\mathrm{next}}\gets\emptyset$
\par\end{raggedright}
\begin{raggedright}
$\qquad\qquad\qquad$\textbf{for} $u\in Q$ \textbf{do}\textcolor{blue}{$\hfill$/$\negmedspace$/ encounter
vertex $u$}
\par\end{raggedright}
\begin{raggedright}
$\qquad\qquad\qquad\quad$$D_{i,u}\gets D_{i,u}\cup\left\{ d\right\} $
\par\end{raggedright}
\begin{raggedright}
$\qquad\qquad\qquad\quad$$\mathrm{visited}[u]\gets\mathrm{true}$
\par\end{raggedright}
\begin{raggedright}
$\qquad\qquad\qquad\quad$\textbf{if} $\left|R_{i,\{u\}}\right|<b$
\textbf{then}
\par\end{raggedright}
\begin{raggedright}
$\qquad\qquad\qquad\quad\quad$$R_{i,\{u\}}\gets R_{i,\{u\}}\cup\left\{ y_{i,v}\right\} $
\par\end{raggedright}
\begin{raggedright}
$\qquad\qquad\qquad\quad$\textbf{if} $\left|\left\{ d':d'\in D_{i,u}\land d'\le d\right\} \right|\le b$
\textbf{then}\textcolor{blue}{$\hfill$/$\negmedspace$/ traversal check for $u$}
\par\end{raggedright}
\begin{raggedright}
$\qquad\qquad\qquad\quad\quad$\textbf{for} $w\in N_{G_{i}}^{\mathrm{in}}(u)$
with $\mathrm{visited}[w]=\mathrm{false}$ \textbf{do}
\par\end{raggedright}
\begin{raggedright}
$\qquad\qquad\qquad\quad\quad\quad$$Q_{\mathrm{next}}\gets Q_{\mathrm{next}}\cup\left\{ w\right\} $
\par\end{raggedright}
\begin{raggedright}
$\qquad\qquad\qquad$$Q\gets Q_{\mathrm{next}}$
\par\end{raggedright}
\begin{raggedright}
\textcolor{blue}{$\hfill$/$\negmedspace$/ building bottom-$b$ sketches
$R_{i,T}$ for the sets $\Reach_{G_{i}}^{(\tau)}\left(T\right)\setminus\Reach_{G_{i}}^{(\tau)}\left(S\right)$
for all $T\in{\cal T}$}
\par\end{raggedright}
\begin{raggedright}
$\qquad$\textbf{if} ${\cal T}$ is a sequence of prefixes $T_{j}=\left\{ v_{1},v_{2},\dots,v_{j}\right\} $
for all $j\in\left\{ 1,\dots,\left|{\cal T}\right|\right\} $ \textbf{then}
\par\end{raggedright}
\begin{raggedright}
$\qquad\qquad$\textbf{for} $j=2,\dots,\left|{\cal T}\right|$ \textbf{do}
\par\end{raggedright}
\begin{raggedright}
$\qquad\qquad\qquad$$R_{i,T_{j}}\gets\bottom\left(R_{i,T_{j-1}}\cup R_{i,\left\{ v_{j}\right\} }\right)$
\par\end{raggedright}
\begin{raggedright}
$\qquad$\textbf{for} $T\in{\cal T}$ \textbf{do}
\par\end{raggedright}
\begin{raggedright}
\textcolor{blue}{$\hfill$/$\negmedspace$/
evaluating bottom-$b$ sketches for the sets $\Reach_{G_{i}}^{(\tau)}\left(T\right)\setminus\Reach_{G_{i}}^{(\tau)}\left(S\right)$
for all $T\in{\cal T}$}
\par\end{raggedright}
\begin{raggedright}
$\qquad\qquad$\textbf{if} $\left|R_{i,T}\right|<b$ \textbf{then}
\par\end{raggedright}
\begin{raggedright}
$\qquad\qquad\qquad X_{i}\left(T\mid S\right)\gets\left|R_{i,T}\right|$
\par\end{raggedright}
\begin{raggedright}
$\qquad\qquad$\textbf{else}
\par\end{raggedright}
\begin{raggedright}
$\qquad\qquad\qquad$$X_{i}\left(T\mid S\right)\gets\frac{b}{\max\left(R_{i,T}\right)}-1$
\par\end{raggedright}
\begin{raggedright}
$X\left(T\mid S\right)=\frac{1}{\ell}\sum_{i=1}^{\ell}X_{i}\left(T\mid S\right)$
for all $T\in{\cal T}$
\par\end{raggedright}
\raggedright{}\textbf{return} $\left\{ X\left(T\mid S\right):T\in{\cal T}\right\} $
\end{algorithm}

We use Algorithm \ref{alg:estimate-gains-sketch-restricted} to compute
and evaluate bottom-$b$ sketches for the sets $\Reach_{G_{i}}^{(\tau)}\left(T\right)\setminus\Reach_{G_{i}}^{(\tau)}\left(S\right)$
for all $T\in{\cal T}$. For each graph $G_{i}$, the algorithm goes
through the following steps: First, we start with building the bottom-$b$
sketches $\bottom\left(\Reach_{G_{i}}^{(\tau)}\left(\left\{ w\right\} \right)\setminus\Reach_{G_{i}}^{(\tau)}\left(S\right)\right)$
for all vertices $w\in V$. We use sets $R_{i,\{w\}}$ to hold these
bottom-$b$ sketches, and we populate them with elements $y_{i,v}\sim\mathrm{Uniform}\left(\left[0,1\right]\right)$.
We do this via a reverse breadth-first search starting from each $v\in V$
in increasing order of $y_{i,v}$. The reverse breadth-first search
ensures that we add $y_{i,v}$ into the bottom-$b$ sketches $R_{i,\{w\}}$
for vertices $w$ whenever $v\in\bottom\left(\Reach_{G_{i}}^{(\tau)}\left(\left\{ w\right\} \right)\setminus\Reach_{G_{i}}^{(\tau)}\left(S\right)\right)$.
We use a traversal check to stop the breadth-first traversal through
a vertex $u$ to enhance the running time. For this, we maintain a
multiset of distances $D_{i,u}$ which collects all distances at which
we performed a traversal check for $u$. If there are more than $b$
shorter distances, we know that all sets $R_{i,w}$ for vertices $w$
that are reachable on a shortest path through $u$ are already complete,
so we do not need to continue the traversal through $u$. We show
in Lemma \ref{lem:restricted-sketch-correctness} that we still obtain
the correct bottom-$b$ sketches $\bottom\left(\Reach_{G_{i}}^{(\tau)}\left(\left\{ w\right\} \right)\setminus\Reach_{G_{i}}^{(\tau)}\left(S\right)\right)=R_{i,\{w\}}$.
This algorithm and analysis is due to \citet{cohen08tighter}. In
the second step, we aggregate bottom-$b$ sketches if the input ${\cal T}$
is a sequence of prefixes. If this is the case, it suffices to merge
one sketch after another. In the last step, we then evaluate the bottom-$b$
sketches. We then output, for each set $T\in{\cal T}$, the average
sketch value over all graphs $G_{i}$. 
\begin{lem}
\label{lem:restricted-sketch-correctness} At the end of Algorithm
\ref{alg:estimate-gains-sketch-restricted}, $R_{i,T}=\bottom\left(\left\{ y_{i,v}:v\in\Reach_{G_{i}}^{(\tau)}\left(T\right)\setminus\Reach_{G_{i}}^{(\tau)}\left(S\right)\right\} \right)$
for all $T\in{\cal T}$.
\end{lem}
\begin{proof}
We first show that at the end of the algorithm,
\[R_{i,\{w\}}=\bottom\left(\left\{ y_{i,v}:v\in\Reach_{G_{i}}^{(\tau)}\left(T\right)\setminus\Reach_{G_{i}}^{(\tau)}\left(S\right)\right\} \right)\]
for all $w\in V$. Fix the iteration in which we consider a graph
$G_{i}$ in the first for loop and perform a reverse breadth-first
search from a vertex $v\in V$. Since we explicitly exclude vertices
$v\in\Reach_{G_{i}}^{(\tau)}\left(S\right)$ from the second for loop,
we never include $y_{i,v}$ for $v\in\Reach_{G_{i}}^{(\tau)}\left(S\right)$
into $R_{i,\{w\}}$. We show that at the end of the reverse breadth-first
search, $v\in R_{i,w}$ if
\[v\in\bottom\left(\left\{ y_{i,v}:v\in\Reach_{G_{i}}^{(\tau)}\left(w\right)\setminus\Reach_{G_{i}}^{(\tau)}\left(S\right)\right\} \right)\]
for all $w\in V$.

We perform a (reverse) breadth-first search from $v$ where we encounter
vertices $u$ when $d=\mathrm{dist}_{G_{i}}\left(u,v\right)$. The
breadth-first search is modified to contain a traversal check at vertex
$u$. Without the traversal check, we will encounter all vertices
$u$ with with $\mathrm{dist}_{G_{i}}\left(u,v\right)\le\tau$. However,
if the traversal check fails at vertex $u$, we may not encounter
vertices $u'\in V$ with $d\left(u',v\right)=d\left(u',u\right)+d\left(u,v\right)$,
which follows from basic properties of the breadth-first search. We
use this fact to show that the following invariant holds at each traversal
check for $u$:
\begin{itemize}
\item $D_{i,u}$ contains the $b$ smallest distances in the multiset $\left\{ \mathrm{dist}_{G_{i}}\left(u,v'\right):y_{v'}\le y_{v}\land\mathrm{dist}_{G_{i}}\left(u,v'\right)\le\tau\right\} $
for all $u\in V$
\end{itemize}
We thus need to argue that we include $\mathrm{dist}_{G_{i}}\left(u',v\right)$
in $D_{i,u'}$ unless there are at least $b$ vertices $v'$ with
$y_{v'}\le y_{v}$ and $\mathrm{dist}_{G_{i}}\left(u',v'\right)\le\tau$.
Assume the invariant holds at the traversal check for a vertex $u$.
If the traversal check fails, we have $\left|\left\{ d\in D_{i,u}:d\le\mathrm{dist}_{G_{i}}\left(u,v\right)\right\} \right|>b$,
which by our assumption means that there are more than $b$ vertices
$v'$ with $\mathrm{dist}_{G_{i}}\left(u,v'\right)\le\mathrm{dist}_{G_{i}}\left(u,v\right)$.
We may not encounter vertices $u'\in V$ with $d\left(u',v\right)=d\left(u',u\right)+d\left(u,v\right)$.
However, for these vertices $u'$ we have more than $b$ vertices
$v'$ with $\mathrm{dist}_{G_{i}}\left(u',v'\right)\le\mathrm{dist}_{G_{i}}\left(u',v\right)+\mathrm{dist}_{G_{i}}\left(u,v\right)\le\mathrm{dist}_{G_{i}}\left(u',v\right)+\mathrm{dist}_{G_{i}}\left(u,v\right)=\mathrm{dist}_{G_{i}}\left(u',v\right)$.
We thus do not need to add the distance to $D_{i,u'}$.

We now use this invariant to show that if the traversal check fails
at any vertex $u$ then
\[y_{i,v}\notin\bottom\left(\left\{ y_{i,v}:v\in\Reach_{G_{i}}^{(\tau)}\left(w\right)\setminus\Reach_{G_{i}}^{(\tau)}\left(S\right)\right\} \right)\]
for all $w$ with $\mathrm{dist}_{G_{i}}\left(w,v\right)=\mathrm{dist}_{G_{i}}\left(w,u\right)+\mathrm{dist}_{G_{i}}\left(u,v\right)$.
If the traversal check fails at vertex $u$, we have that
$\left|\left\{ d\in D_{i,u}:d\le\mathrm{dist}_{G_{i}}\left(u,v\right)\right\} \right|>b$,
which by our invariant means that there are more than $b$ vertices
$v'$ with $\mathrm{dist}_{G_{i}}\left(u,v'\right)\le\mathrm{dist}_{G_{i}}\left(u,v\right)$.
Hence, there are also more than $b$ vertices $v'$ with
\[\mathrm{dist}_{G_{i}}\left(w,v'\right)\le\mathrm{dist}_{G_{i}}\left(w,u\right)+\mathrm{dist}_{G_{i}}\left(u,v'\right)\le\mathrm{dist}_{G_{i}}\left(w,u\right)+\mathrm{dist}_{G_{i}}\left(u,v\right)=\mathrm{dist}_{G_{i}}\left(w,v\right).\]
As such, we also know that
\[y_{i,v}\notin\bottom\left(\left\{ y_{i,v}:v\in\Reach_{G_{i}}^{(\tau)}\left(w\right)\setminus\Reach_{G_{i}}^{(\tau)}\left(S\right)\right\} \right).\]

It remains to argue that we correctly construct $R_{i,T}$ from $R_{i,\{v\}}$
for all $T\in{\cal T}$. If the set ${\cal T}$ consists of singletons,
we are done. If ${\cal T}$ consists of a sequence of prefixes $T_{j}=\left\{ v_{1},\dots,v_{j}\right\} $
for all $j\in\left\{ 1,\dots,\left|{\cal T}\right|\right\} $, we
merge the bottom-$b$ sketches $R_{i,\{w\}}$ for singletons $w\in V$.
The correctness follows inductively since for all $j\ge2$ we have
\[\Reach_{G_{i}}^{(\tau)}\left(T_{j}\right)\setminus\Reach_{G_{i}}^{(\tau)}\left(S\right)=\left(\Reach_{G_{i}}^{(\tau)}\left(T_{j-1}\right)\setminus\Reach_{G_{i}}^{(\tau)}\left(S\right)\right)\cup\left(\Reach_{G_{i}}^{(\tau)}\left(\left\{ v_{j}\right\} \right)\setminus\Reach_{G_{i}}^{(\tau)}\left(S\right)\right)\]
and therefore we have
\begin{align*}
&\bottom\left(\Reach_{G_{i}}^{(\tau)}\left(T_{j}\right)\setminus\Reach_{G_{i}}^{(\tau)}\left(S\right)\right)\\
&\quad=\bottom\bigg(\bottom\left(\Reach_{G_{i}}^{(\tau)}\left(T_{j-1}\right)\setminus\Reach_{G_{i}}^{(\tau)}\left(S\right)\right)\cup\bottom\left(\Reach_{G_{i}}^{(\tau)}\left(\left\{ v_{j}\right\} \right)\setminus\Reach_{G_{i}}^{(\tau)}\left(S\right)\right)\bigg)\\
&\quad=\bottom\left(R_{i,T_{j-1}}\cup R_{i,\{v_{j}\}}\right).
\end{align*}
\end{proof}

\begin{lem}
\label{lem:restricted-sketch-runtime} Algorithm \ref{alg:estimate-gains-sketch-restricted}
runs in time $\widetilde{O}\left(\mtot b\log\left(b\right)\right)$
for $b=\widetilde{O}\left(\frac{1}{\epsest^{2}}\log\left(\frac{\ell n}{\prest}\right)\right)$
and succeeds with probability $\prest$. 
\end{lem}
\begin{proof}
We only analyze the first part of the algorithm where we compute the
sketches $R_{i,\{w\}}$ for all $1\le i\le\ell$ and $w\in V$, since
the second part where we combine and evaluate the sketches is identical
with Algorithm \ref{alg:estimate-gains-sketch-unrestricted}. The
running time for a single iteration $1\le i\le\ell$ is dominated
by the processing time for each vertex $u\in V$ in the reverse breadth-first
search. The processing time for $u$ consists of performing the traversal
check, i.e. determining the $b$-th smallest distance in $D_{i,u}$,
which takes time $O\left(\log b\right)$ using a priority queue. Additionally,
if the traversal-check succeeds we proceed with the breadth-first
traversal through $u$ and the processing time therefore includes
the time to iterate through all neighbors $w\in N_{\mathrm{in}}\left(u\right)$
and adding them to the queue. The total running time is therefore
$O\left(\sum_{u\in V}c_{u}\deg\left(u\right)\log b\right)$ where
we denote with $c_{u}$ the number of times that we execute the traversal
check for $u$. It thus remains to bound the expected value of $c_{u}$.
To this end, fix a single $u\in V$ and let $v_{1},v_{2},\dots,v_{n}$
be the vertices in $V$ in order of increasing distance $\mathrm{\mathrm{dist}}_{G_{i}}\left(u,v_{j}\right)$.
If two vertices have the the same distance, we break ties arbitrarily.
Since we start the reverse breadth-first searches from vertices $v_{j}$
with increasing $y_{i,v_{j}}$, we only succeed with the traversal
check for $v_{j}$ if there are less than $b$ vertices $v\in\left\{ v_{1},v_{2},\dots,v_{j-1}\right\} $
with $y_{i,v_{j}}\le y_{i,v}$. Let $X_{j}$ be the indicator random
variable for this event. We have that $c_{u}=\sum_{j=1}^{n}X_{j}$
and by linearity of expectation $\mathbb{E}\left[c_{u}\right]=\sum_{j=1}^{n}\Pr\left[X_{j}=1\right]$.
To compute $\Pr\left[X_{j}=1\right]$, for each $1\le b'\le b$ let
$E_{b'}$ be the event that $y_{i,v_{j}}$ is the $b'$-th smallest
value among $\left\{ y_{i,v_{1}},y_{i,v_{2}},\dots,y_{i,v_{j}}\right\} $.
By definition of these events, $\Pr\left[X_{j}=1\right]=\sum_{b'=1}^{b}\Pr\left[E_{b'}\right]$.
Since all values $y_{i,v}$ are sampled independently, $\Pr\left[E_{b'}\right]=\frac{1}{j}$
for all $1\le b'\le b$. We therefore get that $\Pr\left[X_{j}\right]=\frac{b}{j}$
and $\E\left[c_{u}\right]=\sum_{j=1}^{n}\frac{b}{n}=O\left(b\log n\right)$.
The total expected running time for the first part of the algorithm
is therefore 
\[
O\left(\log\left(b\right)\sum_{i=1}^{\ell}\sum_{u\in V}\mathbb{E}\left[c_{u}\right]\deg\left(u\right)\right)=O\left(m_{i}b\log\left(b\right)\log\left(n\right)\right).
\]
The additional time for the second part of the algorithm is $O\left(bn_{i}\right)=O\left(bm_{i}\right)$
per graph where $n_{i}=\left|V_{i}\right|$ as in Lemma \ref{lem:sketch}.
The overall expected running time is therefore $\widetilde{O}\left(\mtot b\log\left(b\right)\right)$. 
\end{proof}

\begin{thm}
\label{thm:sketch-restricted} Assume that ${\cal T}$ consists either
of only singletons or is a sequence of prefixes. Given as input $\epsest$
and $\prest$, Algorithm \ref{alg:estimate-gains-sketch-restricted}
runs in expected time $\widetilde{O}\left(\mtot b\log\left(b\right)\right)$
for $b=\widetilde{O}\left(\frac{1}{\epsest^{2}}\log\left(\frac{\ell}{\prest}\right)\right)$
and it achieves a multiplicative error $\epsest$, failure probability
$\prest$, and no additive error ($\cest=0$).
\end{thm}
\begin{proof}
Lemma \ref{lem:restricted-sketch-runtime} shows the expected running
time for Algorithm \ref{alg:estimate-gains-sketch-restricted}. By
Lemma \ref{lem:restricted-sketch-correctness}, Algorithm \ref{alg:estimate-gains-sketch-restricted}
computes the bottom-$b$ sketch for all sets $T\in{\cal T}$. Let
$\widetilde{\delta}=\frac{\prest}{\ell n}$ be the allowed failure
probability for a single sketch for a single graph. We choose $b=\frac{6}{\epsest^{2}}\log\left(\frac{2}{\widetilde{\delta}}\right)$,
so by Lemma \ref{lem:min-hash} and a union bound over the at most
$n$ sets $T\in{\cal T}$ and $\ell$ graphs $G_{i}$, we obtain that
all estimates have $\epsest$-multiplicative error with probability
at least $1-\prest$.
\end{proof}

\subsection{Our algorithm for influence maximization on observed cascades}

\begin{algorithm}
\begin{raggedright}
\caption{\label{alg:observed-cascades} Our algorithm for observed cascades
that instantiates the template algorithm \ref{alg:optimization} using
the estimation via Algorithm \ref{alg:estimate-gains-sketch-unrestricted}.}
\par\end{raggedright}
\begin{raggedright}
$\mathrm{ObservedCascades}\left(\epsilon,\delta\right):$
\par\end{raggedright}
\begin{raggedright}
\textbf{\textcolor{blue}{Input:}}\textcolor{blue}{{} Error parameters
$\epsilon$ and $\delta$}
\par\end{raggedright}
\begin{raggedright}
\textbf{\textcolor{blue}{Output:}}\textcolor{blue}{{} With probability
$1-\delta$, we obtain a solution $S\subseteq V$ with $\left|S\right|\le k$
such that $I^{(\tau)}\left(S\right)\ge\left(1-\frac{1}{e}-\epsilon\right)\opt$.}
\par\end{raggedright}
\begin{raggedright}
$\optguess=n$
\par\end{raggedright}
\begin{raggedright}
$\epsest=\epsilon$
\par\end{raggedright}
\begin{raggedright}
\textbf{while} $\optguess\ge1$ \textbf{do}
\par\end{raggedright}
\begin{raggedright}
$\qquad$$\optguess\gets\frac{1}{2}\optguess$
\par\end{raggedright}
\begin{raggedright}
$\qquad$$S\gets\mathrm{ObservedCascadesGuess}\left(\epsilon,\frac{\delta}{\log n},\epsest,\optguess\right)$
\par\end{raggedright}
\begin{raggedright}
$\qquad$\textbf{if} $\widehat{I^{(\tau)}}\left(S\right)\ge\left(1-\frac{1}{e}-\epsilon-3\epsest\right)\opt$
\textbf{then}
\par\end{raggedright}
\begin{raggedright}
$\qquad\qquad$\textbf{break}
\par\end{raggedright}
\begin{raggedright}
\textbf{return} $S$
\par\end{raggedright}
\begin{raggedright}
$\mathrm{ObservedCascadesGuess}\left(\epsilon,\delta,\epsest,\optguess\right):$
\par\end{raggedright}
\begin{raggedright}
$T\gets\frac{8}{\epsilon}\log\left(\frac{n}{\delta'}\right)$
\par\end{raggedright}
\begin{raggedright}
$S\gets\emptyset$
\par\end{raggedright}
\begin{raggedright}
$\prest=\frac{\delta}{N}$
\par\end{raggedright}
\begin{raggedright}
$t=0$$\hfill$\textcolor{blue}{/$\negmedspace$/ for the analysis}
\par\end{raggedright}
\begin{raggedright}
\textbf{for }$r=0,1,\dots,\frac{2}{\epsilon}$ \textbf{do}
\par\end{raggedright}
\begin{raggedright}
$\qquad$$\alpha\gets\frac{1}{k}\optguess\left(1-\epsilon\right)^{r}$
\par\end{raggedright}
\begin{raggedright}
$\qquad$$U\gets V$
\par\end{raggedright}
\begin{raggedright}
$\qquad$$t_{r}\gets t$$\hfill$\textcolor{blue}{/$\negmedspace$/
for the analysis}
\par\end{raggedright}
\begin{raggedright}
$\qquad$\textbf{for} $T$ iterations \textbf{do}
\par\end{raggedright}
\begin{raggedright}
$\qquad\qquad$$t\gets t+1$$\hfill$\textcolor{blue}{/$\negmedspace$/
for the analysis}
\par\end{raggedright}
\begin{raggedright}
$\qquad\qquad$$\left\{ X\left(\left\{ u\right\} \mid S\right)\colon u\in U\right\} \gets\mathrm{EstimateGainsSketch}\left(\left\{ \left\{ u\right\} \colon u\in U\right\} ,S,\epsest,\prest\right)$
\par\end{raggedright}
\begin{raggedright}
$\qquad\qquad$$U\gets\left\{ u\in U:X\left(\left\{ u\right\} \mid S\right)\ge\alpha\right\} $\textcolor{blue}{$\hfill$/$\negmedspace$/
filtering}
\par\end{raggedright}
\begin{raggedright}
$\qquad\qquad$\textbf{if} $U=\emptyset$ \textbf{then}
\par\end{raggedright}
\begin{raggedright}
\textbf{$\qquad\qquad\qquad$break}
\par\end{raggedright}
\begin{raggedright}
$\qquad\qquad$Let $v_{1},v_{2},\dots,v_{\left|U\right|}$ be a random
permutation of the elements in $U$
\par\end{raggedright}
\begin{raggedright}
$\qquad\qquad$$s\gets\min\left\{ k-\left|S\right|,\left|U\right|\right\} $
\par\end{raggedright}
\begin{raggedright}
$\qquad\qquad$Let $T_{i}=\left\{ v_{1},v_{2},\dots,v_{i}\right\} $
for all $1\le i\le s$\textcolor{blue}{$\hfill$/$\negmedspace$/
prefixes of random permutation}
\par\end{raggedright}
\begin{raggedright}
$\qquad\qquad$$\left\{ X\left(T_{i}\mid S\right)\colon1\le i\le s\right\} \gets\mathrm{EstimateGainsSketch}\left(\left\{ T_{i}:1\le i\le s\right\} ,S,\epsest,\prest\right)$
\par\end{raggedright}
\begin{raggedright}
$\qquad\qquad$$i^{*}\gets\arg\max\left\{ 1\le i\le s:X\left(T_{i}\mid S\right)\ge\left(1-\epsilon\right)\left(1-\epsest\right)\alpha i\right\} $
\par\end{raggedright}
\begin{raggedright}
$\qquad\qquad$$S\gets S\cup T_{i^{*}}$\textcolor{blue}{$\hfill$/$\negmedspace$/
adding elements to the solution}
\par\end{raggedright}
\begin{raggedright}
$\qquad\qquad$\textbf{if} $\left|S\right|=k$ \textbf{then}
\par\end{raggedright}
\begin{raggedright}
$\qquad\qquad\qquad$\textbf{return} $S$
\par\end{raggedright}
\begin{raggedright}
$\qquad$\textbf{if} $U\not=\emptyset$ \textbf{then}
\par\end{raggedright}
\begin{raggedright}
$\qquad\qquad$\textbf{return} Failure
\par\end{raggedright}
\raggedright{}\textbf{return $S$}
\end{algorithm}
 
\begin{thm}
Given as input $\epsilon$ and $\delta$, Algorithm \ref{alg:observed-cascades}
runs in time $\widetilde{O}\left(Nb\mtot\right)$ for $N=\widetilde{O}\left(\frac{1}{\epsilon^{2}}\log\left(\frac{1}{\delta}\right)\right)$
and $b=\widetilde{O}\left(\frac{1}{\epsilon^{2}}\log\left(\frac{\ell N}{\delta}\right)\right)$
and it returns a solution $S$ satisfying $I^{(\tau)}\left(S\right)\geq\left(1-\frac{1}{e}-\epsilon\right)\opt$
with probability $1-\delta$. 
\end{thm}
\begin{proof}
As in the algorithm for general models, there are $2\log n$ iterations
of the search which determines $\optguess$. \ref{lem:estimate-gains-variance},
any call Algorithm \ref{alg:estimate-gains-sketch-unrestricted} fails
to be $\epsest$-multiplicatively correct with probability at most
$\prest=\frac{\delta'}{N}=\frac{\delta}{N\log n}$. By a union bound
over all iterations, we obtain that all calls are $\epsest$-multiplicatively
correct with probability at least $1-\delta$. For the remainder of
the proof, let us thus condition on the event that these calls are
$\epsest$-multiplicatively correct. By Theorem \ref{thm:template-ideal},
any call to $\mathrm{ObservedCascadesGuess}$ fails with probability
at most $\delta'=N\prest+\delta'=\frac{\delta}{\log n}$. By a union
bound over all $2\log n$ iterations of the outer search, the overall
algorithm fails with probability at most $O(\delta)$. For the remainder
of the proof, let us also condition on the event that no call to $\mathrm{ObservedCascadesGuess}$
fails. 

Let us now assume we are in an iteration where $\frac{1}{2}\opt\le\optguess$
and we return a solution $S$ because we satisfy the breaking condition
$X\left(S\mid\emptyset\right)\ge\left(1-\frac{1}{e}-\epsilon-3\epsest\right)\optguess$.
If $\optguess\ge\opt$ then Lemma \ref{lem:sketch} guarantees that
\begin{align*}
I^{(\tau)}\left(S\right) & \ge\left(1-\epsest\right)X\left(S\mid\emptyset\right)\\
 & \ge\left(1-\epsest\right)\left(1-\frac{1}{e}-\epsilon-3\epsest\right)\optguess\\
 & \ge\left(1-\frac{1}{e}-\epsilon-4\epsest\right)\optguess\\
 & \ge\left(1-\frac{1}{e}-\epsilon-4\epsest\right)\opt
\end{align*}
and we are done. If $\optguess\le\opt$ then Theorem \ref{thm:template-ideal}
guarantees a solution $S$ with 
\[
I^{(\tau)}\ge\left(1-\frac{1}{e}-\epsilon-2\epsest\right)\opt.
\]
Furthermore, Lemma \ref{lem:sketch} guarantees
\begin{align*}
X\left(S\mid\emptyset\right) & \ge\left(1-\epsest\right)I^{(\tau)}\\
 & \ge\left(1-\frac{1}{e}-\epsilon-3\epsest\right)\opt\\
 & \ge\left(1-\frac{1}{e}-\epsilon-3\epsest\right)\optguess.
\end{align*}
and we therefore satisfy the breaking condition and return $S$. 

The algorithm uses at most $O\left(\log n\right)=\widetilde{O}\left(1\right)$
calls to $\mathrm{ObservedCascadesGuess}$ in its outer search for
$\opt$. The running time for each call to Algorithm \ref{alg:estimate-gains-sketch-unrestricted}
is dominated by running time for the $N$ calls to the estimation
oracle Algorithm \ref{alg:estimate-gains-sketch-unrestricted}. By
Lemma \ref{lem:sketch}, each call to Algorithm \ref{alg:estimate-gains-sketch-unrestricted}
takes time $O\left(b\mtot\right)$ for $b=\widetilde{O}\left(\frac{1}{\epsest^{2}}\log\left(\frac{\ell}{\prest}\right)\right)$.
Overall, we thus require time
\[
\widetilde{O}\left(Nb\mtot\right).
\]
\end{proof}

\section{\label{sec:concentration} Concentration for independent cascade
models}

\global\long\def\G{\mathcal{G}}%

In this section, we will show a concentration bound for samples from
the IC model, which we stated as Theorem \ref{thm:ic-sample-complexity}.
The IC model has the benefit that we can defer the random sampling
of an edge $(u,v)$ until we activate vertex $u\in V$. The idea of
our proof is thus to decompose the influence propagation into the
individual steps $0\le t\le\tau$. Specifically, for some fixed $t$,
we condition on the set of vertices $\Active^{(t-1)}\left(S\right)$
activated so far. We only consider the random activation $\Active^{(t)}\left(S\right)$
in step $t$, and observe that the activation of all vertices $v\notin\Active^{(t-1)}\left(S\right)$
is independent. In expectation, the remaining steps are submodular
over the set of random activations $\Active^{(t)}\left(S\right)\setminus\Active^{(t-1)}\left(S\right)$.
This is analogous to the induction in \citet{sadeh20}. However, to
analyze the resulting distribution, we use the following theorem due
to \citet{vondrak10} that upper bounds the MGF of any monotone submodular
function over sets $S$ that are sampled by including elements independently.
For a vector $p\in\left[0,1\right]^{n}$, we write $S\sim p$ to denote
that $S$ is a random set where each $i\in[n]$ is included $S$ independently
with probability $p_{i}$.

Using Theorem \ref{thm:submodular-mgf}, we can show the following
bound on the MGF of a single sample from an IC model via a notion
of model reduction for IC models. We use ${\cal G}$ to denote the
distribution over live-edge graphs in the original IC model. We define
the model reduction ${\cal G}\setminus T$ as the distribution of
live-edge graphs where we simply remove $T$ and all edges that have
at least one endpoint in $T$ from the graph. We use identical edge
probabilities for all remaining edges. The following notion of model
reduction is due to \citet{sadeh20} and we use it to set up our induction.
\begin{lem}
\label{lem:model-reduction} Let $T$ be a fixed seed set and $S$
be a random set where we include each node $v\in V\setminus T$ independently
with probability $p_{v}=\Pr_{G\sim{\cal G}}\left[v\in\Reach_{G}^{(1)}\left(T\right)\right]$.
Then, the random variables $\Reach_{G}^{(\tau)}\left(T\right)$ for
$G\sim{\cal G}$ and $T\cup\Reach_{G'}^{(\tau-1)}\left(S\right)$
for $G'\sim{\cal G}\setminus T$ are identically distributed. 
\end{lem}
\begin{proof}
We sample all edges $(u,v)$ with probability $p_{uv}$ and let $G$
be the resulting graph. To obtain a sample $G'$ from ${\cal G}\setminus T$,
we can simply remove $T$ from $G$ and all edges that have at least
one endpoint in $T$. Let $S=\Reach_{G}^{(1)}\left(T\right)\setminus T$
the set of vertices that are reachable within a single step from $T$
in $G$. Note that nodes $v\in V$ are included independently in $S$
since all incoming edges to $v$ are sampled independently and $T$
is deterministic. It remains to show that $\Reach_{G}^{(\tau)}\left(T\right)=\Reach_{G'}^{(\tau-1)}\left(S\right)$.
To see this, note that all vertices $v\in V$ reachable from $T$
in $\tau$ steps are reachable from $S$ in $\tau-1$ steps since
$S=\Reach_{G}^{(1)}\left(T\right)$.
\end{proof}

To differentiate the influence in the original and reduced model,
we write $I_{{\cal G}}^{(\tau)}=\E_{G\sim{\cal G}}\left[\left|\Reach_{G}^{(\tau)}\left(T\right)\right|\right]$
and $I_{{\cal G}\setminus T}^{(\tau)}=\E_{G\sim{\cal G}\setminus T}\left[\left|\Reach_{G}^{(\tau)}\left(T\right)\right|\right]$. 
\begin{lem}
\label{lem:model-reduction-singleton} In the IC model, the maximum
singleton influence does not increase under model reductions, i.e.
$\max_{v\in V}I_{{\cal G}}^{(\tau)}\left(v\right)\ge\max_{v\in V\setminus T}I_{{\cal G}\setminus T}^{(\tau)}\left(v\right)$. 
\end{lem}
\begin{proof}
Let $G\sim{\cal G}$ be a sample for ${\cal G}$ and $G'$ be the
result of removing $T$ and all incident edges from $G$, which is
a sample for ${\cal G}\setminus T$. Since we only remove vertices
from the graph but do not change the model in any other way, we necessarily
have that $\Reach_{G'}^{(\tau)}\left(v\right)\subseteq\Reach_{G}^{(\tau)}\left(v\right)$.
Therefore, we also have that in expectation $I_{{\cal G}\setminus T}^{(\tau)}\left(v\right)\le I_{{\cal G}}^{(\tau)}\left(v\right)$
which proves the claim.
\end{proof}

\begin{thm}
\label{lem:mgf-bound} Let $M=\max_{v\in V}I_{\G}^{(\tau)}(v)$. Let
$\G$ be any live-edge model. For any $\lambda\in\Reals$, we have
\begin{align*}
\E_{G\sim\G}\left[\exp\left(\lambda\cdot\left|\Reach_{G}^{(\tau)}(T)\right|\right)\right] & \leq\exp\left(A^{(\tau)}\cdot\E_{G\sim\G}\left[\left|\Reach_{G}^{(\tau)}(T)\right|\right]\right)\cdot\exp\left(-B^{(\tau)}\left|T\right|\right)\\
 & =\exp\left(A^{(\tau)}\cdot I_{\G}^{(\tau)}(T)\right)\cdot\exp\left(-B^{(\tau)}\left|T\right|\right)
\end{align*}
 where
\begin{align*}
A^{(\tau)} & =\begin{cases}
\lambda & \text{if }\tau=0\\
\frac{1}{M}\left(\exp\left(MA^{(\tau-1)}\right)-1\right) & \text{if }\tau\geq1
\end{cases}\\
B^{(\tau)} & =\begin{cases}
0 & \text{if }\tau=0\\
A^{(\tau)}+B^{(\tau-1)}-\lambda & \text{if }\tau\geq1
\end{cases}
\end{align*}
\end{thm}
\begin{proof}
We proceed by induction on $\tau$. In the base case $\tau=0$, we
have $\Reach_{G}^{(0)}(T)=T$ for all $G\sim{\cal G}$ and thus we
can set $A^{(0)}=\lambda$ and $B^{(0)}=0$.

Consider now any $\tau\geq1$. Let $p_{v}=\Pr_{G\sim\G}\left[v\in\Reach_{G}^{(1)}(T)\right]$
and let $\G\setminus T$ be the reduced model with $T$ removed. By
Lemma \ref{lem:model-reduction}, the events $v\in\Reach_{G}^{(1)}\left(T\right)$
are mutually independent for all $v\in V$. We can thus decompose
the sampling process $G\sim\G$ into a two--stage process where we
first sample $S\sim p$ and then we sample $G'\sim\G\setminus T$.
By Lemma~\ref{lem:model-reduction}, $\Reach_{G}^{(\tau)}\left(T\right)$
where $G\sim\G$ has the same distribution as $T\cup\Reach_{G'}^{(\tau-1)}\left(S\right)$
where $S\sim p$ and $G'\sim\G\setminus T$. Thus,
\begin{align*}
\E_{G\sim\G}\left[\exp\left(\lambda\left|\Reach_{G}^{(\tau)}(T)\right|\right)\right] & =\E_{S\sim p}\left[\E_{G'\sim\G\setminus T}\left(\exp\left(\lambda\left|T\cup\Reach_{G'}^{(\tau-1)}\left(S\right)\right|\right)\right)\right]\\
 & =\exp\left(\lambda\left|T\right|\right)\E_{S\sim p}\left[\E_{G'\sim\G\setminus T}\left(\exp\left(\lambda\left|\Reach_{G'}^{(\tau-1)}\left(S\right)\right|\right)\right)\right]
\end{align*}
By induction,
\begin{align*}
\E_{G'\sim\G\setminus T}\left(\exp\left(\lambda\left|\Reach_{G'}^{(\tau-1)}\left(S\right)\right|\right)\right) & \leq\exp\left(-B^{(\tau-1)}\left|T\right|\right)\cdot\exp\left(A^{(\tau-1)}\cdot I_{\G\setminus T}^{(\tau-1)}(S)\right)\\
\implies & \E_{S\sim p}\left[\E_{G'\sim\G\setminus T}\left(\exp\left(\lambda\left|\Reach_{G'}^{(\tau-1)}\left(S\right)\right|\right)\right)\right]\\
 & \leq\exp\left(-B^{(\tau-1)}\left|T\right|\right)\cdot\E_{S\sim p}\left[\exp\left(A^{(\tau-1)}\cdot I_{\G\setminus T}^{(\tau-1)}(S)\right)\right]
\end{align*}
We apply Theorem \ref{thm:submodular-mgf} to $f(S)=\frac{1}{M}I_{\G\setminus T}^{(\tau-1)}(S)$.
Lemma \ref{lem:model-reduction-singleton} shows that the model reduction
ensures that the singleton influence does not increase, so we have
$\max_{v\in V\setminus T}I_{\G\setminus T}^{(\tau-1)}(v)\leq\max_{v\in V}I_{\G}^{(\tau)}(v)=M$.
Thus $\max_{v}f(v)\leq1$ and, by Theorem \ref{thm:submodular-mgf},
we have
\begin{align*}
\E_{S\sim p}\left[\exp\left(A^{(\tau-1)}\cdot I_{\G\setminus T}^{(\tau-1)}(S)\right)\right] & =\E_{S\sim p}\left[\exp\left(MA^{(\tau-1)}\cdot\frac{1}{M}I_{\G\setminus T}^{(\tau-1)}(S)\right)\right]\\
 & \leq\exp\left(\left(\exp\left(MA^{(\tau-1)}\right)-1\right)\cdot\E_{S\sim p}\left[\frac{1}{M}I_{\G\setminus T}^{(\tau-1)}(S)\right]\right)\\
 & =\exp\left(\underbrace{\frac{1}{M}\left(\exp\left(MA^{(\tau-1)}\right)-1\right)}_{=A^{(\tau)}}\left(I_{\G}^{(\tau)}(T)-\left|T\right|\right)\right)\\
 & =\exp\left(-A^{(\tau)}\left|T\right|\right)\cdot\exp\left(A^{(\tau)}I_{\G}^{(\tau)}(T)\right)
\end{align*}
Putting everything together,
\begin{align*}
\E_{G\sim\G}\left[\exp\left(\lambda\left|\Reach_{G}^{(\tau)}(T)\right|\right)\right] & =\exp\left(\lambda\left|T\right|\right)\E_{S\sim p}\left[\E_{G'\sim\G\setminus T}\left(\exp\left(\lambda\left|\Reach_{G'}^{(\tau-1)}\left(S\right)\right|\right)\right)\right]\\
 & \leq\exp\left(\lambda\left|T\right|\right)\exp\left(-B^{(\tau-1)}\left|T\right|\right)\E_{S\sim p}\left[\exp\left(A^{(\tau-1)}\cdot I_{\G\setminus T}^{(\tau-1)}(S)\right)\right]\\
 & \leq\exp\left(\lambda\left|T\right|\right)\exp\left(-B^{(\tau-1)}\left|T\right|\right)\exp\left(-A^{(\tau)}\left|T\right|\right)\exp\left(A^{(\tau)}I_{\G}^{(\tau)}(T)\right)\\
 & =\exp\Big(-\underbrace{\left(A^{(\tau)}+B^{(\tau-1)}-\lambda\right)}_{=B^{(\tau)}}\left|T\right|\Big)\exp\left(A^{(\tau)}I_{\G}^{(\tau)}(T)\right)
\end{align*}
\end{proof}

As a corollary, we obtain the following weaker bound which suffices
for our purposes:
\[
\E_{G\sim\G}\left[\exp\left(\lambda\cdot\left|\Reach_{G}^{(\tau)}(T)\right|\right)\right]\leq\exp\left(A^{(\tau)}\cdot I_{\G}^{(\tau)}(T)\right)
\]
The recursive term $A^{(\tau)}$ grows super-exponentially in $\lambda$
but surprisingly, we can control the term by choosing $\lambda$ appropriately.
Next, we show how to set $\lambda$ so that we can bound $A^{(\tau)}$. 
\begin{lem}
\label{lem:mgf-param} If $0\le\lambda\leq\frac{1}{M\left(\tau+1\right)}$
then $A^{(t)}\leq\frac{1}{\frac{1}{\lambda}-Mt}$ for all $0\leq t\leq\tau$.
Thus, for any $\eta\geq0$, setting $\lambda\le\min\left\{ \frac{\eta}{1+\eta}\cdot\frac{1}{M\tau},\frac{1}{M\left(\tau+1\right)}\right\} $
ensures $A^{(\tau)}\leq\left(1+\eta\right)\lambda$. Similarly, if
$-\frac{1}{2M}\le\lambda\le0$ then $A^{(t)}\le\frac{1}{\frac{1}{\lambda}-2Mt}$
for all $0\le t\le\tau$. Thus, for any $\eta\ge0$, setting $\lambda\ge\max\left\{ -\frac{1}{2M},\frac{-\eta}{1-\eta}\cdot\frac{1}{2\tau M}\right\} $
ensures $A^{(\tau)}\le\left(1-\eta\right)\lambda$.
\end{lem}
\begin{proof}
We ensure that $MA^{(t)}\leq1$ for all $t\leq\tau$ so that we can
use the inequality $e^{x}-1\leq x\left(1+x\right)$ which holds for
all $x\le1$. Applying this inequality, we obtain the recursive inequality
\[
A^{(t)}=\frac{1}{M}\left(\exp\left(MA^{(\tau-1)}\right)-1\right)\leq A^{(t-1)}\left(1+MA^{(t-1)}\right)
\]
with initial condition $A^{(0)}=\lambda$.

Let $x^{(t)}=MA^{(t)}$. Note that the above inequality implies $x^{(t)}-x^{(t-1)}\leq\left(x^{(t-1)}\right)^{2}$.
The solution to the differential equation $\frac{dx}{dt}=x^{2}$ is
$x(t)=\frac{1}{c-t}$.

We consider the cases $\lambda\ge0$ and $\lambda\le0$ separately.
Let us first assume that $\lambda\ge0$. To go from the continuous
to the discrete setting, note that
\[
\frac{1}{c-\left(t-1\right)}\left(1+\frac{1}{c-\left(t-1\right)}\right)\le\frac{1}{c-\left(t-1\right)}+\frac{1}{\left(c-\left(t-1\right)\right)\left(c-t\right)}=\frac{1}{c-t}.
\]
Using this, we can show that $MA^{(t)}\le\frac{1}{c-t}$ with $c=\frac{1}{MA^{(0)}}=\frac{1}{M\lambda}$.
The base case $t=0$ is clear. Consider $t\geq1$. Using the inductive
hypothesis and the above inequality, we obtain
\begin{align*}
MA^{(t)} & \leq MA^{(t-1)}\left(1+MA^{(t-1)}\right)\\
 & \leq\frac{1}{c-\left(t-1\right)}\left(1+\frac{1}{c-\left(t-1\right)}\right) & \left(\text{inductive hypothesis}\right)\\
 & \leq\frac{1}{c-t}.
\end{align*}
We need to ensure $MA^{(t)}\leq1$ for all $t\leq\tau$, so it suffices
to ensure
\[
\frac{1}{\frac{1}{M\lambda}-\tau}\leq1\Leftrightarrow\lambda\leq\frac{1}{M\left(\tau+1\right)}.
\]
For the second part, consider any $\eta\geq0$. To ensure $A^{(\tau)}\leq\left(1+\eta\right)\lambda$,
it suffices to ensure
\[
\frac{1}{M}\frac{1}{\frac{1}{M\lambda}-\tau}\leq\left(1+\eta\right)\lambda\Leftrightarrow\lambda\leq\frac{\eta}{1+\eta}\cdot\frac{1}{M\tau}.
\]

Let us now consider the case when $\lambda\le0$. For this case, we
will use that for any $c\le-2$ and $t\ge1$,
\begin{align*}
\frac{1}{c-2\left(t-1\right)}\left(1+\frac{1}{c-2\left(t-1\right)}\right) & =\frac{\left(c-2t\right)\left(c-2\left(t-1\right)+1\right)}{\left(c-2t\right)\left(c-2\left(t-1\right)\right)^{2}}\\
 & =\frac{\left(c-2\left(t-1\right)\right)^{2}-c+2t-4}{\left(c-2t\right)\left(c-2\left(t-1\right)\right)^{2}}\\
 & \le\frac{1}{c-2t}.
\end{align*}
We thus require that $\lambda\ge-\frac{1}{2M}$ such that $c=\frac{1}{M\lambda}\le-2$.
We can show that $MA^{(t)}\le\frac{1}{c-2t}$ with $c=\frac{1}{MA^{(0)}}=\frac{1}{M\lambda}$
and that $MA^{(t)}\ge-\frac{1}{2}$ for all $t\le\tau$. The base
case follows by definition. Consider $t\ge1$. Using the inductive
hypothesis and the above inequality, we obtain
\begin{align*}
MA^{(t)} & \leq MA^{(t-1)}\left(1+MA^{(t-1)}\right)\\
 & \leq\frac{1}{c-2\left(t-1\right)}\left(1+\frac{1}{c-2\left(t-1\right)}\right) & \left(\text{inductive hypothesis}\right)\\
 & \leq\frac{1}{c-2t}.
\end{align*}
To ensure that $A^{(\tau)}\le\left(1-\eta\right)\lambda$, it thus
suffices to choose
\[
\frac{1}{M}\frac{1}{\frac{1}{M\lambda}-2\tau}\le\left(1-\eta\right)\lambda\Leftrightarrow\lambda\ge\frac{-\eta}{1-\eta}\cdot\frac{1}{2\tau M}.
\]

\end{proof}

Via standard arguments, we use this to derive a tail bound on the
estimation
\[\widehat{I}^{(\tau)}\left(S\right)=\frac{1}{\ell}\sum_{i=1}^{\ell}\left|\Reach_{G_{i}}^{(\tau)}\left(S\right)\right|.\] 
\begin{lem}
(Theorem \ref{thm:ic-sample-complexity}) Let $\mathcal{F}=\left\{ S\subseteq V\colon\left|S\right|\leq k\right\} $
be the set of all feasible solutions. Let $\widehat{I^{(\tau)}}\colon2^{V}\to\Reals$,
$\widehat{I^{(\tau)}}\left(S\right)=\frac{1}{\ell}\sum_{i=1}^{\ell}\left|\Reach_{G_{i}}^{(\tau)}\left(S\right)\right|$
where $G_{1},\dots,G_{\ell}$ are live-edge graphs sampled independently
from the independent cascade model. For any values $\epsilon,\delta\in\left[0,1\right]$,
we can ensure that 
\[
\Pr\left[\forall S\in\mathcal{F}\colon\left|\widehat{I^{(\tau)}}\left(S\right)-I^{(\tau)}\left(S\right)\right|\leq\epsilon\opt\right]\geq1-\delta
\]
 using $\ell=O\left(\frac{M\tau\log\left(\left|\mathcal{F}\right|/\delta\right)}{\opt\epsilon^{2}}\right)$
samples, where $M=\max_{v\in V}I^{(\tau)}\left(v\right)$ is the maximum
singleton influence and $\opt=\max_{S\colon\left|S\right|\leq k}I^{(\tau)}\left(S\right)$
is the value of the optimal solution. Since $M\leq\opt$ and $\left|\mathcal{F}\right|=O\left(n^{k}\right)$,
we have $\ell\leq O\left(\frac{\tau k\log\left(n/\delta\right)}{\epsilon^{2}}\right)$.
\end{lem}
\begin{proof}
Let $c=\epsilon\opt$. We first use a union bound which allows us
to consider one failure event per tail, for a single set $T\in{\cal F}$:

\[
\Pr\left[\left|\widehat{I^{(\tau)}}\left(T\right)-I^{(\tau)}\left(T\right)\right|\ge c\right]=\Pr\left[\widehat{I^{(\tau)}}\left(T\right)\ge I_{\G}^{(\tau)}\left(T\right)+c\right]+\Pr\left[\widehat{I^{(\tau)}}\left(T\right)\le I^{(\tau)}\left(T\right)-c\right]
\]

Let us first consider the upper tail. By Markov's inequality, independence
between the samples $G_{i}\sim\G$, and Lemma~\ref{lem:mgf-bound},

\begin{align*}
\Pr\left[\widehat{I^{(\tau)}}\left(T\right)\ge I^{(\tau)}\left(T\right)+\epsest\right] & =\Pr\left[\exp\left(\lambda\ell\widehat{I^{(\tau)}}\left(T\right)\right)\geq\exp\left(\lambda\ell\left(\E\left[I^{(\tau)}\left(T\right)\right]+\epsest\right)\right)\right]\\
 & \leq\frac{\E\left[\exp\left(\lambda\ell\widehat{I^{(\tau)}}\left(T\right)\right)\right]}{\exp\left(\lambda\ell\left(I^{(\tau)}\left(T\right)+\epsest\right)\right)}\\
 & =\frac{\left(\E_{G\sim\G}\left[\exp\left(\lambda\left|\Reach_{G}^{(\tau)}(T)\right|\right)\right]\right)^{\ell}}{\exp\left(\lambda\ell\left(I_{\G}^{(\tau)}\left(T\right)+c\right)\right)}\\
 & \leq\exp\left(A^{(\tau)}\ell I^{(\tau)}\left(T\right)-\lambda\ell\left(I^{(\tau)}\left(T\right)+c\right)\right).
\end{align*}
By setting $\lambda=O\left(\frac{c}{\tau MI^{(\tau)}\left(T\right)}\right)$,
we ensure $A^{(\tau)}\leq\lambda\left(1+\frac{c}{2I^{(\tau)}\left(T\right)}\right)$
due to Lemma~\ref{lem:mgf-param}. We obtain
\begin{align*}
\exp\left(A^{(\tau)}\ell I^{(\tau)}\left(T\right)-\lambda\ell\left(I^{(\tau)}\left(T\right)+c\right)\right) & \le\exp\left(\lambda\frac{c}{2I^{(\tau)}\left(T\right)}\ell I^{(\tau)}\left(T\right)-\lambda\ell c\right)\\
 & =\exp\left(-\frac{c}{2}\lambda\ell\right).
\end{align*}
It therefore suffices to set 
\[
\ell=\Omega\left(\frac{\log\left(1/\delta'\right)}{c\lambda}\right)=\Omega\left(\tau MI^{(\tau)}\left(T\right)\frac{\log\left(1/\delta'\right)}{c^{2}}\right)=\Omega\left(\tau\frac{\log\left(1/\delta'\right)}{\epsilon^{2}}\right)
\]
and observe that $\opt\ge M$ to obtain a failure probability of $\delta'/2$
for any $\delta'>0$.

The bound for the lower tail follows analogously. For $\lambda<0$,
we get
\begin{align*}
\Pr\left[\widehat{I^{(\tau)}}\le I_{\G}^{(\tau)}\left(T\right)-c\right] & =\Pr\left[\exp\left(\lambda\ell\widehat{I^{(\tau)}}\right)\ge\exp\left(\lambda\ell\left(\E\left[I^{(\tau)}\right]-c\right)\right)\right]\\
 & \leq\frac{\E\left[\exp\left(\lambda\ell\widehat{I^{(\tau)}}\right)\right]}{\exp\left(\lambda\ell\left(I^{(\tau)}\left(T\right)-c\right)\right)}\\
 & =\frac{\left(\E_{G\sim\G}\left[\exp\left(\lambda\left|\Reach_{G}^{(\tau)}(T)\right|\right)\right]\right)^{\ell}}{\exp\left(\lambda\ell\left(I^{(\tau)}\left(T\right)-c\right)\right)}\\
 & \leq\exp\left(A^{(\tau)}\ell I^{(\tau)}\left(T\right)-\lambda\ell\left(I^{(\tau)}\left(T\right)-c\right)\right)
\end{align*}
By setting $\lambda=\Omega\left(-\frac{c}{\tau MI^{(\tau)}\left(T\right)}\right)$
we ensure $A^{(\tau)}\le\lambda\left(1-\frac{1}{2}\frac{c}{I^{(\tau)}\left(T\right)}\right)$
due to Lemma~\ref{lem:mgf-param}. We obtain
\begin{align*}
\exp\left(A^{(\tau)}\ell I^{(\tau)}\left(T\right)-\lambda\ell\left(I^{(\tau)}\left(T\right)-c\right)\right) & \le\exp\left(-\frac{1}{2}\lambda\frac{c}{I^{(\tau)}\left(T\right)}\ell I^{(\tau)}\left(T\right)+\lambda\ell c\right)\\
 & =\exp\left(\frac{c}{2}\lambda\ell\right).
\end{align*}
The above choice of $\ell$ also suffices to obtain a failure probability
of $\delta'/2$ for the lower tail. Overall, we therefore get a failure
probability of $\delta'$ for obtaining an $c$-additive approximation
on $T$. To get a $c$-additive approximation for all sets $T$ with
$\left|T\right|\le k$ with high probability, we choose $\delta'=\delta{n \choose k}^{-1}$.
By a union bound, we then obtain the theorem statement.
\end{proof}

\section{\label{sec:analysis-live-edge} Algorithm for influence maximization
in independent cascade models}

In this section, we provide the pseudocode and analysis (Theorem \ref{thm:ic-main})
for our algorithm on the independent cascade model. 

\begin{algorithm}
\begin{raggedright}
\caption{\label{alg:independent-cascades} Our algorithm for independent cascade
(IC) models.}
\par\end{raggedright}
\begin{raggedright}
$\mathrm{IndependentCascade}\left(\epsilon,\delta\right):$
\par\end{raggedright}
\begin{raggedright}
\textbf{\textcolor{blue}{Input:}}\textcolor{blue}{{} Error parameters
$\epsilon$ and $\delta$}
\par\end{raggedright}
\begin{raggedright}
\textbf{\textcolor{blue}{Output:}}\textcolor{blue}{{} With probability
$1-\delta$, we obtain a solution $S\subseteq V$ with $\left|S\right|\le k$
such that $I^{(\tau)}\left(S\right)\ge\left(1-\frac{1}{e}-\epsilon\right)\opt$.}
\par\end{raggedright}
\begin{raggedright}
Sample $\ell=\frac{9\tau k}{\epsilon^{2}}\log\left(\frac{2}{\delta}\right)$
graphs $G_{1},G_{2},\dots,G_{\ell}$ independently from the independent
cascade model $\G$
\par\end{raggedright}
\begin{raggedright}
$S\gets\mathrm{ObservedCascades}\left(\frac{\epsilon}{3},\frac{\delta}{2}\right)$
where the estimation is for graphs $G_{1},G_{2},\dots,G_{\ell}$
\par\end{raggedright}
\begin{raggedright}
\textbf{return} $S$
\par\end{raggedright}
\raggedright{}
\end{algorithm}

\begin{lem}
\label{lem:ic-sampling} There is an algorithm that generates $\ell$
samples from an IC model in time $\widetilde{O}\left(m+\ell\overline{m}\right)$
where $m=\left|E\right|$ is the number of edges in the graph and
$\overline{m}=\sum_{e\in E}p_{e}$ is the expected number of edges
in a sample.
\end{lem}
\begin{proof}
The algorithm for sampling uses following observation. Let $\mathrm{Geom}\left(p\right)$
be the Geometric distribution with success probability $p$. Consider
an edge $e$, and let $i_{1},i_{2},\dots,i_{a}$ be the indices of
the samples that contain $e$. Since $e$ becomes live independently
with probability $p_{e}$, each of the random variables $i_{1}$ and
$i_{j}-i_{j-1}$ for all $2\leq j\leq a$ are $\mathrm{Geom}\left(p_{e}\right)$
random variables. Thus we can construct the samples by considering
each edge $e$ in turn, and generating the indices of the subgraphs
that contain $e$ by repeatedly sampling values $X_{1},X_{2},\dots,X_{a}$
from the $\mathrm{Geom}\left(p_{e}\right)$ distribution until $X_{1}+X_{2}+\dots+X_{a}\geq\ell$.
The number of samples from the Geometric distribution is equal to
the total number of edges in the sampled subgraphs. We can construct
a sample from the $\mathrm{Geom}\left(p_{e}\right)$ distribution
in expected time $O\left(1+\log\left(1/p_{e}\right)/\log m\right)$
\citep{BringmannF13}, which is $O\left(1\right)$ for $p_{e}\geq1/m$;
for $p_{e}\leq1/m$, we can round down $p_{e}$ to $0$ and the influence
changes by at most $1/m$.
\end{proof}

\begin{thm}
(Theorem \ref{thm:ic-main}) For any $\epsilon,\delta$ given as input,
Algorithm \ref{alg:independent-cascades} achieves an approximation
of $1-\frac{1}{e}-\epsilon$ with probability $1-\delta$ using $\ell=O\left(\frac{\tau k}{\epsilon^{2}}\log\left(\frac{n}{\delta}\right)\right)$
samples. The expected running time of our algorithm is $\widetilde{O}\left(m+Nb\ell\overline{m}\log b\right)$
where $m=\left|E\right|$ is the number of edges in the graph, $\overline{m}=\sum_{e}p_{e}$
is the expected number of edges in the graph, $N=\widetilde{O}\left(\frac{1}{\epsilon^{2}}\log\left(\frac{1}{\delta}\right)\right)$
and $b=\widetilde{O}\left(\frac{1}{\epsilon^{2}}\log\left(\frac{\ell N}{\delta}\right)\right)$.
\end{thm}
\begin{proof}
By Theorem \ref{thm:observed-cascades-main}, we obtain for the set
$S$ returned by $\mathrm{ObservedCascades}$ that
\[
\widehat{I^{(\tau)}}\left(S\right)\ge\left(1-\frac{1}{e}-\frac{\epsilon}{3}\right)\max_{S\subseteq V:\left|S\right|\le k}\widehat{I^{(\tau)}}\ge\left(1-\frac{1}{e}-\frac{\epsilon}{3}\right)\widehat{I^{(\tau)}}\left(S^{*}\right)
\]
with probability at least $1-\frac{1}{2}\delta$. By Lemma \ref{thm:ic-sample-complexity}
and our choice of $\ell$, we obtain that $\left|\widehat{I^{(\tau)}}(T)-I^{(\tau)}\left(T\right)\right|\le\frac{\epsilon}{3}\opt$
on all sets $T\subseteq$ with $\left|T\right|\le k$ with probability
at least $1-\frac{1}{2}\delta$. We condition that both events are
successful. Then,
\begin{align*}
I^{(\tau)}\left(S\right) & \ge\widehat{I^{(\tau)}}\left(S\right)-\frac{\epsilon}{3}\opt\\
 & \ge\left(1-\frac{1}{e}-\frac{\epsilon}{3}\right)\widehat{I^{(\tau)}}\left(S^{*}\right)-\frac{\epsilon}{3}\opt\\
 & \ge\left(1-\frac{1}{e}-\frac{\epsilon}{3}\right)\left(\opt-\frac{\epsilon}{3}\opt\right)-\frac{\epsilon}{3}\opt\\
 & \ge\left(1-\frac{1}{e}-\epsilon\right)\opt.
\end{align*}

By Lemma \ref{lem:ic-sampling}, the expected running time to create
the $\ell$ samples is $\tilde{O}\left(m+\ell\overline{m}\right)$.
By Theorem \ref{thm:observed-cascades-main}, the total running time
for the optimization is $\widetilde{O}\left(Nb\mtot\log b\right)$
for $N=\widetilde{O}\left(\frac{1}{\epsilon^{2}}\log\left(\frac{1}{\delta}\right)\right)$
and $b=\widetilde{O}\left(\frac{1}{\epsilon^{2}}\log\left(\frac{\ell N}{\delta}\right)\right)$.
In expectation, we therefore obtain a total running time of $\widetilde{O}\left(m+Nb\ell\overline{m}\log b\right)$
for $\ell=O\left(\frac{\tau k}{\epsilon^{2}}\log\left(\frac{n}{\delta}\right)\right)$,
$N=\widetilde{O}\left(\frac{1}{\epsilon^{2}}\log\left(\frac{1}{\delta}\right)\right)$,
and $b=\widetilde{O}\left(\frac{1}{\epsilon^{2}}\log\left(\frac{\ell N}{\delta}\right)\right)$.
\end{proof}

\section{\label{sec:appendix-experiments} Further experimental details and
results}

\subsection{\label{sec:appendix-experimental-details} Further implementation
details}

We now detail implementation details for our algorithms and the baselines.
In all of our algorithms, instead of performing a binary search to
obtain a guess $\optguess$, we estimate the maximum singleton value
$M$ and run the algorithm with $\optguess=kM$. We use a value of
$\epsilon=0.1$ across instantiations of the template Algorithm \ref{alg:template-idealized}.
We further use lazy evaluations as described in \citet{chen21}. For
general models, we use a constant value $\cest=0.01$ in Algorithm
\ref{alg:general-model-empirical-variance}. We implement the Greedy
approach of \citet{sadeh20} using a median-of-means estimator, which
we also describe in Algorithm \ref{alg:estimate-gains-variance} using
$n_{\mathrm{p}}=k\log\left(2n\right)$ and $n_{\mathrm{s}}=\frac{k^{2}\tau}{\epsilon^{2}}$.
For observed cascades, we implement the bottom-$b$ in Algorithms
\ref{alg:estimate-gains-sketch-restricted} and \ref{alg:estimate-gains-sketch-unrestricted}
sketch using a fixed value of $b=64$ as suggested in \citet{cohen14}. 

\subsection{\label{sec:appendix-omitted-results} Further experimental results}

We provide experimental results which we omitted from the main body. 

\paragraph{General threshold models}

\begin{figure*}
\centering{}\includegraphics[width=0.4\linewidth]{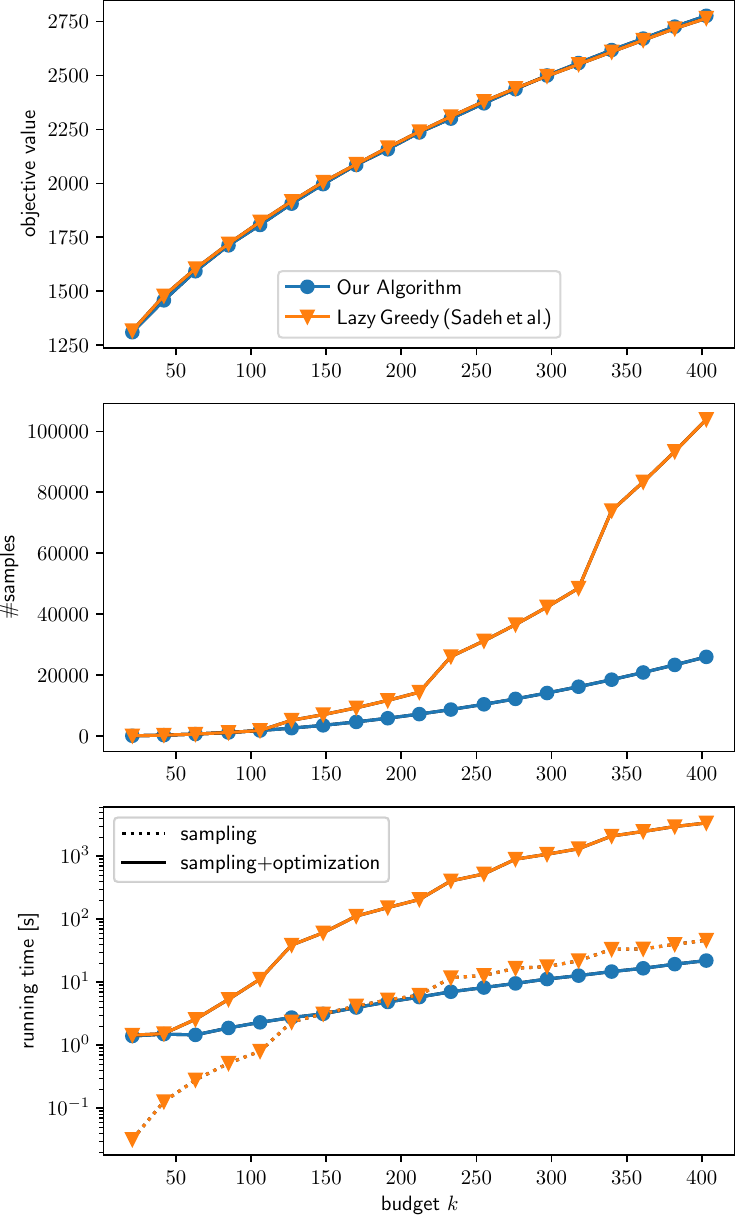}\hspace*{1cm}\includegraphics[width=0.4\linewidth]{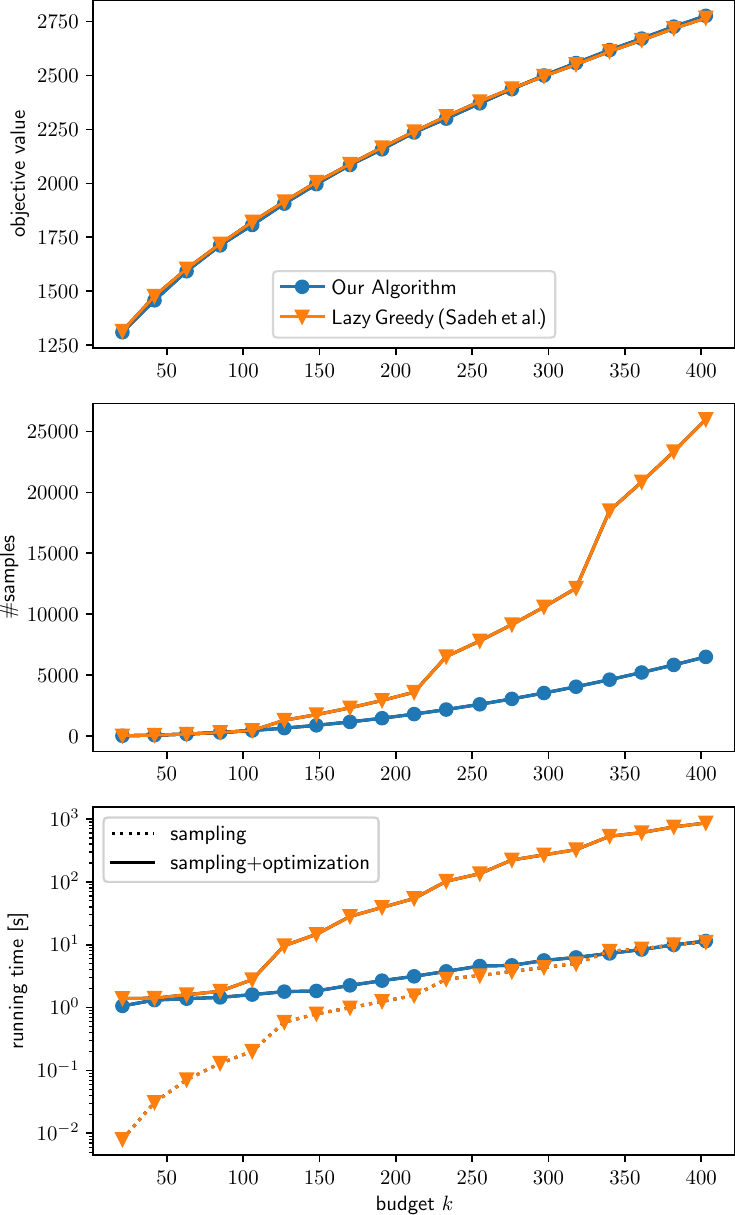}\caption{\label{fig:gtm-facebook-large} Influence maximization for a general
threshold model on the Facebook network for $\tau=n$ (left) and $\tau=10$
(right).}
\end{figure*}

Figure \ref{fig:gtm-facebook-large} shows results for general threshold
models for $\tau=10$ and $\tau=n$. The results are comparable to
$\tau=5$ in Figure \ref{fig:lt-facebook-large-2} in the main body,
and show that our adaptive method greatly reduces the number of samples
and the running time, while only suffering a slight decrease in objective
value. 

\paragraph{Observed cascades}

\begin{figure*}
\centering{}\includegraphics[width=0.27\linewidth]{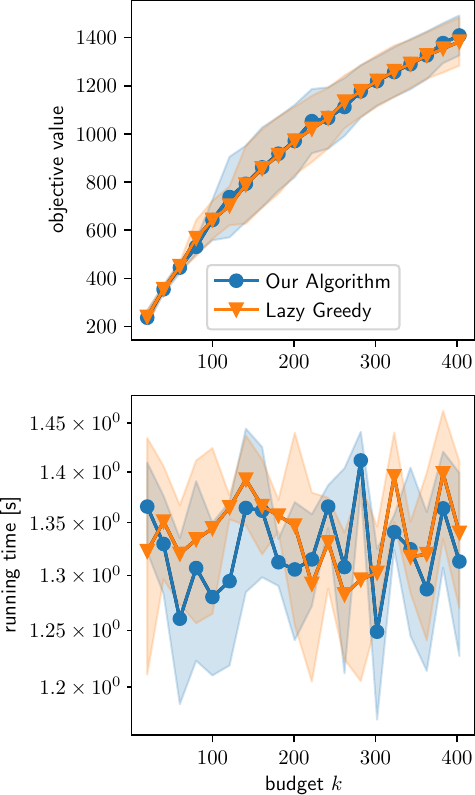}\hspace*{0.5cm}\includegraphics[width=0.27\linewidth]{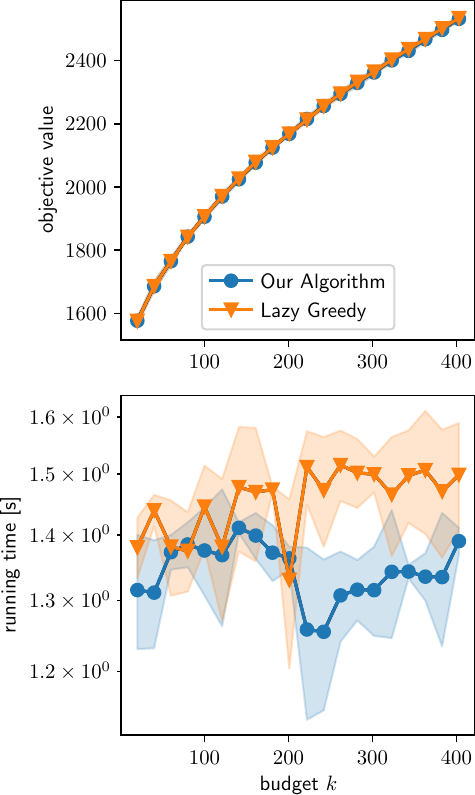}\hspace*{0.5cm}\includegraphics[width=0.27\linewidth]{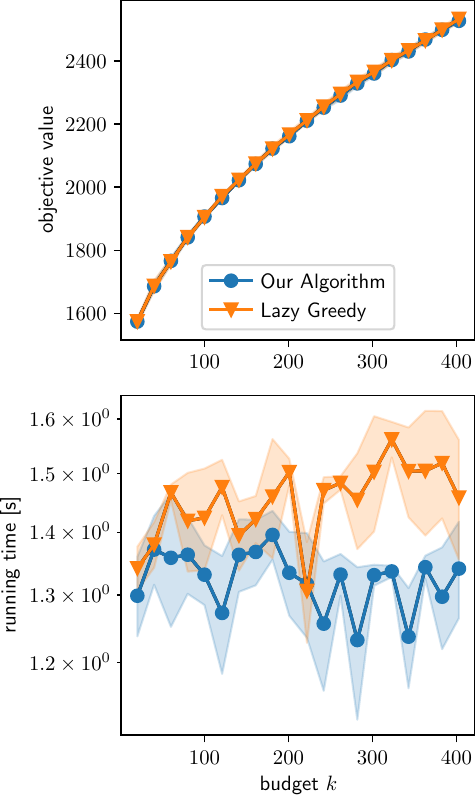}\caption{\label{fig:oc-facebook} Influence maximization in the observed cascades
model for $100$ observed cascades on the Facebook network for $\tau=n$
(left) and $\tau=10$ (right).}
\end{figure*}

\begin{figure}
\centering{}\includegraphics[width=0.4\linewidth]{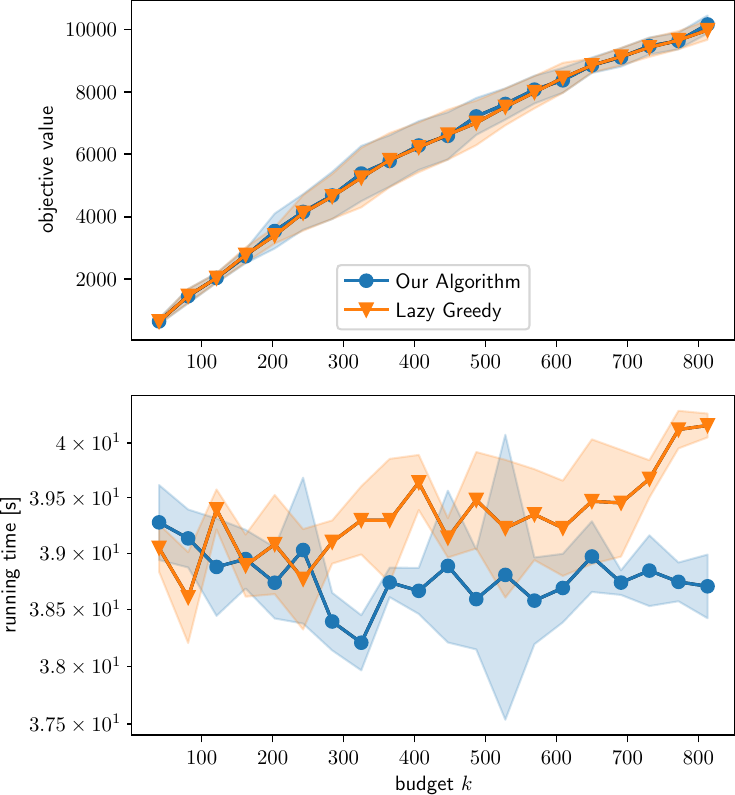}\hspace*{1cm}\includegraphics[width=0.4\linewidth]{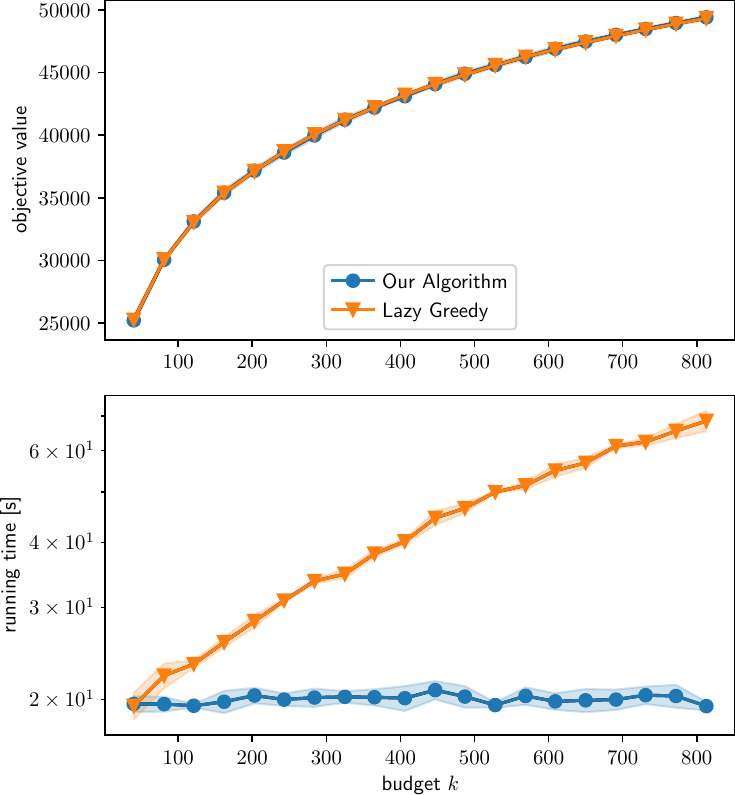}\caption{\label{fig:oc-twitter-2} Influence maximization for $100$ observed
cascades with $\tau\in\left\{ n,5,10\right\} $ on the Twitter network.}
\end{figure}

\begin{figure}
\centering{}\includegraphics[width=0.27\linewidth]{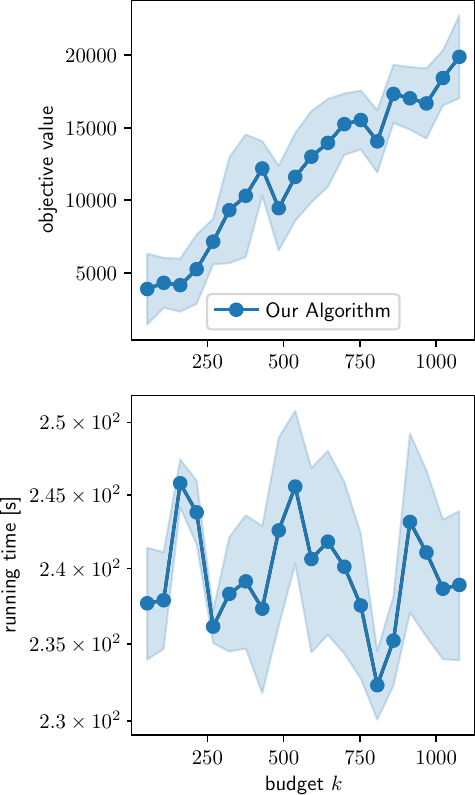}\hspace*{0.5cm}\includegraphics[width=0.27\linewidth]{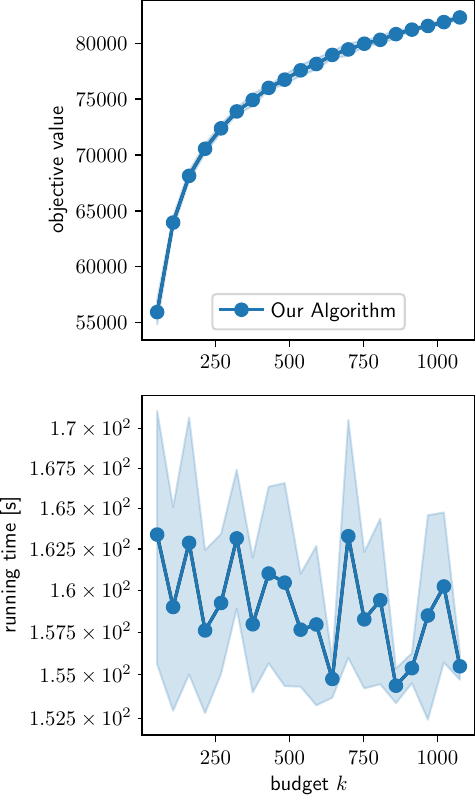}\hspace*{0.5cm}\includegraphics[width=0.27\linewidth]{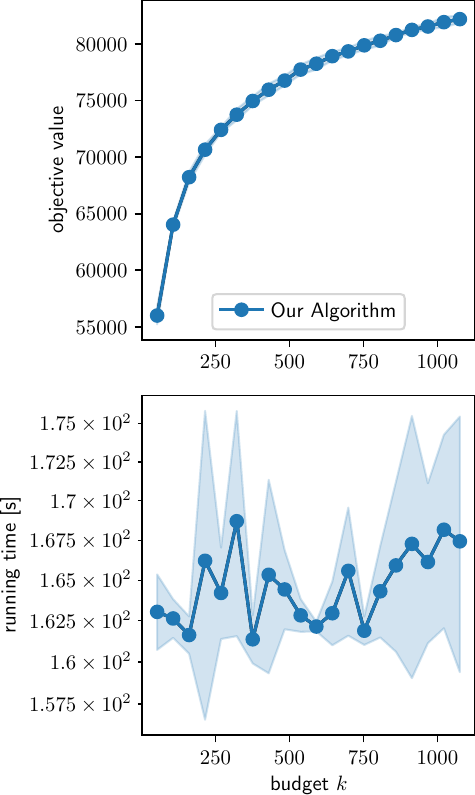}\caption{\label{fig:oc-gplus} Influence maximization for $100$ observed cascades
with $\tau\in\left\{ n,5,10\right\} $ on the Google-Plus network.}
\end{figure}

Figures \ref{fig:oc-facebook}, \ref{fig:oc-twitter-2}, and \ref{fig:oc-gplus}
shows further results for the observed cascade model. The results
on the Facebook and Twitter results confirm the findings of Figure
\ref{fig:lt-facebook-large-1-2} where we obtain an almost identical
objective value but require much less running time compared to the
lazy Greedy algorithm. At the same time, our experiments on the Google-Plus
and Pokec (cf. Figure \ref{fig:lt-facebook-large-1-1-1}) networks
show that our methods can scale to large networks, as opposed to the
lazy Greedy algorithm which takes hours even for small budgets. 

\paragraph{Independent cascade models}

\begin{figure*}
\centering{}\includegraphics[width=0.4\linewidth]{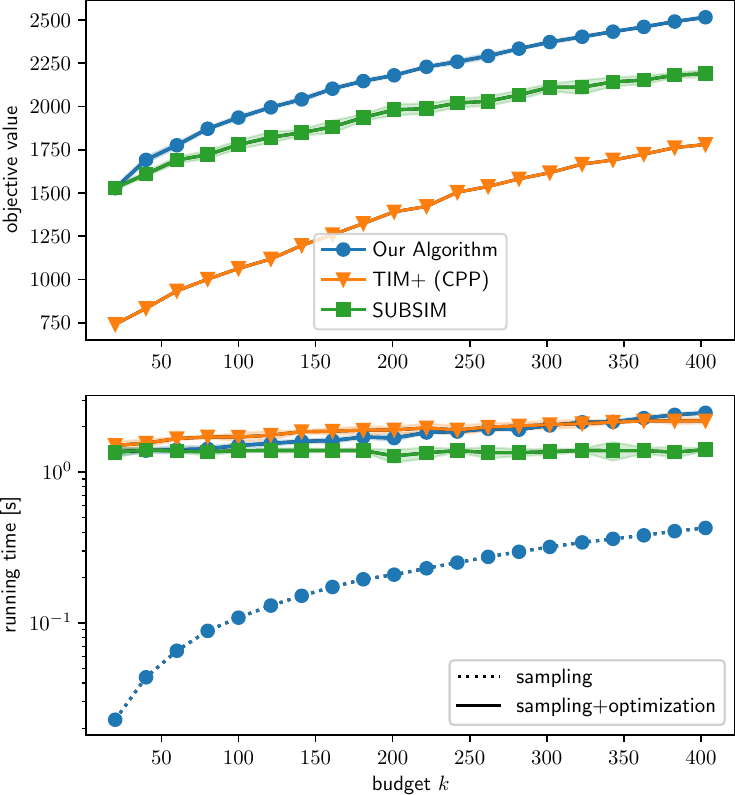}\hspace*{1cm}\includegraphics[width=0.4\linewidth]{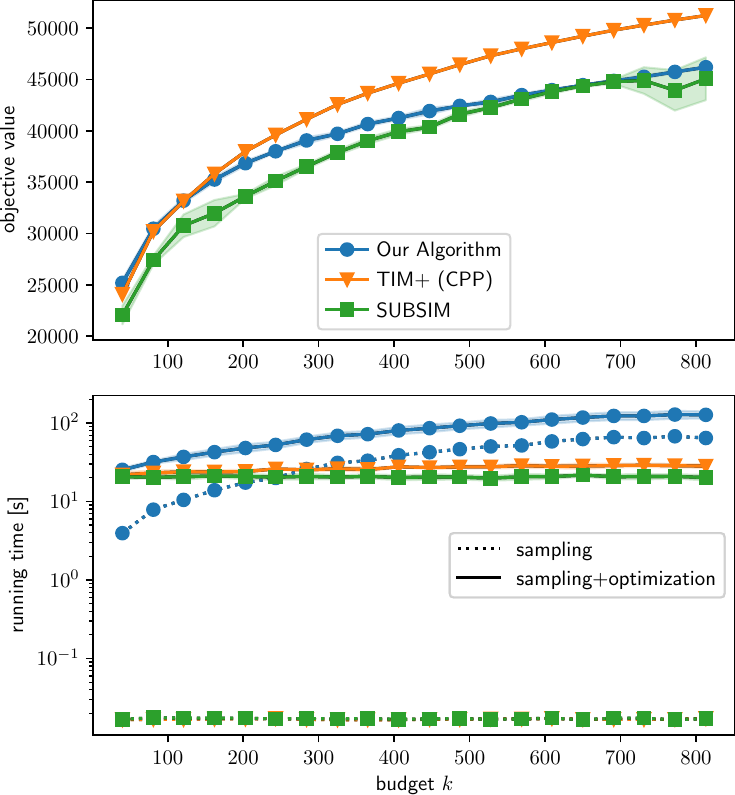}\caption{\label{fig:ic-facebook-twitter} Influence maximization in the IC
model on the Facebook and Twitter networks for $\tau=n$.}
\end{figure*}

We create an instance for the independent cascade model from each
network by assigning probabilities $p_{(u,v)}=\frac{1}{\deg(v)}$
to each directed edge $(u,v)\in E$, which reflects that a node $v$
has only limited capacity of being influenced from each of its neighbors.
We use Algorithm \ref{alg:independent-cascades} with the sketching
according to Algorithm \ref{alg:estimate-gains-sketch-unrestricted}
for $\tau=n$. We compare this with state-of-the-art algorithms for
the unrestricted influence ($\tau=n$) in the independent cascade
model, TIM+ \citep{tang14} and SUBSIM \citep{guo22}. Our results
are in Figure \ref{fig:ic-facebook-twitter} for the Facebook and
Twitter networks.

\textbf{Discussion:} We observe that the objective value of our algorithm
is competitive with state-of-the-art algorithms. However, our algorithm
is not able to compete against specialized algorithms for the IC model
in terms of running time. As shown in the comparison of Section \ref{sec:ic-model-alg},
our algorithm can outperform the prior work in its asymptotic running
time in certain parameter regimes. As such, we view our algorithm
as a theoretical contribution and leave practical improvements as
an open direction for future work.

\subsection{\label{sec:appendix-ablation} Ablation study }

\begin{figure*}
\centering{}\includegraphics[width=0.4\linewidth]{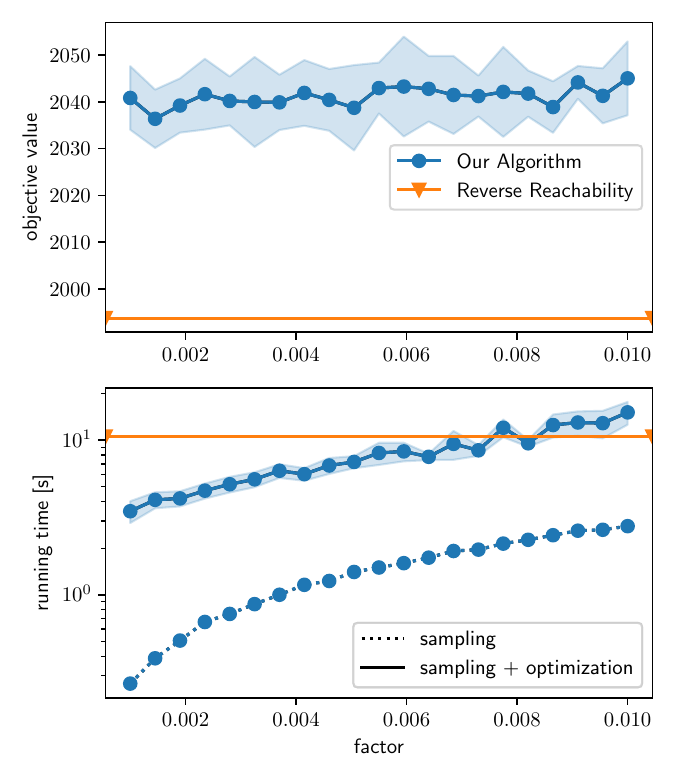}\includegraphics[width=0.4\linewidth]{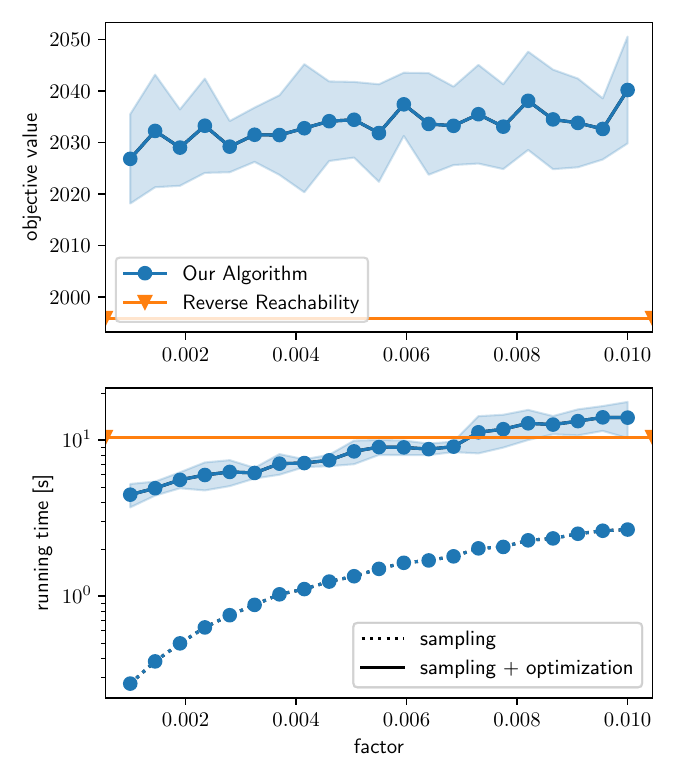}\caption{\label{fig:sample-size} Influence maximization in the independent
cascade model on the Facebook network for a budget of $k=400$ seed
nodes and $\tau=5$ (left) and $\tau=10$ (right). We vary the factor
that reduces the number of samples.}
\end{figure*}
\begin{figure*}
\centering{}\includegraphics[width=0.27\linewidth]{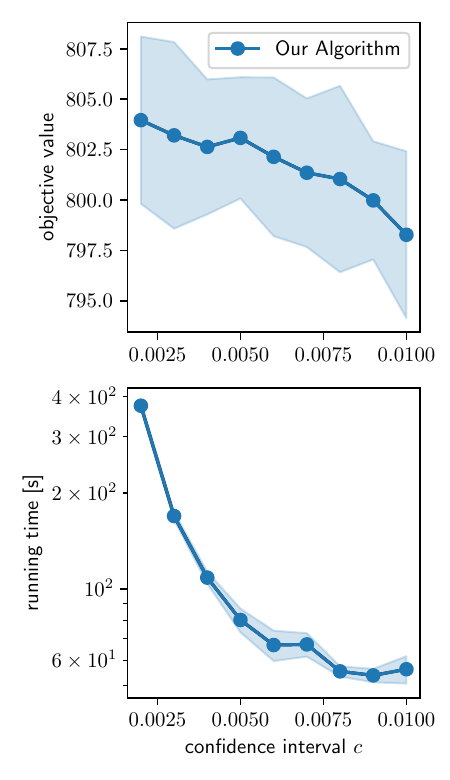}\includegraphics[width=0.27\linewidth]{plots/plot_algo_budgets-facebook_combined.txt-tau5-01-27-2025--13-46-54-044820}\includegraphics[width=0.27\linewidth]{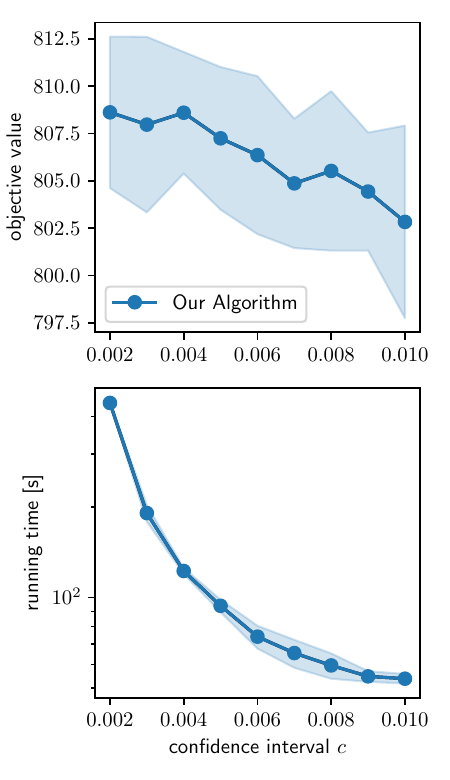}\caption{\label{fig:lt-facebook-large-3} Influence maximization for a general
threshold model on the Facebook network with a budget of $k=200$
and $\tau=n$ (left), $\tau=5$ (center), and $\tau=10$ (right).
We vary the additive estimation error $\protect\cest$.}
\end{figure*}

The number of adaptive samples depends on the the parameter $\cest>0$
of Algorithm \ref{alg:estimate-gains-empirical-variance} which is
the allowed additive estimation error for the estimated marginal gain.
In Algorithm \ref{alg:general-model-empirical-variance} we use a
theoretical choice of $\cest=\frac{\epsilon}{12k}\optguess\le\frac{\epsilon}{24k}\opt$.
Figure \ref{fig:lt-facebook-large-3} shows how different choices
of $\cest$ influence the performance of Algorithm \ref{alg:estimate-gains-empirical-variance}.
We observe that sufficiently small values of $\cest$ only slightly
decrease the objective value but allow us to drastically reduce the
running time.

In Figure \ref{fig:sample-size} we further investigate the effect
of reducing the number of samples that are required in Algorithm \ref{alg:independent-cascades}.
In particular, we reduce the number of samples by a constant factor,
and we range this factor from $1/1000$ to $1/100$. We see that even
using $1/1000$-th of the theoretically required samples does not
result in a substantial drop in objective value, but can greatly reduce
the running time of our algorithm.